\documentclass[12pt]{article}      
\usepackage[margin=1in]{geometry} 

\usepackage[english]{babel}
\usepackage[utf8]{inputenc}
\usepackage{amssymb,amsmath,amsthm,amsfonts}
\usepackage{algpseudocode}
\usepackage{xspace,xpunctuate}
\usepackage{hyperref}
\usepackage{cleveref}
\usepackage{multicol}
\usepackage{nicefrac}
\usepackage{mathtools}
\usepackage{comment}
\usepackage{graphicx}
\usepackage{paracol}
\usepackage{enumitem,boxedminipage}
\usepackage{xparse}

\usepackage{tikz}
\usetikzlibrary{backgrounds,positioning,arrows}

\DeclareMathOperator{\newbowtie}{\bowtie}
\DeclareMathOperator{\notnewbowtie}{\not\bowtie}
\makeatletter
\g@addto@macro\bfseries{\boldmath}
\makeatother

\usepackage{accents}
\usepackage{yhmath}

\usepackage{thm-restate,xcolor}
\newcommand{\Nat}{\mathbb{N}}
\newcommand{\relE}{E}
\newcommand{\card}{\operatorname{card}}

\definecolor{darkgreen}{rgb}{0.0, 0.2, 0.13}
\definecolor{darkcyan}{rgb}{0.0, 0.55, 0.55}

\let\ooplus\oplus 
\renewcommand{\oplus}{\cup}

\newcommand{\W}{{\textnormal{W}}}
\newcommand{\NP}{{\textnormal{NP}}}
\newcommand{\FPT}{\textnormal{FPT}\xspace} 
\newcommand{\XP}{\textnormal{XP}\xspace} 

\usepackage{soul}
\usepackage{xcolor}
\definecolor{mycolor}{RGB}{230,230,230}
\sethlcolor{mycolor}

\newtheorem{lemma}{Lemma}
\newtheorem{theorem}{Theorem}
\newtheorem*{theorem*}{Theorem}
\newtheorem{corollary}{Corollary}

\theoremstyle{definition}
\newtheorem{definition}{Definition}
\newtheorem{observation}{Observation}

\newcommand{\N}{\ensuremath{{\mathbb{N}}}\xspace} 
\newcommand{\R}{\ensuremath{{\mathbb{R}}}\xspace} 
\newcommand{\Q}{\ensuremath{{\mathbb{Q}}}\xspace}
\newcommand{\bigoh}{\ensuremath{{\mathcal{O}}}\xspace}

\newcommand{\PP}{\mathcal{P}}
\newcommand{\MM}{M}

\newcommand{\undersatformulatext}[1][\(\alpha\)]{#1-loosening}
\newcommand{\oversatformulatext}[1][\(\alpha\)]{#1-tightening}
\NewDocumentCommand \undersatformula { o o } {%
	\IfNoValueTF{#1}
	{%
		\IfNoValueTF {#2}
		{%
			{\phi}_{\alpha}%
		}
		{%
			{\phi}_{#1}%
		}%
	}
	{%
	{#2}_{#1}%
	}%
}
\NewDocumentCommand \oversatformula { o o } {%
\IfNoValueTF{#1}
{%
	\IfNoValueTF {#2}
	{%
		{\phi}^{\alpha}%
	}
	{%
		{\phi}^{#1}%
	}%
}
{%
	{#2}^{#1}%
}%
}

\usepackage{afterpage}

\makeatletter
\newcommand*{\changepagecolor}{%
	\@ifnextchar[\@changepagecolor@i\@changepagecolor@ii
}
\def\@changepagecolor@i[#1]#2{%
	\@changepagecolor@do{[{#1}]{#2}}%
}
\newcommand*{\@changepagecolor@ii}[1]{%
	\@changepagecolor@do{{#1}}%
}
\newcommand*{\@changepagecolor@do}[1]{%
	\begingroup
	\offinterlineskip
	\hbox to 0pt{%
		\kern-\paperwidth
		\vtop to 0pt{%
			\color#1%
			\hrule width 2\paperwidth height \paperheight
			\vss
		}%
		\hss
	}%
	\endgroup
	\afterpage{\pagecolor#1}%
}
\makeatother

{\changepagecolor{green!30}}%
{\changepagecolor{white}}

\newcommand{\gran}{\ensuremath \gamma \xspace}

\usepackage{bm}
\newcommand{\lle}{\text{\scalebox{0.55}{$\bm\le$}}}
\newcommand{\gge}{\text{\scalebox{0.55}{$\bm\ge$}}}

\newcommand{\FV}{{\rm FV }}
\newcommand{\negFV}{\overline{\rm FV }}
\newcommand{\FVplus}{\textnormal{FV}^{\setlength\fboxsep{1pt}\scalebox{0.40}{\boxed{\tiny\forall}}}}
\newcommand{\roundedFVplus}{\tilde{\textnormal{FV}}^{\setlength\fboxsep{1pt}\scalebox{0.40}{\boxed{\tiny\forall}}}}

\newcommand{\x}{{\bar x}}

\renewcommand{\phi}{\varphi}

\newcommand{\MSO}{{\rm MSO}}
\newcommand{\CMSO}{{\rm CMSO}}

\newcommand{\eps}{\varepsilon}
\newcommand{\tw}{\operatorname{tw}}
\newcommand{\cw}{\operatorname{cw}}
\newcommand{\umodels}[1][\alpha]{\models_{#1}}
\newcommand{\omodels}[1][\alpha]{\models^{#1}}
\newcommand{\roundedN}[1][\alpha,N,\gran]{\tilde\N_{#1}}

\newcommand{\alphabot}{b}

\newcommand{\mso}{\textnormal{MSO}}
\newcommand{\LinEMSO}{\textnormal{LinECMSO}}
\newcommand{\MSOcomp}{\textnormal{CMSO}\ensuremath{[\lessgtr]}\xspace}
\newcommand{\MSOcompone}{\textnormal{CMSO}\ensuremath{_1[\lessgtr]}}

\newcommand{\smallsquare}{\raisebox{0.2mm}{\scalebox{0.57}{$\square$}}}
\newcommand{\verysmallsquare}{\raisebox{0.2mm}{\scalebox{0.37}{$\square$}}}
\newcommand{\blockMSO}{\(\setlength\fboxsep{1pt}\boxed{\tiny\forall}\)\textnormal{CMSO}}
\newcommand{\blockMSOone}{\(\setlength\fboxsep{1pt}\boxed{\tiny\forall}\)\textnormal{CMSO}$_1$}
\newcommand{\blockMSOtwo}{\(\setlength\fboxsep{1pt}\boxed{\tiny\forall}\)\textnormal{CMSO}$_2$}

\newcommand{\terms}{\tau}
\newcommand{\comp}{\textnormal{comp}}

\newcommand{\Block}{\textnormal{Block}}
\newcommand{\tabl}{\textnormal{table}}
\newcommand{\val}{\textnormal{max}}
\newcommand{\witness}{\textnormal{witness}}
\newcommand{\target}{\hat t}
\newcommand{\thres}{g}

\newcommand{\cmsoset}{\textnormal{CMSO}_1(\sigma,q,\bar X \bar Y)}
\newcommand{\cmsosetx}{\textnormal{CMSO}_1(\sigma,q,\bar X)}

\newcommand{\mybox}[1]{%
	\bigskip
    \noindent{\center\fbox{%
        \parbox{\textwidth - 3\fboxsep}{%
			#1%
    }}}%
	\bigskip
}

\begin{document}

\clearpage
\thispagestyle{empty} 

	\title{Approximate Evaluation of \\ Quantitative Second Order Queries}
	\author{Jan Dreier\\TU Wien\\ \href{mailto:dreier@ac.tuwien.ac.at}{dreier@ac.tuwien.ac.at} \and Robert Ganian\thanks{acknowledges support by Austrian Science fund (FWF) START project Y1329}\\TU Wien \\ \href{mailto:rganian@gmail.com}{rganian@gmail.com} \and Thekla Hamm\thanks{acknowledges support by Austrian Science fund (FWF) project J4651-N}\\Utrecht University \\ \href{mailto:t.l.s.hamm@uu.nl}{t.l.s.hamm@uu.nl}}
	\date{}
	\maketitle
	
	\begin{abstract}
	{
	\sloppy
Courcelle's theorem and its adaptations to cliquewidth have shaped the field of exact parameterized algorithms and are widely considered the archetype of algorithmic meta-theorems. In the past decade, there has been growing interest in developing parameterized approximation algorithms for problems which are not captured by Courcelle's theorem and, in particular, are considered not fixed-parameter tractable under the associated widths.

We develop a generalization of Courcelle's theorem that yields efficient approximation schemes for any problem that can be captured by an expanded logic we call \blockMSO, capable of making logical statements about the sizes of set variables via so-called weight comparisons. The logic controls weight comparisons via the quantifier-alternation depth of the involved variables, allowing full comparisons for zero-alternation variables and limited comparisons for one-alternation variables. We show that the developed framework threads the very needle of tractability: on one hand it can describe a broad range of approximable problems, while on the other hand we show that the restrictions of our logic cannot be relaxed under well-established complexity assumptions.

The running time of our approximation scheme is polynomial in \(1/\eps\), allowing us to fully interpolate between faster approximate algorithms and slower exact algorithms. This provides a unified framework to explain the tractability landscape of graph problems parameterized by treewidth and cliquewidth, as well as classical non-graph problems such as Subset Sum and Knapsack.	
}
\end{abstract}

\clearpage
\thispagestyle{empty} 
\tableofcontents

\setcounter{page}{1}
\section{Introduction}

Courcelle’s celebrated theorem~\cite{Courcelle90} establishes 
that all problems expressible in monadic second order logic (\MSO$_2$) can be solved
in linear time on graph classes of bounded treewidth.
Its modern formulations \cite{ArnborgLS91,CourcelleMR2000} strengthen this result to problems expressible in \LinEMSO$_2$, i.e., optimization problems where for a given CMSO$_2$ formula\footnote{The framework is often referred to as ``LinEMSO'', but here we opt to consciously include the ``C'' to emphasize that the framework supports CMSO formulas~\cite{GH2015,Oliveira21}.} with free set variables,
one asks for a satisfying assignment of the free variables that optimizes a given linear function.
Together with the variant for \LinEMSO$_1$-definable problems on graph classes of bounded cliquewidth~\cite{CourcelleMR2000},
these famous results are considered the
archetype of algorithmic meta-theorems---universal theorems that establish the
tractability of a broad range of computational problems on structured inputs.

While CMSO extends MSO with the ability to compare set sizes modulo a fixed constant,
the arguably largest drawback of the formalism is the inability
to compare set sizes in the standard (i.e., non-modular) arithmetic.
This prevents the framework from capturing problems that make statements about set sizes beyond simple maximization or minimization,
such as \textsc{Equitable Coloring}~\cite{BodlaenderF05,LuoSSY10}, \textsc{Capacitated Vertex Cover}~\cite{GuhaHKO03} or \textsc{Bounded Degree Vertex Deletion}~\cite{GanianKO21}. 
And indeed, this is for good reason: there are well-established lower bounds that exclude the
fixed-parameter tractable evaluation of even very simple queries with size
comparisons when parameterized by the treewidth and query size. But while the
associated reductions rule out the possibility of handling exact statements
about set sizes, in many cases of interest (and especially those motivated by
practical applications) it would be more than sufficient to compare them in an
approximate sense. 

Numerous extensions of \LinEMSO\ have been
proposed that incorporate some degree of access to set
sizes~\cite[Subsection 1.2]{KnopKMT19} and which trade greater expressive power for, e.g., weaker notions of tractability.
In this work, we provide a unifying investigation whether---and to
what extent---one can strengthen Courcelle's theorem 
(for \LinEMSO\ on treewidth and cliquewidth) to also incorporate 
exact and approximate optimization and comparison of set sizes and weights in standard arithmetic.
We identify a clear trade-off between the tractability of a problem
and the degree of approximation one allows.
This reproves many exact and approximate results in a unified framework
and characterizes the complexity of many additional (exact and approximate)
problems on graphs of bounded treewidth or cliquewidth (presented in \Cref{sec:problems}).

\subsection{Overview of Our Contribution} 

The arguably most general formulation of Courcelle's theorem 
states that \LinEMSO$_1$-queries can be efficiently evaluated on graph classes of bounded treewidth or cliquewidth.
Since our formalism extends this concept, let us start by defining it.
A \LinEMSO$_1$-query consists of a formula $\phi(X_1 \dots X_k)$ in Counting Monadic Second Order (CMSO$_1$) logic and an \emph{optimization target}
$\target(X_1 \dots X_k)$,
and the task is to find a satisfying assignment
of the free set variables $X_1,\dots,X_k$ maximizing $\target$ on some
(potentially vertex-weighted) input graph~$G$~\cite{CourcelleMR2000}.
Roughly speaking, the optimization target is a linear combination of the cardinalities and vertex-weights associated with $X_1,\dots,X_k$.

A natural attempt at extending Courcelle's theorem to accommodate statements about set sizes would be to not only optimize a target but
to also allow comparisons between linear \emph{weight terms} of the form
$$
t(X_1 \dots X_\ell) := a + \sum_{1 \le i,j \le \ell}a_{ij} w_i(X_j) 
$$
as atomic building blocks of the logic; in particular, these weight terms have the same structure as the optimization targets of \LinEMSO$_1$-queries. 
On a weighted graph, such terms sum up the \(i\)th weight $w_i$ of all vertices in the set \(X_j\), 
scale it by a factor \(a_{ij}\) and add an offset \(a\).
We call the resulting extended logic \MSOcomp\ (defined properly in~\Cref{sec:defExtendedLogic}). 
In \MSOcomp, we have, for example, atoms \(|X| \le |Y|\) comparing the cardinalities of two sets.
Therefore, \MSOcomp\ can easily capture many interesting problems 
such as the aforementioned \textsc{Equitable Coloring}, \textsc{Capacitated Vertex Cover} and \textsc{Bounded Degree Vertex Deletion}, which are already known to be \W[1]-hard when parameterized by treewidth. In fact, as we will later see in Theorem~\ref{thm:nphard-notnice}, \MSOcomp\ can even capture problems that are \NP-hard on trees of bounded depth.
In other words, this logic is too powerful, and we cannot hope for a fixed-parameter evaluation algorithm for \MSOcomp.

\paragraph*{Approximation Framework.}
We therefore turn to approximation.
For this, we follow the formalism of Dreier and Rossmanith for approximate size comparisons in first order logic~\cite{dreier2020approximate}, and consider how solutions might react to small shifts of the constraints within a given small ``accuracy bound'' \(0 < \eps \leq 0.5\).
For example, let us consider the \textsc{Bounded Degree Vertex Deletion} problem which asks for a minimum-size vertex set $W$ such that deleting $W$ bounds the maximum degree to an input-specified integer $d$, where the interesting case is when $d$ is large. This problem can be seen as a minimization problem where the following \MSOcomp-constraint describes the feasible solutions
$$
\phi(X) := \forall x\not \in X~\forall Y: (|Y|\leq d) \vee (\text{``\(Y\) is not the set of neighbors of $x$ after deleting $X$''}).
$$
We say the \emph{\((1+\eps)\)-loosening} or \emph{\((1+\eps)\)-tightening} of a formula is obtained by relaxing or strengthening all comparisons by factors \((1+\eps)\).
In our example, this corresponds (modulo arithmetic details) to the two \MSOcomp-formulas
$$
\forall x\not \in X~\forall Y: (|Y|\leq (1+\eps)\cdot d) \vee (\text{``\(Y\) is not the set of neighbors of $x$ after deleting $X$''}),
$$
$$
\forall x\not \in X~\forall Y: (|Y|\leq (1-\eps)\cdot d) \vee (\text{``\(Y\) is not the set of neighbors of $x$ after deleting $X$''}).
$$
Clearly, the loosened (or tightened) constraint is slightly easier (or harder) to satisfy than the original one,
and thus we may find smaller deletion sets \(W\) with respect to the loosened constraint than the original one, and these may be yet smaller than for the tightened one.
However, the only time when optimal solutions for the loosened and tightened constraint do not agree, is when
\((1+\eps)\)-perturbations in the coefficients either remove minimal solutions or add even smaller solutions.
It can  be  argued  that answering
queries in such an approximate  fashion is in many settings almost as good as
an exact answer because constraints based on large numbers are often only
ballpark estimates (e.g., does vertex $v$ have capacity exactly $1000$, or
could it be $998$?). 

Since \textsc{Bounded Degree Vertex Deletion} is \W[1]-hard when parameterized by treewidth~\cite{BetzlerBNU12}, we cannot hope to find a minimal solution to the problem in FPT time.
The algorithmic meta-theorem that we will present soon finds in time \(f(|\phi|,\cw(G),\eps)\cdot d^{0.001}\cdot|V(G)|^2\) (i.e., fixed-parameter tractable w.r.t.\ the length of the formula where numbers are counted as single symbols, the cliquewidth of $G$, and the accuracy $\eps$)
a set \(W^+\) that is at least as good as the optimal solution, but may not satisfy the constraint.
We guarantee, however, that \(W^+\) satisfies the \((1+\eps)\)-loosened constraint, 
and thus our answer \(W^+\) would be fine if we allowed for a bit of ``slack'' in the constraint.
Hence, in our example, the deletion of \(W^+\) results in a graph of maximum degree at most $(1+\eps)\cdot d$
and \(W^+\) is smaller or equal than any deletion set to degree $d$.

At times, even slightly overshooting a constraint can have dramatic consequences.
For such situations, our algorithmic meta-theorem additionally returns a set \(W^-\) 
that certainly satisfies the constraint but may be worse than the minimal solution.
However, our algorithm guarantees that \(W^-\) is at least as good as any solution to the \((1+\eps)\)-tightened constraint.
Thus, \(W^-\) beats any solution that satisfies a small ``safety margin''.
In our example, the deletion of \(W^-\) results in a graph of maximum degree at most $d$
and \(W^-\) is smaller or equal than any deletion set to degree $(1-\eps)\cdot d$.

Generally speaking, our algorithm returns two sets \(W^-\) and \(W^+\), where
\begin{itemize}
    \item \(\bar W^-\) satisfies all the constraints and is at least as good as the optimum for the tightened constraints,
    \item \(\bar W^+\) satisfies all the loosened constraints and is at least as good as the optimum for the original query.
\end{itemize}
We may call \(W^-\) a \emph{\((1+\eps)\)-conservative} solution and 
\(W^+\) a \emph{\((1+\eps)\)-eager} solution to our optimization problem.
We know that it is computationally infeasible to get a solution that both satisfies the constraint and is optimal,
but the conservative and eager solutions allow us to strike a nuanced trade-off between these two requirements (see Figure~\ref{fig:conservativeeager}). 

The precise notion of approximate solutions is presented in \Cref{def:approxAnswer} within \Cref{sub:approx}.
We believe this notion provides a robust framework that neatly
captures many natural approximation problems.
For example, translating known approximation results for
\textsc{Bounded Degree Vertex Deletion} as well as, e.g., \textsc{Capacitated Dominating Set} or \textsc{Equitable Coloring} on graphs of bounded treewidth~\cite{Lampis14}
to our framework (see also \Cref{sec:problems}) precisely corresponds to computing \((1+\eps)\)-eager solutions \(W^+\) for the problem.
We furthermore naturally generalize the notion of Dreier and Rossmanith~\cite{dreier2020approximate} from decision queries to optimization queries.

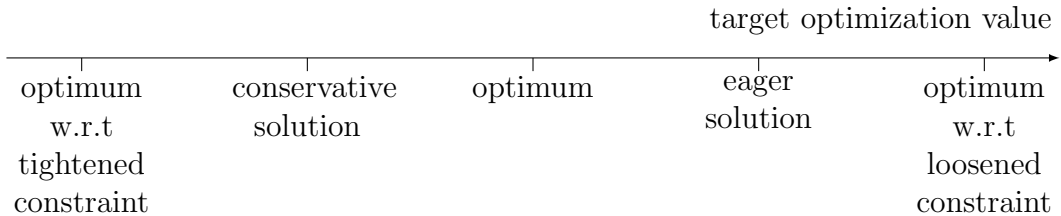
\begin{figure}[h]
	\centering
\begin{tikzpicture}
	\draw[-latex] (0,0) -- node[above right = .5cm and 2.2cm,anchor=west]{target optimization value} (14,0);
	\node[text width=2cm,align = center,anchor=north] (oversat) at (1,-.1) {optimum w.r.t tightened constraint};
	\node[text width=2cm,align = center,anchor=north] (return1) at (4,-.1) {conservative solution};
	\node[text width=2cm,align = center,anchor=north] (opt) at (7,-.1) {optimum};
	\node[text width=2cm,align = center,anchor=north] (return2) at (10,-.1) {eager solution};
	\node[text width=2cm,align = center,anchor=north] (undersat) at (13,-.1) {optimum w.r.t loosened constraint};
	\foreach \x in  {1,4,...,13}
	\draw (\x,0) -- (\x,-5pt);
\end{tikzpicture}
\caption{
Depiction of the two types of answers of our approximation algorithm.
The conservative solution still satisfies the original constraint and the eager solution still satisfies the loosened constraint.
If the optimal solution for the loosened and tightened constraint coincide, we are guaranteed to obtain an optimal solution.
}
\label{fig:conservativeeager}
\end{figure}

\paragraph{Relation to Other Approximation Frameworks.}
Note that our notion differs from the notion of approximation used by most approximation algorithms. 
\begin{itemize}
    \item
    Usually, one strives for solutions to optimization problems that satisfy a given constraint and are guaranteed to be at most a multiplicative factor off from the optimal solution, as measured by some optimization target.
    The ratio between the optimal solution and the found solution is called the \emph{approximation ratio}.
    We may call this setting \emph{optimization target approximation}.
    \item
    In our setting, solutions are bounded in terms of the best possible values obtainable by slightly perturbing the constraint of the problem.
    We may call this setting \emph{constraint approximation}.
\end{itemize}
A (conservative) solution in the constraint approximation setting may have an arbitrarily bad approximation ratio,
if slightly changing the constraints leads to a dramatic drop in solution quality.
On the other hand, it is also possible that slightly changing the constraints incurs no loss in solution quality,
in which case constraint approximation returns optimal solutions, while optimization target approximation may not. 
Thus, in general, optimization target approximation and constraint approximation are incomparable,
and we believe both notions to be worth studying.
When it comes to logic-based meta-theorems which are also capable of capturing \NP-hard \emph{decision} problems (which may not have a direct formulation as optimization problem),
it is natural to consider constraint approximation.

\paragraph*{Finding the Right Fragment.}
Perhaps surprisingly, we show that a meta-theorem with respect to constraint approximation cannot be obtained for the logic \MSOcomp.
In fact, we prove that already for a constant (and arbitrarily
bad) accuracy, approximately evaluating a \MSOcomp-query $\phi$ is
\W[1]-hard when parameterized by $|\phi|$ even on the class of paths (\Cref{thm:whard-notnice}).
But there is hope.
We notice that many problems of interest can be described using only comparisons between
sets that are either free or quantified in a shallow way.
For example, the natural encoding of \textsc{Equitable Coloring} with $\ell$-many colors only uses comparisons between the free variables.
$$
\phi(X_1 \dots X_\ell)=(\text{``all pairs of sets have equal size \((\pm 1)\)''}) \wedge (``\text{the sets form an $\ell$-coloring''})
$$
On the other hand, for the \textsc{Bounded Degree Vertex Deletion} problem mentioned above or for the encodings of problems such as \textsc{Capacitated Dominating Set} and \textsc{Capacitated Vertex Cover}, 
we need to access the cardinality only of the free variables and of outermost universally quantified sets (see Figure~\ref{fig:usualcomparisons} below and \Cref{sec:problems} for the encoding of these problems).

By keeping track of the approximate sizes of free set variables, one could show that approximate
evaluation is fixed-parameter tractable if we restrict ourselves
to formulas where comparisons are only allowed between free variables.
This would lead to an approximation algorithm for, e.g., \textsc{Equitable Coloring}, but not for the other examples listed above.
On the other hand,
surprisingly, we show that approximate evaluation is intractable if we additionally allow comparisons involving outermost universally quantified sets (\Cref{thm:whard-notnice}).
Thus, at first glance, the dividing line looks as in \Cref{fig:usualcomparisons} below.
However, we want to define a logic that captures the aforementioned tractable approximation problems,
and simply restricting the comparisons to shallow quantifiers is not sufficient to capture these---we have to look closer into the structure of the formulas to draw the border between tractability and hardness.

\begin{figure}[h]
    \center
\includegraphics{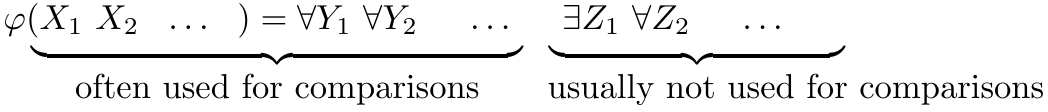}

\bigskip
\bigskip

\includegraphics{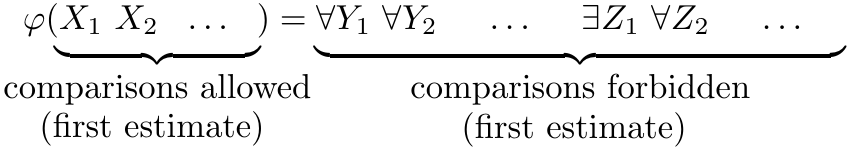}
\caption{
Top: The encoding of many natural problems uses only size comparisons of the outermost quantified sets.
Bottom: The first estimate of tractable formulas is however not suitable to capture them.}
\label{fig:usualcomparisons}
\end{figure}

\paragraph{Blocked CMSO.}
Our central insight in that regard is that hardness stems from the ability to compare cardinalities of
universally quantified variables to each other. 
Indeed, it turns out forbidding such interactions is the main prerequisite for
an expressive extension of Courcelle's theorem that allows approximate comparisons between weight terms.
We center our extension around a logic we call \emph{Blocked CMSO},
which will turn out to be powerful enough
to capture many major problems that were known to be approximable
and many more that were previously not known to be,
but which stays weak enough to still admit an efficient approximate evaluation algorithm.

While we formally define \emph{Blocked CMSO} (denoted \blockMSO) in~\Cref{sec:bMSO},
let us at least give an informal definition for now.
Remember that \MSOcomp\ extended \CMSO\ by also allowing size comparisons between \emph{weight terms} as atoms.
For example, the atom
\[
   |X_1|+2|X_2| < 100 + |Y|
\]
compares the left weight term with variables \(X_1,X_2\) to the right weight term with variable~\(Y\).
Now \blockMSO\ is the fragment of \MSOcomp\ that
contains all formulas which can be decomposed via disjunctions, conjunctions
and existential quantification into subformulas called \emph{block formulas} satisfying
three properties.
\begin{itemize} 
    \item Each block formula is of the form
    $\forall Y_1\dots\forall Y_k~\phi(X_1\dots X_\ell Y_1 \dots Y_k)$, where $X_1,\dots,X_\ell$ are free set
    variables and $\phi$ is an arbitrary \MSOcomp-formula containing only natural coefficients in its weight terms.
    \item All weight terms of $\phi$, except at most one, involve only set variables from $X_1,\dots,X_\ell$
    (we call such weight terms \emph{existential}).  
    \item The exceptional weight term may involve only variables from $X_1,\dots,X_\ell,Y_1,\dots,Y_k$ 
    (we call such a weight term \emph{universal}).
\end{itemize}

It is worth noting that \blockMSO\ only restricts the structure of formulas
with weight comparisons and thus every \CMSO-formula is also a \blockMSO-formula.
Hence, every \LinEMSO-query can be seen as a degenerate
\blockMSO-optimization query with a single block and no weight comparisons at all.

Let us illustrate this definition with the help of an example problem.
For two intervals of natural numbers \(\sigma = [\ell_\sigma,\dots,r_\sigma]\) and \(\rho = [\ell_\rho,\dots,r_\rho]\),
the \textsc{($\sigma,\rho$)-Domination} problem
asks for a (minimum- or maximum-size) vertex set \(W\) in a graph such that
for every vertex \(v \in W\), the number of neighbors in \(W\) lies in the range \(\sigma\) and for every other vertex \(v \not\in W\),
the number of neighbors in \(W\) lies in the range \(\rho\).
For example, with \(\sigma = \{0\}\) and \(\rho = \N\), this captures \textsc{Independent Set}
and with \(\sigma = \N\) and \(\rho = \N \setminus \{0\}\), this captures \textsc{Dominating Set}.
The following \MSOcomp-formula\footnote{For technical reasons discussed in \Cref{sec:defoflogic}, we prefer to not have any negations in front of comparisons, making these formulas look slightly awkward.} expresses the problem.
\begin{align*}
    \phi(X) = 
       & \Bigl[\forall x \in X~\forall Y: \big((\ell_\sigma \le |Y|) \land (|Y| \leq r_\sigma)\big) \lor \text{``\(Y\)  is not \textnormal{Neighborhood}$(x)\cap W$''} \Bigr] \\
 \land       & \Bigl[\forall x \not\in X~\forall Y: \big((\ell_\rho \le |Y|) \land (|Y| \leq r_\rho)\big) \lor \text{``\(Y\)  is not \textnormal{Neighborhood}$(x)\cap W$''} \Bigr].
\end{align*}
Note that this formula does \emph{not} lie in \blockMSO, since in each block (marked with square brackets), there are \emph{two} weight terms containing the universally quantified variable \(Y\).
In \Cref{thm:whard-notnice}, we prove that problems of this kind are hard to approximate\footnote{To make the reduction simpler, our proof falls short of proving hardness of approximation for \(\sigma\)-\(\rho\)-\textsc{Domination} specifically but holds for formulas of a similar structure.}.
However, if, for example, \(\ell_\sigma = 1\) then we can replace \(\ell_\sigma \le |Y|\) with the pure \CMSO-formula ``Y is nonempty''
and if \(r_\rho = \infty\), we can replace \(|Y| \leq r_\rho\) with the pure \CMSO-formula ``true''.
In this case, we obtain the formula
\begin{align*}
    \phi(X) = 
       & \Bigl[\forall x \in X~\forall Y: \big(\text{``$Y$ is nonempty''} \land (|Y| \leq r_\sigma)\big) \lor \text{``\(Y\)  is not \textnormal{Neighborhood}$(x)\cap W$''} \Bigr] \\
\land        & \Bigl[\forall x \not\in X~\forall Y: (\ell_\rho \le |Y|) \lor \text{``\(Y\)  is not \textnormal{Neighborhood}$(x)\cap W$''} \Bigr].
\end{align*}
This formula is in \blockMSO, since in each block, \(Y\) appears in only a single weight term.

\paragraph*{An Algorithmic Meta-Theorem for Approximation.}
As our main result, we establish the desired generalization of Courcelle's
Theorem to \blockMSO, for both treewidth and cliquewidth. The full formal statements of our meta-theorems are provided in the dedicated Section~\ref{sec:metat} after the necessary preliminaries. Below, we provide a simplified version that highlights the main features of these results.

\begin{theorem}[Simplified Main Result]
    \label{thm:mainsimple}
    Given a \blockMSOone- or \blockMSOtwo-query \((\phi,\target)\), an accuracy \(0 < \eps \leq 0.5\) and a graph $G$, we can compute a \((1+\eps)\)-eager solution and a \((1+\eps)\)-conservative solution in time $\bigoh^*\big((\frac{1}{\eps})^{f(|\phi|,\cw(G))}\big)$ or $\bigoh^*\big((\frac{1}{\eps})^{f(|\phi|,\tw(G))}\big)$, respectively, where $f$ is some computable function and $\bigoh^*$ hides polynomial factors.
\end{theorem}

Since \blockMSO-queries with no weight comparisons capture exactly \LinEMSO-queries,
and on such queries, approximate solutions coincide with exact solutions, Theorem~\ref{thm:mainsimple} can be seen as a generalization of the \LinEMSO-meta-theorem for cliquewidth and treewidth when $\eps$ is set to a constant. Moreover, our result implies fully polynomial-time approximation schemes (FPTAS) for all problems expressible by a fixed \blockMSO-query on graphs of bounded width, as well as \XP-time algorithms for evaluating \blockMSO-queries exactly (see Section~\ref{sec:metat}). 

Overall, Theorem~\ref{thm:mainsimple} captures a broad variety of problems that are \W[1]-hard when parameterized by treewidth or cliquewidth, including \textsc{Bounded Degree Deletion}, \textsc{Equitable Coloring}, \textsc{Capacitated Dominating Set} and \textsc{Graph Motif}. We also provide more genreal running time bounds that are better suited for handling graphs with weights encoded in binary, allowing us to directly express and approximate classical \NP-hard number-based problems such as \textsc{Subset Sum} and \textsc{Knapsack}. All of these as well as other applications of the meta-theorem are detailed in Section~\ref{sec:problems}.
Our main results not only provide a unified justification
for all these algorithms (and, in some cases, identify previously unknown tractability results), but also yield schemes with a flexible trade-off between the accuracy $\eps$ and the running time. 

Lastly, one may wonder whether the restrictions imposed by our logic \blockMSO\ are truly necessary. 
The restrictions enforce that outermost universally quantified variables may appear 
only in weight terms on a single side of a single comparison (for each block).
We show that a slight relaxation of this, allowing outermost universally quantified variables to appear in \emph{two} comparisons,
leads to \W[1]-hardness of approximation for every constant accuracy \(\alpha\), even on paths (\Cref{thm:whard-notnice}).
Exact evaluation of this relaxed kind of formulas also becomes \NP-hard on trees of depth three (\Cref{thm:nphard-notnice}).
This suggests that \blockMSO\ is the ``correct'' logic for approximation on graphs of bounded treewidth or cliquewidth.

\subsection{Proof Techniques and Technical Challenges}

\paragraph{Basic Idea.}
The established way~\cite{CourcelleMR2000} to evaluate formulas on graphs of bounded treewidth or cliquewidth is
to compute via dynamic programming the so-called \emph{\(q\)-types},
that is, the truth value of every formula up to a certain quantifier nesting depth \(q\).
For optimization problems, one extends this notion to sets and keeps for each \(q\)-type a set satisfying this type and optimizing the target function.
The size of such \(q\)-types depends only on \(q\) and the width of the graph (up to syntactical normalization) and they can
easily be propagated under the recoloring and connection operations of a cliquewidth expression.
The interesting case is the disjoint union operation, where one uses the \emph{Feferman--Vaught theorem}~\cite{Feferman1959} to propagate the \(q\)-types.
Assume we have graphs \(G_1,G_2\) whose \(q\)-types we know,
and we want to decide whether
\[
    G_1 \oplus G_2 \models \exists X \phi(X).
\]
The theorem gives us a small ``Feferman--Vaught set'' $\FV(\phi)$ such that this is equivalent to
\[
    G_1 \models \exists X \phi_1(X) \quad\text{ and }\quad G_2 \models \exists X \phi_2(X) \quad\text{ for some } (\phi_1,\phi_2) \in \FV(\phi).
\]
From the \(q\)-types of \(G_1\) and \(G_2\), we know whether \(G_1 \models \exists X \phi_1(X)\) and \(G_2 \models \exists X \phi_2(X)\),
and the above combination rule tells us whether we should mark \(\exists X \phi(X)\) as true for the \(q\)-type of \(G_1 \cup G_2\).

For our logic, we have to extend this mechanism to also consider weight comparisons.
So let us assume, for example, that \(n= |V(G_1)|+|V(G_2)|\) and we want to know for some \(0 \le m \le n\) whether 
\[
    G_1 \oplus G_2 \models \exists X \phi(X) \land |X| = m.
\]
Using the Feferman--Vaught theorem and also branching over all ways a set of size \(m\) can be split into two sets of size \(m_1\) and \(m_2\),
arriving at the equivalent statement
\begin{multline*}
    G_1 \models \exists X \phi_1(X) \land |X|=m_1 \quad\text{ and }\quad G_2
    \models \exists X \phi_2(X) \land |X|=m_2 \\ \text{ for some }
    (\phi_1,\phi_2) \in \FV(\phi) \text{ and \(m_1,m_2 \in \N\) satisfying } m_1 + m_2 = m.
\end{multline*}
Note that while before \(q\)-types had constant size, we now have to consider each possible value of \(m\), and thus have to store records of size \(\bigoh(n)\).
When keeping track of not just one, but \(k\) thresholds, the record sizes become \(\bigoh(n^k)\) and thus too large for FPT algorithms.

\paragraph{Rounding the Numbers.}
Naturally, the aforementioned issue also arises when attempting to solve individual problems when parameterized by treewidth or cliquewidth, as was done, e.g., by Lampis~\cite{Lampis14}. 
To remedy the issue, Lampis did not store all possible values for thresholds \(m\), but only considered numbers
of the form \((1+\delta)^i, i \in \N\), of which there are roughly \(\log(n)/\log(1+\delta) \approx \log(n)/\delta\) many in the range between \(0\) and \(n\).
Assume we store the answers to
\[
    G_1 \models \exists X \phi_1(X) \land |X| \le m_1 \quad\text{ and }\quad G_2
    \models \exists X \phi_2(X) \land |X| \le m_2 
\]
for all numbers \(m_1,m_2\) of the form \((1+\delta)^i\).
This means that we only need to store records or size roughly \(\log(n)/\delta\), but it also
means we know the ``critical threshold'' for \(m_1\) and \(m_2\) where the statements jump between true and false only up to a factor of \((1+\delta)\).
Since multiplicative errors are preserved under addition, we can use the Feferman--Vaught theorem
to obtain the ``critical threshold'' \(m\) for the following statement 
\[
    G_1 \oplus G_2 \models \exists X \phi(X) \land |X| = m
\]
up to a factor of \((1+\delta)\).
However, the sum of two numbers of the form \((1+\delta)^i\) may not be itself of the form \((1+\delta)^i\),
and thus we have to again round our thresholds to be of the form \((1+\delta)^i\).
This introduces another factor of \((1+\delta)\) of uncertainty to our knowledge.
These factors accumulate over time and hence performing \(d\) disjoint unions would lead to an uncertainty of \((1+\delta)^d\).
Lampis~\cite{Lampis14} solves this problem in a randomized way via an involved data structure called \emph{Approximate Addition Trees}
that uses a weighted random rounding and a careful probabilistic analysis to bound the error with high probability.
We take a conceptually simpler deterministic approach via so-called \emph{balanced cliquewidth-expressions},
which are cliquewidth-expressions of depth roughly \(\log(n)\).
This bounds the uncertainty by roughly \((1+\delta)^{\log(n)}\).
By choosing \(\delta \approx \eps/\log(n)\), our records remain sufficiently small and we get the desired final uncertainty of \(1+\eps\).
We remark that Lampis also noted the possibility of using balanced decompositions to derive some of his results (acknowledging feedback from an anonymous reviewer), and discussed the advantages and disadvantages of both approaches~\cite[Section 4 in the full version]{Lampis14}.

Let us describe the implied records for the query evaluation problem.
If \(\terms_1\) is the set of weight terms of the input formula
and \(\bar X\) stands for a tuple of variables, 
then our records approximate for each formula \(\phi\) (of small quantifier-rank) and 
for all thresholds \(g_{t,\lle}, g_{t,\gge}\) of the form \((1+\delta)^i\)
the truth value of the statement
\[
    \exists \bar X \colon \phi(\bar X) \land
    \bigwedge_{t \in \terms_1} t(\bar X) \le g_{t,\lle} \land t(\bar X) \ge g_{t,\gge}.
\]
For solving decision problems it is sufficient to existentially quantify over \(\bar X\),
but since we mainly do optimization, one would actually store the assignment to \(\bar X\) maximizing the target function.
In both cases, the Feferman--Vaught theorem can be adapted to propagate these records,
inducing a multiplicative ``loss of accuracy'' of at most \(1+\delta\)
for each application.

\paragraph{Deeper Comparisons.}
While it requires some careful technical considerations, what we discussed is so far conceptually still quite straightforward
and limited.
It allows us to do approximate formulas with comparisons in the free (or outermost existentially quantified) variables,
such as \textsc{Equitable Coloring},
but it does not capture problems such as \textsc{Bounded Degree Deletion}, \textsc{Capacitated Vertex Cover} or \textsc{Capacitated Dominating Set} for which deeper comparisons seem to be required.
\begin{figure}[h]
    \center
\includegraphics{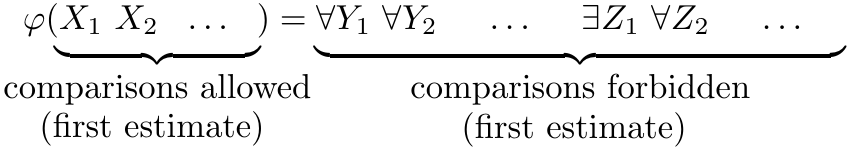}
\end{figure}

To capture these deeper comparisons, we not only store whether there exists \(\bar X\) of certain \(q\)-type
satisfying certain size bounds,
but we also store whether additionally all \(\bar Y\) of certain other \(q\)-types also satisfy other size bounds.
Let \(\terms_1\) be the simpler existential weight terms of the formula we want to approximately evaluate,
and let \(\terms_2\) be the universal weight terms of the formula, of which there may be at most one per block.
We further denote by \(\cmsoset\) all \CMSO$_1$-formulas of quantifier-rank at most \(q\), free variables from \(\bar X\bar Y\)
and signature \(\sigma\) (of which there are only few).

Our first central technical contributions is the computation of records
that approximate for all thresholds \(g_{t,\lle}, g_{t,\gge},g_{t,\psi,\lle},g_{t,\psi,\gge}\) of the form \((1+\delta)^i\)
the truth value of the formula 
\begin{multline}
\label{tableintro}
    \exists \bar X
    \Bigl(\bigwedge_{t \in \terms_1} t(\bar X) \le g_{t,\lle} \land t(\bar X) \ge g_{t,\gge} \Bigr) ~ \land~ \\
    \Bigl( \bigwedge_{t \in \terms_2} \bigwedge_{\psi \in \cmsoset} \forall \bar Y \psi(\bar X \bar Y) \lor \bigl( t(\bar X \bar Y) \le g_{t,\psi,\lle} \land t(\bar X \bar Y) \ge g_{t,\psi,\gge} \bigr) \Bigr).
\end{multline}
The second line is new compared to the previous simpler approach. 
It enforces for all universal terms \(t\) and formulas \(\psi\) that the value of \(t\) is bounded between 
\(g_{t,\psi,\gge}\) and \(g_{t,\psi,\lle}\) for all tuples \(\bar Y\) satisfying \(\neg\psi\).
Our second central technical contribution shows that
approximate answers to any \blockMSO-query can be derived from records of this kind.

\paragraph{Main Proof Structure.}
Let us give a bit more detail on how we break down these tasks into multiple steps, each of which are non-trivial.

\begin{itemize}
    \item In \Cref{sec:exactFV}, we first treat our records as exact (that is, they contain the exact truth value for all thresholds, not just the rounded ones).
        Using the Feferman--Vaught theorem as the basic building block, we describe a procedure (\Cref{thm:base}) that propagates these records over disjoint unions. 
        The correctness proof proceeds via a careful syntactical and semantical rewriting of these terms.
    \item In \Cref{sec:approxFV}, we show that relaxing the input assumptions to merely approximate input records (i.e., records which are only allowed to speak about ``rounded numbers'' that can be expressed as some power of $(1+\delta)$)
    still lets us compute
        approximate output records over disjoint unions. 
        A single invocation of this computation (given by \Cref{thm:baseApprox2}) incurs the previously described ``loss of accuracy'' due to rounding the numbers to the nearest power. 
        We often need to rescale coefficients to maintain ``conservative'' and ``eager'' estimates on the accuracy loss,
        introducing non-integer coefficients in our records.
        Therefore, we consider coefficients with some ``granularity'' \(\gran \in \N\), that is, we write them as fractions with denominator \gran,
        and show that during the rescaling of coefficients, we do not need to increase \gran too much.
    \item In \Cref{sec:tablecomp}, we perform dynamic programming over a cliquewidth-expression of logarithmic depth, invoking the previous item for disjoint unions
        and simple term-rewriting procedures for recolorings and edge insertion. This step holds no surprises and works the way one would expect. 
    \item At last, in \Cref{sec:lookup}, we show that an approximate answer to
        any \blockMSO-query can be obtained by combining various approximate
        answers to our record formulas.
        In dynamic programming algorithms, it is typically trivial to derive the final answers from the records stored at the root of the computation tree.
        In our case, this step is surprisingly involved and one of the main technical challenges we had to overcome.
        In particular, proving robustness with respect to approximation requires great care.
\end{itemize}

It is surprising that the presented records strike the right balance between being restricted enough
that they can be composed over disjoint unions and stored efficiently, and expressive enough to derive answers to \blockMSO-queries from and to capture all the problems discussed in \Cref{sec:problems}.
As implied by the hardness results of \Cref{sec:hardness}, hypothetical records for more
powerful logics either become too large or are no longer efficiently composable over disjoint unions.

A graphical overview of the proof progression for our main technical result, paralleling evaluation of \LinEMSO$_1$-queries via the Feferman--Vaught theorem, is presented in Figure~\ref{fig:overviewfigure}.
Due to the additional challenges we face, our proof detours via the formulas in our records as described above (we call these \emph{table formulas}).

\begin{figure}
	\label{fig:overviewfigure}
	\begin{tikzpicture}
		\begin{scope}[on background layer]
			\node at (-1,2.5) {\color{gray} \textsc{Known}};
			\draw[dashed,gray] (1.8,2.75) to (1.8,-5.5);
			\node at (2.5,2.5) {\color{gray} \textsc{New}};
			\draw[dotted,gray] (5.9,2.75) to (5.9,-5.5);
			\draw[dotted,gray] (9.95,2.75) to (9.95,-5.5);
		\end{scope}
		
		\node[rectangle,align=center] (MSOf) {\CMSO\(_1\)-formulas\\};
		
		\node[align=center] at (3.75,0) (bMSOf) {\blockMSOone-formulas\\ Sec.~\ref{sec:bMSO}};
		\node[align=center] at (8,0) (tabf) {exact table formulas\\ Sec.~\ref{sec:defoftables}};
		\node[align=center] at (12.5,0) (atabf) {approx.\ table formulas\\ Sec.~\ref{sec:approxtabd}};

		\node[rectangle,draw,below of = atabf,anchor = north,align=center] (atabFVT) {``\underline{\smash{approx.}}\ Feferman-\\Vaught'' for approx.\ \\ table formulas\\
			Thm.~\ref{thm:baseApprox2}};
		\node[rectangle,draw,align=center] at (tabf |- atabFVT) (tabFVT) {``non-\FPT\\Feferman--Vaught''\\for exact \\ table formulas \\
			Thm.~\ref{thm:base}};
		\node[rectangle,draw,below = 0.5cm of atabFVT,anchor = north,align=center] (atabCourcelle) {dynamic prog.\ for \\approx.\ table formulas
			\\Thm.~\ref{thm:computeTable}};	
		
		\node[rectangle,draw] (FVT) at (MSOf |- atabFVT) {Feferman--Vaught};
		\node[rectangle,draw,align=center] (Courcelle) at (MSOf |- atabCourcelle) {Meta-Theorem \\ for \LinEMSO$_1$};
		
		\node[rectangle,draw, ultra thick,align=center] (bMSOCourcelle) at (bMSOf |- Courcelle) {Meta-Theorem\\for \blockMSOone\\Thm.~\ref{thm:maintheorem}};
		
		\draw[latex-latex] (bMSOf) to[bend left = 35]  node[align=center] {
			Step~1 in proof of \\ Thm.~\ref{thm:almostlookup}\\~\\
			\phantom{\Huge Space}} (tabf);
		\draw[latex-latex] (tabf) to[bend left = 35]  node[align=center] {\underline{\smash{rounding}},\\
			Step~2 in proof of \\ Thm.~\ref{thm:almostlookup}\\~\\
			\phantom{\Huge Space}} (atabf);
		
		\draw[-latex] (FVT) to (Courcelle);
		
		\draw[-latex,ultra thick] (FVT) to (tabFVT);
		\draw[-latex,ultra thick] (tabFVT) to (atabFVT);
		\draw[-latex,ultra thick] (atabFVT) to (atabCourcelle);
		\draw[-latex,ultra thick] (atabCourcelle) to (bMSOCourcelle);
	\end{tikzpicture}
	\caption{A graphical overview of the proof progression for our main result (\Cref{thm:maintheorem}). The building blocks for the classical \LinEMSO$_1$-theorem are depicted on the left. The columns indicated by dashed lines distinguish the boxed theorems depending on which kind of formulas they apply to.
	Our targeted formulas are \blockMSOone-formulas.
	Our proof consists of two components (a) and (b).
	(a) Transitioning between the columns without incurring too much loss in accuracy:
	This component is confined entirely to \Cref{sec:lookup} and most importantly to the proof of \Cref{thm:almostlookup}, and depicted by the bent arrows.
	(b) Describing a dynamic program based on a generalization of Feferman--Vaught theorem:
	We first obtain a non-algorithmic version of the Feferman--Vaught theorem for exact table formulas of the form described in~(\ref{tableintro}) with exact thresholds, and use this to obtain an algorithmic version for approximate table formulas of the form described~(\ref{tableintro}) with rounded thresholds.
	This in turn allows us to obtain a dynamic programming algorithm for computing approximate records, which can be combined with (a) to prove our approximate meta-theorem for \blockMSOone.
	Underlines indicate both places where, in our main result, a loss in accuracy in the satisfied comparisons occurs.
	}
\end{figure}

\paragraph{Article Structure.}
General preliminaries are given in \Cref{sec:prelims}, followed by a formalization of our logic in Section~\ref{sec:defoflogic}. An overview and exact statements of the new meta-theorems for cliquewidth and treewidth are provided in Section~\ref{sec:metat}, while Section~\ref{sec:problems} presents a range of examples showcasing a few potential applications of our results, including not only standard graph problems but also classical number problems such as \textsc{Subset Sum}.
Sections~\ref{sec:exactFV} through Section~\ref{sec:final} provide the proof of out main technical result.
The final Section~\ref{sec:hardness} contains proofs for the accompanying lower bounds. 

\subsection{Related Work}
Extensions to Monadic Second Order logic towards counting and evaluation in standard arithmetic have been investigated by a number of authors. Szeider investigated an enrichment via the addition of so-called local cardinality constraints, which restrict the size of the intersection between the neighborhood of each vertex and the instantiated free set variables~\cite{Szeider11}, and established \XP-tractability for the fragment when parameterized by treewidth and query size. The addition of weight comparisons to \MSO\ exclusively over free set variables was investigated by Ganian and Obdr\v z\' alek~\cite{GanianO13}, who show that the resulting logic admits fixed-parameter model checking when parameterized by the vertex cover number, i.e., the size of the minimum vertex cover in the input graph. Knop, Koutecký, Masa\v r\' ik and Toufar~\cite{KnopKMT19} studied various means of combining local cardinality constraints with weight comparisons and established their complexity with respect to the vertex cover number and a generalization of it to dense graphs called neighborhood diversity. Crucially, none of the frameworks for resolving queries in these extensions of Monadic Second Order logic are fixed-parameter tractable when parameterized by the query size and treewidth (or cliquewidth).

The approximate evaluation of logical queries was pioneered by Dreier and Rossmanith in the context of counting extensions to first order queries~\cite{dreier2020approximate}. Before that, Lampis~\cite{Lampis14} developed algorithms providing approximate solutions to several problems parameterized by treewidth and cliquewidth that are captured by our framework, including \textsc{Equitable Coloring}~\cite{BodlaenderF05,LuoSSY10}, \textsc{Bounded Degree Deletion}~\cite{FellowsGMN11,GanianKO21}, \textsc{Graph Balancing}~\cite{EbenlendrKS14}, \textsc{Capacitated Vertex Cover}~\cite{GuhaHKO03} and \textsc{Capacitated Dominating Set}~\cite{KaoCL15}.
	
	\section{Preliminaries}
	\label{sec:prelims}
	For integers $i$, $j$ where $i\leq j$, we use $[i,j]$ to denote the set $\{i, i+1, \dots, j-1, j\}$ and \([i]\) to denote the set \(\{1,\dots,i\}\).
    We denote the \emph{power set} of a set \(V\) by \(\PP(V) = \{ A \mid A \subseteq V \}\).
    We will often write \(\bar X\) to denote a tuple \(X_1 \dots X_\ell\). The length of this tuple is denoted by \(|\bar X|=\ell\).
    The \(i\)th element of a tuple \(\bar X\) is always denoted by \(X_i\).
    For two tuples \(\bar{X}\) of length \(\ell\) and \(\bar{Y}\) of length \(\ell'\), we denote by the concatenation \(\bar{X} \bar{Y} := \bar{Z}\) the tuple of length \(\ell + \ell'\) given by \(Z_i = \begin{cases}
    	X_i & \mbox{ if } i \leq \ell,\\
    	Y_{i - \ell} & \mbox{ if } \ell < i \leq \ell + \ell'.
    \end{cases}\)

    For two tuples of sets \(\bar{X}\) and \(\bar{Y}\) of equal length \(\ell\) we define their entry-wise union as the tuple \(\bar{X} \cup \bar{Y} := \bar{Z}\) of length \(\ell\) given by \(Z_i = X_i \cup Y_i\) for all \(i \in [\ell]\).

	\subsection{Graphs}
	A \emph{(vertex-)colored (vertex-)weighted graph} is a tuple
	\[
	G = \bigl(V(G), E(G), (P^G_1,\dotsc,P^G_k), (w^G_1, \dotsc, w^G_l)\bigr),
	\]
	where $V(G)$ and $E(G) \subseteq {V(G) \choose 2}$ are the vertex and edge set,
	$P^G_i \subseteq V(G)$ are the \emph{color sets}, 
	and $w^G_i : V(G) \to \N$ are the \emph{weights}.
    For sets $W \subseteq V(G)$ we define the \emph{accumulated weights} $w^G_i(W) := \sum_{v \in W}w^G_i(v)$.
	Unless mentioned otherwise, all graphs in this paper are vertex-colored and vertex-weighted.
	
	When we specifically talk about treewidth, 
    it will be useful to consider graphs with weights and colors also applied to edges and not only vertices. In this case we deal with \emph{edge-colored edge-weighted graphs} of the form 	\[
	G = \bigl(V(G), E(G), (P^G_1,\dotsc,P^G_k), (\hat P^{G}_1,\dotsc,\hat P^{G}_h), (w^G_1, \dotsc, w^G_l), (\hat w^{G}_1, \dotsc, \hat w^{G}_t)\bigr),
	\]
	where additionally $\hat P^{G}_i \subseteq E(G)$ are \emph{edge color sets} and $\hat w^{G}_i : E(G) \to \N$ are \emph{edge weights}.

	\subsection{Cliquewidth} 
    For a positive integer \(k\), let a \emph{\(k\)-graph} be a graph with at most $k$ \emph{labels} and at most one label per vertex.
	The \emph{cliquewidth} of a (possibly colored and weighted) graph \(G\) is the smallest integer \(k\) such that
    the underlying graph of $G$ (without colors and weights but with the same edges and vertices) can be constructed from single-vertex \(k\)-graphs by means of iterative
	application of the following three operations:
	\begin{enumerate}
		\item Disjoint union (denoted by \(\oplus\));
		\item Relabeling: changing all labels \(i\) to \(j\) (denoted by \(\rho_{i\to j}\)) for some $1 \le i,j \le k$;
		\item Edge insertion: adding an edge from each vertex labeled by \(i\)
		to each vertex labeled by \(j\) for some $1 \le i \neq j \le k$ (denoted by
		\(\eta_{i,j}\)).
	\end{enumerate}
	We also write \(\cw(G)\) for the cliquewidth of a graph \(G\), and
	we call an expression of operations witnessing that \(\cw(G) \le k\) a \emph{\(k\)-expression} for \(G\).
    
	Let us explicitly point out that we distinguish between labels used in a \(k\)-expression of a graph and colors of that graph.
    In particular, if a $k$-expression witnesses the cliquewidth of a graph $G$, then the same $k$-expression also witnesses the cliquewidth of every coloring of $G$.

    While there is no known fixed-parameter algorithm for computing optimal $k$-expressions parameterized by the cliquewidth of a graph $G$, it is possible to obtain $k$-expressions where $k$ is upper-bounded by a function of $\cw(G)$. The most recent result of this kind is:
    \begin{theorem}[{\cite{FominK22}}]
        Given a graph \(G\), it is possible to compute a \(2^{\cw(G) + 2}\)-expression \(\chi\) of \(G\) in time \(\bigoh(2^{2^{\cw(G)}}|V(G)|^2)\).
    \end{theorem}
	
	The \emph{size} of a $k$-expression $\chi$, denoted $|\chi|$, is the number of operations it contains; we remark that $|\chi|$ can always be bounded by $f(k)\cdot |V(G)|$ for some computable function $k$. We refer to the \emph{depth} of a \(k\)-expression as the maximum number of nested operations. While the depth of a $k$-expression does not have a significant impact on the running time of evaluating \MSO{} queries~\cite{CourcelleMR2000}, it will have a significant impact on our approximate queries.
    We will therefore work with $k$-expressions of logarithmic depth, which can always be obtained by the following theorem.

    \begin{theorem}[{\cite{CourcelleV03}}]\label{thm:TransformLogarithmicDepth}
		Given a \(k\)-expression \(\chi\) for a graph \(G\), one can compute an equivalent \(k\)-expression 
        of depth at most \(f(k)\cdot\log(|V(G)|)\) and size at most \(f(k)\cdot|\chi|\) in time at most \(f(k)\cdot|\chi|\cdot\log(|\chi|)\), for some computable function \(f\).
	\end{theorem}
	\begin{proof}
		This statement is not given as an isolated theorem in \cite{CourcelleV03} but comprises the first steps in the proof of \cite[Theorem~3]{CourcelleV03}.
		In particular, using the notation and terminology of that proof, \(t\) is our given \(k\)-expression \(\chi\).
		From this we derive via \cite[Theorem~1]{CourcelleV03} in time in \(\bigoh(|\chi|\log(|\chi|))\), the expression \(t''\) over single-vertex \(k\)-graphs of a sequence of finitely many auxiliary binary operations \(\ooplus_\lambda\) obtainable from a finite sequence \(\lambda\) of \(\rho_{i \to j}\) and \(\eta_{i,j}\), as \(\chi_1 \ooplus_\lambda \chi_2 := \lambda(\chi_1 \oplus \chi_2)\).
		By the construction in the proof of \cite[Theorem~3]{CourcelleV03}, \(t''\) has depth at most \(3\log(|t|) + 3\).
		Lastly, \(t''\) is translated into the desired \(k\)-expression by resolving the auxiliary binary operation \(\ooplus_\lambda\) into a finite sequence of \(\rho_{i \to j}\), \(\eta_{i,j}\) and \(\oplus\) operations.
	\end{proof}
    Combining both results allows us to compute an approximately optimal \(k\)-expression of logarithmic depth.
    \begin{theorem}\label{thm:LogarithmicDepth}
        Given a graph \(G\), it is possible to compute a \(2^{\cw(G) + 2}\)-expression of 
        depth at most \(f(k)\cdot\log(|V(G)|)\) and size at most \(f(k)\cdot|V(G)|\) in time at most \(f(k)\cdot|V(G)|^2\), for some computable function \(f\).
    \end{theorem}

\subsection{Treewidth}
	
	We will not explicitly work with treewidth in our proofs, since the following observation
    allows us to obtain edge set quantification from vertex set quantification by subdivision of edges.
    \begin{observation}[{e.g., \cite[Section 7]{cygan2015parameterized}}]
		\label{obs:subdivision}
		For any graph \(G\) and the graph \(G'\) arising from \(G\) by subdividing each edge, it holds that \(\tw(G) = \tw(G')\).
	\end{observation}
	
	Moreover, bounded treewidth allows fast computation of \(k\)-expressions where \(k\) is bounded in a function of the treewidth.
	\begin{theorem}[{\cite{Korhonen21} together with \cite{CorneilR05}}]
		Given a graph \(G\), it is possible to compute a \((3 \cdot 4^{\tw(G) - 1} + 1)\)-expression \(\chi\) of \(G\) 
        in time \(f(\tw(G)) \cdot |V(G)|\), for some computable function \(f\).
	\end{theorem}

	Again, combining this with \Cref{thm:TransformLogarithmicDepth}, we obtain the following.
	\begin{theorem}\label{thm:LogarithmicDepthtw}
        Given a graph \(G\), it is possible to compute a \((3 \cdot 4^{\tw(G) - 1} + 1)\)-expression of depth at most \(f(\tw(G))\cdot\log(|V(G)|)\) and size at most \(f(\tw(G))\cdot|V(G)|\) in time at most \(f(\tw(G)) \cdot |V(G)|\), for some computable function \(f\).
	\end{theorem}

    \subsection{Logic}

    \paragraph*{Syntax and Semantics.}
    In our proofs, we consider counting monadic second order logic \emph{without} quantification over edge sets, commonly denoted \CMSO$_1$.
    In this logic, formulas are built
    via logical connectives (\(\lor\), \(\land\) and \(\neg\))
    and universal and existential quantification (\(\forall\) and \(\exists\))
    over vertices and vertex sets (typically lower case letters \(x,y,z,\dots\) for vertices
    and upper case letters \(X,Y,Z,\dots\) for sets)
    from the following atoms
    \begin{itemize}
        \item \(x \in X\), where \(x\) is a vertex variable and \(X\) is a vertex set variable,
        \item equality \(x=y\) between vertex variables \(x\) and \(y\),
        \item \(\relE(x,y)\), expressing the existence of an edge between vertex variables \(x\) and \(y\),
        \item \(U(x)\), expressing that a vertex variable \(x\) has color \(U\),
        \item \(\card_{a,r}(X)\) where \(a, r \in \N\) and \(a < r \geq 2\), and \(X\) is a vertex set variable with the interpretation that \(|X| = a \mod r\).
    \end{itemize}

    Additional connectives and atoms may be constructed using those listed
    above (e.g., the implication $\rightarrow$ or equality \(X=Y\) between sets) and therefore are
    not explicitly included as basic building blocks.
    Moreover, we will make the simplifying assumption that no variable is quantified over multiple times.
    By renaming variables for which this is the case, we can ensure this without any loss of expressiveness.

	\CMSO\(_2\) (counting monadic second order logic \emph{with} quantification over edge sets) is the extension of \CMSO\(_1\) with quantification over edges and edge sets.
    In this logic, atomic predicates for containment, equality, colors, and \(\card_{a,r}\) are extended to also account for edge and edge set variables,
    and new atoms \(\textnormal{inc}(v,e)\) between vertex variables \(v\) and edge variables \(e\) are introduced,
    expressing that \(v\) is an endpoint of \(e\).
	It is not difficult to see that the evaluation of a \CMSO\(_2\)-sentence on a graph \(G\) can be reduced to the evaluation of a \CMSO\(_1\)-sentence with proportional length on the subdivision of \(G\), where the new subdivision vertices receive an auxiliary color that marks them as ``edges'' of the original graph.
    We can therefore focus on \CMSO$_1$ in our proofs and statements.
	Whenever we omit the index, what is written applies analogously to both \CMSO\(_1\) and \CMSO\(_2\).

    \paragraph{Basic Notation.}
    The \emph{length} of a formula \(\phi\), denoted by \(|\phi|\), is the number of symbols it contains.
    We write \(\phi(\bar X_1,\dots,\bar X_k)\) to indicate that a formula \(\phi\) has free set variables \(X_1,\dots,X_k\).
    This is often abbreviated as \(\phi(\bar X)\), where \(\bar X\) stands for a tuple of set variables \(X_1,X_2,\dots\) of unnamed length, denoted by \(|\bar X|\).

    The \emph{signature} \(\sigma\) of a formula \(\phi\) is the set of atomic relations (besides containment and equality) used by \(\phi\).
    Similarly, the signature of a graph consists of its edge relation together with all used vertex and edge colors.
    If the signature of a sentence \(\phi\) matches the signature of a graph \(G\), (that is, \(\phi\) only speaks of colors that \(G\) knows of)
    we write \(G \models \phi\) to indicate that \(\phi\) is true in \(G\).
    We will implicitly assume that signatures between formulas and graphs match without explicitly mentioning this.
    
    Instead of choosing a tuple \(\bar W = W_1,\dots,W_\ell\) with \(W_i \subseteq V(G)\),
    we will from now on equivalently say $\bar W \in \PP(V(G))^{\ell}$.
    For a formula \(\phi(\bar X)\) with $\bar X = X_1 \dots X_\ell$, 
    and a tuple $\bar W \in \PP(V(G))^{\ell}$, 
    we write $G \models \phi(\bar W)$ to indicate that \(\phi\) is true on \(G\) when interpreting each free variable \(X_i\) by \(W_i\).
	We say the \emph{solutions of \(\phi\) in \(G\)} are the elements of the set \(\phi(G) := \{\bar W \in \PP(V(G))^{\ell} \mid G \models \phi(\bar W)\}\).

    \paragraph{Normalization.}

    The \emph{quantifier rank} of a formula is the maximal nesting depth of its quantifiers.
    We denote the set of all \CMSO\(_1\)-formulas with quantifier rank at most \(q\),
    whose signature is a subset of some signature \(\sigma\) and whose free variables are
    contained in \(\bar{X}\) by \(\cmsosetx\).
    We assume throughout this paper that all formulas are syntactically normalized by standardizing the variable names and deleting duplicates from conjunctions and disjunctions.
    Therefore, the lengths of our formulas can be bounded as follows.

    \begin{observation}\label{obs:msosizebound}
        The size of \(\cmsosetx\) and the length of the formulas within is bounded by a computable function of \(|\sigma|\), \(q\), and \(|\bar X|\).
    \end{observation}

\paragraph*{Feferman--Vaught Theorem.}
This classical result~\cite{Feferman1959} lets us propagate information about \CMSO$_1$-formulas across disjoint unions of graphs.
The modern formulation, as given, e.g., in~\cite[Theorem 7.1]{makowsky2004algorithmic} states the following.

\begin{theorem}\label{thm:fvBooleanFunction}
	For every formula $\phi(\bar X) \in \cmsosetx$ one can effectively compute
	sequences of formulas $\phi_1^1,\dots,\phi_1^m \in \cmsosetx$, $\phi_2^1,\dots,\phi_2^m \in \cmsosetx$,
	and a Boolean function $B \colon \{0,1\}^{2m} \to \{0,1\}$ such that
	for all graphs $G_1, G_2$ and $\bar W_1 \in \PP(V(G_1))^{|\bar X|}$, $\bar W_2 \in \PP(V(G_2))^{|\bar X|}$,
	\begin{multline*}
		G_1 \oplus G_2 \models \phi(\bar W_1 \cup \bar W_2) \iff \\
		B\bigl(G_1 \models \phi_1^1(\bar W_1),\dots,G_1 \models \phi_1^m(\bar W_1),
		G_2 \models \phi_2^1(\bar W_2),\dots,G_1 \models \phi_1^m(\bar W_2)\bigr).
	\end{multline*}
\end{theorem}

Converting the Boolean function $B$ into disjunctive normal form yields the following nicer statement,
associating with every formula \(\phi(\bar X)\) a ``Feferman--Vaught set'' \(\FV(\phi)\).

\begin{theorem}\label{thm:fv}
	For every formula $\phi(\bar X) \in \cmsosetx$ one can effectively compute a set
	$\FV(\phi) \subseteq \cmsosetx \times \cmsosetx$ such that
	for all graphs $G_1, G_2$ and $\bar W_1 \in \PP(V(G_1))^{|\bar X|}$, $\bar W_2 \in \PP(V(G_2))^{|\bar X|}$,
	$$
	G_1 \oplus G_2 \models \phi(\bar W_1 \cup \bar W_2) \iff
	G_1 \models \phi_1(\bar W_1) \text{ and } G_2 \models \phi_2(\bar W_2) \text{ for some } (\phi_1,\phi_2) \in \FV(\phi).
	$$
\end{theorem}
\begin{proof}
	If we convert the Boolean function $B$ of \Cref{thm:fvBooleanFunction} into a disjunction of conjunctive clauses, we obtain
	that $G_1 \oplus G_2 \models \phi(\bar W_1 \cup \bar W_2)$
	is equivalent to a disjunction of statements of the form
	$$
	\bigwedge_{i=1}^m
	G_1 \models (\neg) \phi_1^i(\bar W_1) \land
	\bigwedge_{i=1}^m
	G_2 \models (\neg) \phi_2^i(\bar W_2),
	$$
	where the notation $(\neg)$ symbolizes an optional negation in front of each conjunct.
	This is of course equivalent to
	$$
	G_1 \models
	\bigwedge_{i=1}^m
	(\neg) \phi_1^i(\bar W_1) \land
	G_2 \models
	\bigwedge_{i=1}^m
	(\neg) \phi_2^i(\bar W_2).
	$$
	We choose $\FV(\phi)$ to be all the tuples of the form
	$(\bigwedge_{i=1}^m (\neg) \phi_1^i(\bar X), \bigwedge_{i=1}^m (\neg) \phi_2^i(\bar X))$
	that occur in this disjunction.
\end{proof}

\section{Meta-Theorems and Extended Logics}\label{sec:defoflogic}

We present the relevant logical and algorithmic notions step by step,
starting with optimization for \CMSO\ which we then extend with size comparisons and approximation.

\subsection{Basic Optimization}
We first consider the classical optimization meta-theorem for graphs of bounded cliquewidth by Courcelle, Makovsky and Rotics~\cite{CourcelleMR2000},
paraphrasing it in the language used throughout this paper.
For this, we define the notion of \emph{weight terms} which the classical result optimizes
and which will also be the basic building block of our extended size-comparison logic. 
	
\paragraph*{Weight Terms.}
We define a \emph{weight term} as a term $t$ of the form 
$$
t(X_1 \dots X_\ell) := a + \sum_{1 \le i,j \le \ell}a_{ij} w_i(X_j) 
$$
where $a_{ij}, a \in \R$, each $w_i$ is a \emph{weight symbol} and each $X_j$ is a vertex set variable.
In some applications, one may want to use vertex variables \(x_j\) instead of \(X_j\) inside a weight term.
To keep things clean, we exclude this from our formal definition, 
but it is easily simulated by requiring \(X_j\) to be a singleton-set in an accompanying \CMSO$_1$ formula.
The weight symbols will be interpreted by the weights of a given graph.
To be precise, given a graph $G$ of matching signature and sets $W_1,\dots,W_k \subseteq V(G)$,
weight terms are \emph{interpreted} as
$$
t^G(W_1 \dots W_\ell) := a + \sum_{1 \le i,j \le \ell}a_{ij} w^G_i(W_j).
$$
In the context of \CMSO$_2$, this notation is naturally extended to weight terms \emph{with edge weights} that take both edge set and vertex set variables as input.
As a shorthand, we often write, for example, \(|X|\) for the weight term measuring the size of \(X\),
that is, the term \(w(X)\) where the implicit weight symbol \(w\) is interpreted as \(w^G(W) = |W|\).
When the graph $G$ is clear from the context, we usually drop it from the superscript,
unifying the logical symbols and their corresponding interpretation in the graph.
This slight abuse of notation means that we for example
commonly write $P_i \subseteq V(G)$ instead of \(P_i^G\) for color sets, $w_i : V(G) \to \N$ instead of \(w_i^G\) for weights and
$t : V(G) \times \dots \times V(G) \to \N$ instead of \(t^G\) for interpretations of weight terms,
when it is clear from the context that we do not talk about the terms or symbols themselves, 
but rather about their interpretation in a graph $G$.

\paragraph*{Controlling the Coefficients.}
We need some additional definitions to control the coefficients occurring in weight terms.
If all coefficients \(a_{ij}, a\) of a weight term are non-negative, we call \(t(X_1\dots X_k)\) a \emph{non-negative} weight term. 
If for all coefficients it holds that \(a_{ij}, a \in \mathbb{N}\) or \(a_{ij}, a \in \mathbb{Z}\),
we call it a \emph{natural} or \emph{integral} weight term, respectively.
Moreover, we say a number \(q \in \Q\) has \emph{granularity} \gran\ if it can be written as $\frac{a}{\text{\gran}}$ where \(a,\gran \in \N\).
If all coefficients \(a_{ij}, a\) are numbers with granularity \gran, we say the weight term has \emph{granularity \gran}.
Note that the interpretation of weight symbols on a graph \(G\) is derived from the weights of \(G\) which we defined to be non-negative integers.
Thus, a weight term is non-negative/natural/integral/has granularity \gran\ if and only it is guaranteed to evaluate to a 
non-negative/natural/integral/granularity-\gran\ number on every graph.

\paragraph*{LinECMSO.}
With these definitions in mind, we can now paraphrase the classical optimization result for \CMSO.

\begin{definition}[{\cite[Definition 10]{CourcelleMR2000}}]
A \emph{\LinEMSO-query} is a pair consisting of a \CMSO-formula  \(\phi(X_1 \dots X_\ell)\), which we refer to as the \emph{constraint} of the query, and an integral weight term \(\target(X_1\dots X_\ell)\), which we refer to as the \emph{optimization target} of the query.

An \emph{answer} to a \LinEMSO-query \((\phi,\target)\) on a graph \(G\) of matching signature 
is either a solution to \(\phi\) on \(G\) that maximizes the optimization target \(\target\) or the information that no solution exists.
\end{definition}
	
\begin{theorem}[{\cite[Theorem 4]{CourcelleMR2000}}]
    Given a \LinEMSO\(_1\)-query \((\phi,t)\), a graph \(G\) and a \(k\)-expression \(\chi\) of \(G\),
    one can find a solution to \((\phi,t)\) on \(G\) in time \(f(k, \phi, t)\cdot|\chi|\) for some computable function \(f\).
\end{theorem}

As discussed in the preliminaries, the result naturally extends to \LinEMSO\(_2\) on graphs of bounded treewidth by simply subdividing all edges.

\subsection{Adding Comparisons Between Weights}\label{sec:defExtendedLogic}

While \LinEMSO\ merely optimizes weight terms, we now proceed to integrate them deeply into the logic, specifically into its atoms.
If $t$ and $t'$ are general weight terms then $t \le t'$, $t < t'$, $t \ge t'$ and $t > t'$ are \emph{weight comparisons}. 
We define \MSOcomp\ as the extension of \CMSO\ that also allows weight comparisons as atomic building blocks (with semantics as defined previously).
For example, the \MSOcompone-sentence
\[
    \exists X~ w_1(X) \ge 100 \land w_2(X) \le 50
\]
expresses the existence of a set $W \subseteq V(G)$ with $\sum_{v \in W} w^G_1(v) \ge 100$ and $\sum_{v \in W} w^G_2(v) \le 50$.

To make some definitions easier, we require that \MSOcomp\ formulas are \emph{negation normalized}, in the sense that 
negation is never applied to a subformula containing a weight comparison.
Note that every formula can be negation normalized by pushing negations to lower levels of subformulas (via De Morgan) until each subformula to which negation is applied contains no weight comparison or is simply a weight comparison.
Negated weight comparisons can be replaced by equivalent non-negated weight comparisons (for example, \(\neg (t \le t')\) by \(t > t\)).

We extend all basic logical notations from \CMSO\ to \MSOcomp.
Most notably, the \emph{length} \(|\phi|\) of a \MSOcomp\ formula \(\phi\)
remains the number of symbols of \(\phi\), where each coefficient \(a, a_{ij}\) of a weight term is a \emph{single symbol}.
Thus, even for very large coefficients, their contribution to the length of the formula does not depend on their value but simply counts as \emph{one}.
We may similarly extend the notion of optimization queries from \LinEMSO\ to \MSOcomp, replacing \CMSO\ with the more powerful logic \MSOcomp.
Unfortunately, this is not very fruitful as already the ``simple'' problem of deciding whether \(G \models \phi\) for a fixed \MSOcomp formula $\phi$
is \NP-hard on trees of constant depth (see \Cref{thm:nphard-notnice}). 
This completely excludes an analogue of Courcelle's theorem for \MSOcomp-queries.

\paragraph*{Term Ranges.}
Note that \MSOcomp\ can easily express the \textsc{Subset Sum} problem as a constant-length formula using an edgeless graph associating every number with a vertex whose weight equals that number (see also Section~\ref{sec:problems}). 
The problem is polynomial time solvable for unary-encoded numbers but \NP-complete for binary encoded numbers.
To reflect this in the running times of our algorithms, we define \emph{term ranges}.
Given a graph $G$ and a set $\terms$ of weight terms, we say that the \emph{$\tau$-range of $G$} is the set of integers $\{0,\dots,N\}$,
where $N$ is minimal such that for all $t(\bar Z) \in \terms$
we have $t^G(V(G)^{|\bar Z|}) \leq N$.
We will use \(\terms\)-ranges later within our proofs, but can also immediately bound them by measures which are directly reflected by a graph and the weight terms of a formula.

\begin{lemma}\label{lem:boundTermsRange}
	Let $G$ be a graph and $\terms$ be a set of weight terms with at most $k$ free set variables.
	Let $\mu$ be the highest number that occurs as coefficient in any weight term of $\terms$ or as any weight in $G$.
	Then the $\tau$-range of $G$ is contained in $\{0,\dotsc,\mu + k^2 \mu^2 |V(G)|\}$.
\end{lemma}
\begin{proof}
	Let $t(\bar Z) \in \terms$ with weights $0 \le a,a_{ij} \le \mu$.
	Then
	\[
	t^G(V(G)^{|\bar Z|}) = a + \sum_{1 \le i,j \le |\bar Z|}a_{ij} \sum_{v \in V(G)} w^G_i(v) \le \mu + k^2 \mu |V(G)| \mu. \qedhere
	\]
\end{proof}

\subsection{Our Tractable Fragment}
\label{sec:bMSO}

We will define a fragment \blockMSO\ of \MSOcomp\ that restricts the interaction of universally quantified variables via weight comparisons
and is better suited for algorithmic use.
For this, we say a \emph{block formula} is a formula of the form $\forall \bar Y \phi(\bar X\bar Y)$
where 
\begin{itemize}
    \item $\phi$ is a \MSOcomp\ formula containing only non-negative weight terms,
	\item every weight term of $\phi$ except at most one contains only free variables from $\bar X$,
	\item the possible exceptional weight term occurs only on a single side of a single weight comparison and contains only free variables from $\bar X \bar Y$.
\end{itemize}
We define \blockMSO\ to be the set of all formulas that can be obtained from block formulas via iterative
disjunction, conjunction and existential quantification. 
More formally, each block formula is contained in \blockMSO, 
and if $\phi$, $\psi$ are contained in \blockMSO, then $(\psi\wedge \psi)$, $(\psi\vee \psi)$, $(\exists x \phi)$ and $(\exists X \phi)$ are, too. We remark that readers may find multiple examples of various \blockMSO\ formulas in Section~\ref{sec:problems}.

\begin{definition}
	A \emph{\blockMSO-query} is a pair consisting of a \blockMSO-formula \(\phi(X_1 \dots X_\ell)\) containing only natural weight comparisons which we refer to as the \emph{constraint} of the query, and an integral weight term \(\target(X_1  \dots X_\ell)\) (with possibly negative coefficients) which we refer to as the \emph{optimization target} of the query.

	An \emph{answer} to a \blockMSO-query \((\phi,\target)\) on a graph \(G\) of matching signature is a solution to \(\phi\) in \(G\) that maximizes the optimization target \(\target\) or the information that no solution exists.
\end{definition}

Note that it is no limitation that we focus on maximization only, as a minimization problem can be turned into a maximization problem by multiplying the target by \(-1\).

\subsection{Approximation}
\label{sub:approx}

The multiplication \(\alpha t\) of a weight term \(t\) and a real number \(\alpha \in \R\) is obtained by simply scaling all coefficients in \(t\) by \(\alpha\).
To be precise, if \(t = a + \sum_{1 \le i,j \le |\bar{X}|}a_{ij} w_i(X_i)\) then 
\(\alpha t := \alpha a + \sum_{1 \le i,j \le |\bar{X}|}\alpha a_{ij} w_i(X_i)\).

\paragraph*{Loosening and Tightening.}
    For a \MSOcomp-formula \(\phi(\bar X)\) and \(\alpha \in \mathbb{R}\), the \emph{\undersatformulatext} \(\undersatformula(\bar X)\) of \(\phi(\bar X)\) is given by replacing all weight comparisons in \(\phi(\bar X)\) of the form
    \begin{multicols}{2}
    \begin{itemize}
        \item \(t \leq t'\) by \(\frac{1}{\alpha} t \leq \alpha t'\),
        \item \(t < t'\) by \(\frac{1}{\alpha} t < \alpha t'\),
        \item \(t \geq t'\) by \(\alpha t \geq \frac{1}{\alpha} t'\), and
        \item \(t > t'\) by \(\alpha t > \frac{1}{\alpha} t'\).
    \end{itemize}
    \end{multicols}
    Similarly, the \emph{\oversatformulatext} \(\oversatformula(\bar X)\) of \(\phi(\bar X)\) is defined by replacing all weight comparisons in \(\phi(\bar X)\) of the form
    \begin{multicols}{2}
    \begin{itemize}
        \item \(t \leq t'\) by \(\alpha t \leq \frac{1}{\alpha} t'\),
        \item \(t < t'\) by \(\alpha t < \frac{1}{\alpha} t'\),
        \item \(t \geq t'\) by \(\frac{1}{\alpha} t \geq \alpha t'\), and
        \item \(t > t'\) by \(\frac{1}{\alpha} t > \alpha t'\).
    \end{itemize}
    \end{multicols}

Remember that we assumed all \MSOcomp\ formulas to be negation normalized in the sense that there are no negations in front of size comparisons.
This means that making a size comparison harder or easier to satisfy also makes the whole formula harder or easier to satisfy.
Thus, for \(\alpha \ge 1\) we have that \(\oversatformula(\bar X)\) is the strictest formula (i.e., the most difficult to satisfy) and
\(\undersatformula(\bar X)\) is the least strict formula (i.e., the easiest to satisfy) that can be attained from \(\phi(\bar X)\) by multiplying coefficients by factors between \(1/\alpha\) and \(\alpha\).

We say that \(G\) \emph{\(\alpha\)-undersatisfies} a \MSOcomp-formula \(\phi\) (written as \(G \umodels \phi\)) if \(G \models \undersatformula\).
Similarly, \(G\) \emph{\(\alpha\)-oversatisfies} \(\phi\) (written as \(G \omodels \phi\)) if \(G \models \oversatformula\).
Hence, for \(\alpha \ge 1\), it is easier to undersatisfy a formula than to oversatisfy it
and we get the chain of implications
\begin{equation}\label{eq:chain1}
G \omodels \phi \quad\Longrightarrow\quad G \models \phi \quad\Longrightarrow\quad G \umodels \phi.
\end{equation}
This chain collapses if and only if 
multiplying coefficients by factors between \(1/\alpha\) and \(\alpha\) does not change whether the formula holds on \(G\).

Previously, this concept has lead to approximation algorithms for decision problems. 
Dreier and Rossmanith~\cite{dreier2020approximate} showed that for a
certain counting extension of first order logic, evaluating formulas from
that logic exactly is \(\W\)[1]-hard on trees,
but there is an efficient approximation algorithm for sparse graphs that takes as input a graph $G$, a formula $\phi$ of that logic and an accuracy $\varepsilon > 0$
and returns either ``yes'', ``no'' or ``I don't know'', such that
\begin{itemize}
    \item if the algorithm returns ``yes'', then $G \models \phi$,
    \item if the algorithm returns ``no'', then $G \not\models \phi$,
    \item if the algorithm returns ``I don't know'', then \(G \models \phi_{1+\eps}\) but \(G \not\models \phi^{1+\eps}\).
\end{itemize}
Thus, the algorithm could only return an insufficient answer if the formula is ``at the verge'' of being true.

\paragraph*{Over- and Undersatisfied Maxima.}
Lifting these notions from the decision realm to the setting where we additionally optimize a target requires some care, but is worth the effort.
We first define the \emph{maximum} of a \MSOcomp-query \((\phi,\target)\) on a graph \(G\) as
\[\val(G, \phi, \target) := \max\left\{\target(W_1, \dotsc, W_\ell) ~\big|~ G \models \phi(W_1, \dotsc, W_\ell)\right\}.\]
The \emph{undersatisfied maximum} and \emph{oversatisfied maximum} are defined as 
\begin{align*}
    \textnormal{undersatisfy}(G, \phi, \target, \alpha)& := \max\left\{\target(W_1, \dotsc, W_\ell) ~\big|~ G \umodels[\alpha] \phi(W_1, \dotsc, W_\ell)\right\}, \\
    \textnormal{oversatisfy}(G, \phi, \target, \alpha)& := \max\left\{\target(W_1, \dotsc, W_\ell) ~\big|~ G \omodels[\alpha] \phi(W_1, \dotsc, W_\ell)\right\}.
\end{align*}

Intuitively, the 
oversatisfied maximum captures the highest value \(\target\) could attain after \(\alpha\)-tightening all terms (making them harder to satisfy),
while the 
undersatisfied maximum captures the highest value $\target$ could attain after \(\alpha\)-loosening all terms (making them easier to satisfy).
Thus, in accordance with (\ref{eq:chain1}), it holds for every \(\alpha \ge 1\) that
\[
\textnormal{oversatisfy}(G, \phi, \target, \alpha) \le
\val(G,\phi,\target) \le
\textnormal{undersatisfy}(G, \phi, \target, \alpha).
\]

\paragraph*{Approximate Answers.}
We finally arrive at our notion of approximate answers to optimization queries that forms the heart of our contribution.
We first give the rather involved definition and then break it down step by step.

\begin{definition}\label{def:approxAnswer}
    For \(\alpha \ge 1\),
    an \emph{\(\alpha\)-approximate answer} to a \blockMSO-query \((\phi,\target)\) on a graph \(G\) 
    consists of two values $\val^-, \val^+ \in \N \cup \{-\infty\}$ such that
    $$
    \textnormal{oversatisfy}(G, \phi, \target, \alpha) \le
    \val^- \le
    \val(G,\phi,\target) \le
    \val^+ \le
    \textnormal{undersatisfy}(G, \phi, \target, \alpha).
    $$
    If \(\val^- \neq -\infty\), then the answer additionally contains a ``conservative'' witness
    $\bar W^-$ with
    \[
        G \models \phi(\bar W^-) \text{\quad and \quad} \target(\bar{W}^-) = \val^- \ge \textnormal{oversatisfy}(G, \phi, \target, \alpha).
    \]
    If \(\val^+ \neq -\infty\), then the answer additionally contains an ``eager'' witness $\bar W^+$ with 
    \[ 
    G \umodels[\alpha] \phi(\bar W^+) \text{\quad and \quad} \target(\bar{W}^+) = \val^+ \ge \val(G,\phi,\target).~~~~~~~~~~~~~
    \]
\end{definition}
In our algorithms, we strive for \((1+\eps)\)-approximate answers for arbitrarily small \(\eps > 0\).
Observe that a \((1+\eps)\)-approximate answer to a query contains an upper
and lower bound \(\val^-\) and \(\val^+\) to the best attainable value
\(\val(G,\phi,\target)\).
The quality of these bounds depends on the solutions to \(\phi\) on \(G\) that
are added or removed by \((1+\eps)\)-loosening or tightening the weight
comparisons.
Thus, the only time when the upper and lower bounds do not agree is when
\(\phi\) is ``at the edge'' of being true for some very good solution candidates,
that is, multiplying coefficients by factors between \(1/(1+\eps)\) and
\(1+\eps\) either removes optimal solutions or adds even better solutions.

While we cannot return an optimal answer to the query, we can return a ``\((1+\eps)\)-conservative'' answer \(\bar W^-\) that certainly satisfies \(\phi\),
and is still better than any solution to the \((1+\eps)\)-tightening of \(\phi\).
However, it may be that \(\phi\) itself has a solution while the \((1+\eps)\)-tightening does not, in which case no ``conservative'' answer is returned.
Nevertheless, we still return an additional ``\((1+\eps)\)-eager'' answer \(\bar W^+\) that is at least as good
as the optimal solution to \(\phi\). 
While \(\bar W^+\) may not satisfy \(\phi\), it is still guaranteed to satisfy the \((1+\eps)\)-loosening of \(\phi\).
To reiterate what we observed earlier, 
\begin{itemize}
    \item \(\bar W^-\) satisfies the constraint and is at least as good as the optimum of the tightened constraint,
    \item \(\bar W^+\) satisfies the loosened constraint and is at least as good as the optimum.
\end{itemize}

Furthermore, note that for any \(0 < \eps' < \eps\), a
\((1+\eps')\)-approximate answer is also a \((1+\eps)\)-approximate answer, and any
\(1\)-approximate answer is simply an answer.
When applying our results to specific problems, the following observation can help us state the approximation guarantees in a cleaner way.

\begin{observation}\label{obs:nicerGuarantees}
    Let \(0 < \eps \leq 0.5\) and \(\phi\) be a \blockMSO-constraint.
    \begin{itemize}
        \item
            If a tuple \(\bar W\) satisfies 
            the constraint obtained from \(\phi\) by replacing every inequality 
                \begin{multicols}{2}
    \begin{itemize}
        \item \(t \leq t'\) by \( t \leq (1-\eps)\cdot t'\),
        \item \(t < t'\) by \( t < (1-\eps)\cdot t'\),
        \item \(t \geq t'\) by \( t \geq (1+\eps)\cdot t'\), and
        \item \(t > t'\) by \(t > (1+\eps)\cdot t'\),
    \end{itemize}
    \end{multicols}
            then \(\bar W\) also satisfies the \((1+\eps/3)\)-tightened constraint \(\phi^{1+\eps/3}\).
        \item
            If a tuple \(\bar W\) satisfies the \((1+\eps/3)\)-loosened constraint \(\phi_{1+\eps/3}\), then it also satisfies
            the constraint obtained from \(\phi\) by replacing every inequality
                \begin{multicols}{2}
            \begin{itemize}
                \item \(t \leq t'\) by \( t \leq (1+\eps)\cdot t'\),
                \item \(t < t'\) by \( t < (1+\eps)\cdot t'\),
                \item \(t \geq t'\) by \( t \geq (1-\eps)\cdot t'\), and
                \item \(t > t'\) by \(t > (1-\eps)\cdot t'\).
            \end{itemize}
            \end{multicols}
    \end{itemize}
\end{observation}
\begin{proof}
    The \((1+\eps/3)\)-loosening of an inequality \(t \le t'\)
    gives \(t/(1+\eps/3) \le t' \cdot (1+\eps/3)\), and is thus equivalent to \(t \le t' \cdot (1+\eps/3)^2 \).
    Similarly, the \((1+\eps/3)\)-tightening of an inequality \(t \le t'\)
    gives \(t \cdot (1+\eps/3) \le t'/(1+\eps/3)\), which is equivalent to \(t \le t'/(1+\eps/3)^2\).
The claims involving weight comparisons that were altered by a factor of $(1+\eps)$ follow by the fact that for every \(0 < \eps \leq 0.5\),
    \[
        (1+\eps/3)^2 \le 1+\eps.
    \]
Similarly, claims involving weight comparisons that were altered by a factor of $(1-\eps)$ follow 
(using the previous statement and \(a \le b\) iff \(\frac{1}{a} \ge \frac{1}{b}\) for all \(a,b\)) by the inequality
    \[
        \frac{1}{(1+\eps/3)^2}\geq \frac{1}{1+\eps} \ge 1-\eps.\qedhere
    \]
\end{proof}
To elaborate on the intended use of Observation~\ref{obs:nicerGuarantees}, let us imagine we are given a \blockMSO-query $(\phi,\target)$. A $(1+\eps)$-approximate answer to $(\phi,\target)$ provides guarantees in terms of the $(1+\eps)$-loosened or $(1+\eps)$-tightened weight comparisons; in particular, each weight comparison in $\phi$ of the form, e.g., $t\leq t'$, can be thought of as becoming a $(1+\eps)$-loosened weight comparison $t\leq t'\cdot (1+\eps)^2$ or a $(1+\eps)$-tightened weight comparison $(1+\eps)^2 \cdot t \leq t'$. 
However, what one would typically expect is to be able to state the provided guarantees in a simpler form such as $t\leq t'\cdot (1+\eps')$ or $t\leq t'\cdot (1-\eps')$ for some $\eps'$.
Observation~\ref{obs:nicerGuarantees} allows us to convert a requested guarantee $\eps'$ in the latter ``nice'' forms to a value of $\eps$ such that a $(1+\eps)$-approximate answer to $(\phi,\target)$ is guaranteed to satisfy the ``nice'' guarantees. In particular, we can simply set $\eps:=\eps'/3$.
	
	\section{Our Algorithmic Meta-Theorems}
	\label{sec:metat}
    We are finally ready to formally state the main results of this paper in the form of two theorems: one for cliquewidth, and one for treewidth. The proofs of these statements form the culmination of Sections~\ref{sec:exactFV}--\ref{sec:final} and are presented in Subsection~\ref{sec:metaproofs}.
    We remark that the first two theorems each provide two running time bounds: 
    The former bound is useful for larger binary-encoded numbers, and the latter yields the kind of running time one would typically see for smaller numbers that are polynomial in the size of the input. 
    The latter running time is a simple corollary of the former one (using \Cref{obs:runtime}).

	\begin{restatable}[Approximation of \blockMSOone]{theorem}{mainbase}
		\label{thm:mainapproxcw}
		Given a \blockMSOone-query \((\phi,\target)\), an accuracy \(0 < \eps \leq 0.5\) and
		a graph $G$ with matching signature,
we can compute a \((1+\eps)\)-approximate answer to \((\phi,\target)\) on \(G\) in time at most
		\begin{align*}
		&\left(\frac{\log(\MM)}{\eps}\right)^{f_1(|\phi|,\cw(G))} \cdot |V(G)|^2, \quad\quad
        \end{align*}
        \text{and in particular in time at most}
		\begin{align*}
		&\left(\frac{1}{\eps}\right)^{f_2(|\phi|,\cw(G))} \cdot |V(G)|^2 \cdot \MM^{0.001}
\end{align*}
		where $f_1$, $f_2$ are computable functions and 
		$\MM$ is two plus the highest number that occurs in any weight term of $\phi$ or as any weight in $G$.
	\end{restatable}

We remark that the constant \(0.001\) in the above theorem can be replaced with an arbitrary \(\delta > 0\) (while updating the function \(f_2\) accordingly).
The same holds for the replacement of constants $1.001$ by $1+\delta$ in the following meta-theorem for treewidth.

	\begin{restatable}[Approximation of \blockMSOtwo]{theorem}{mainbasetw}
		\label{thm:mainapproxtw}
		Given a \blockMSOtwo-query \((\phi,\target)\), an accuracy \(0 < \eps \leq 0.5\) and
		a graph $G$ with matching signature,
we can compute a \((1+\eps)\)-approximate answer to \((\phi,\target)\) on \(G\) in time at most
		\begin{align*}
		&\left(\frac{\log(\MM)}{\eps}\right)^{f_1(|\phi|,\tw(G))}\cdot |V(G)|^{1.001}, \quad 
        \end{align*}
        \text{and in particular in time at most}
		\begin{align*}
		&\left(\frac{1}{\eps}\right)^{f_2(|\phi|,\tw(G))}\cdot \MM^{0.001} \cdot |V(G)|^{1.001}
\end{align*}
		where $f_1$, $f_2$ are computable functions and 
		$\MM$ is two plus the highest number that occurs in any weight term of $\phi$ or as any weight in $G$.
	\end{restatable}

	We remark that by choosing the accuracy \(\eps\) to be large enough, the above two theorems allow us to obtain (\(1\)-approximate) answers to \blockMSOone- and \blockMSOtwo-queries. In particular, if $\eps>(M+|V(G)|+|\phi|)^{\bigoh(1)}$, then every \((1+\eps)\)-approximate answer actually contains an exact answer.
	Because the running times in our above meta-theorems scale polynomially in \(\frac{1}{\eps}\), this yields \XP-time exact variants of them.

\begin{restatable}[Evaluating Queries Exactly]{theorem}{exact}
\label{thm:xp}
        Given a \blockMSOone-query (or \blockMSOtwo-query) \((\phi,\target)\) and a graph \(G\) with matching signature and cliquewidth (or treewidth) at most \(k\), we can compute an answer to \((\phi,\target)\) on \(G\) in time
        \[(\MM+|V(G)|)^{f(|\phi|,k)}\]
        where \(f\) is some computable function and 
        $\MM$ is the highest number that occurs in any weight term of $\phi$ or as any weight in $G$.
\end{restatable}
	
	\section{Implications for Specific Problems}
	\label{sec:problems}
	In this section, we provide examples of concrete problems we can encode in our logic and the resulting algorithms our meta-theorem implies for each of them.
	Many (but not all) of these have been obtained previously by specifically targeted algorithms, while here we show that all of them follow by a simple application of our meta-theorem. 

    Let us start by discussing once again how the notion of approximation studied in this paper differs from other notions. Consider the \textsc{Subset Sum} problem---or, more precisely, the two classical variants of \textsc{Subset Sum} that are sometimes used interchangeably in the literature:
    
    \noindent    
\begin{center}
\begin{boxedminipage}{0.98 \columnwidth}
\textsc{Subset Sum} (Optimization Version)\\[5pt]
\begin{tabular}{l p{0.90 \columnwidth}}
Input: & A multiset $\mathcal{S}$ of integers and a target integer $T$.\\
Task: & Find a subset $S\subseteq \mathcal{S}$ such that $\sum_{s\in S}s\leq T$ which maximizes $\sum_{s\in S}s$.
\end{tabular}
\end{boxedminipage}
\end{center}

\noindent    
\begin{center}
\begin{boxedminipage}{0.98 \columnwidth}
\textsc{Subset Sum} (Decision Version)\\[5pt]
\begin{tabular}{l p{0.70 \columnwidth}}
Input: & A multiset $\mathcal{S}=$ of integers and a target integer $T$.\\
Question: & Does there exist a subset $S\subseteq \mathcal{S}$ that sums up precisely to $T$?
\end{tabular}
\end{boxedminipage}
\end{center}

Accordingly, given a small accuracy \(0.5 \ge \eps > 0\) and using the language from the introduction, we may consider two different notions of approximation for \textsc{Subset Sum}.

    \begin{itemize}
        \item Optimization target approximation:
            We say a \emph{solution} is a subset of \(\mathcal S\) that sums up to at most \(T\),
            and we are asked to find a solution whose sum is at least \(1-\eps\) times the largest value of a solution.
        \item Constraint approximation: 
            Here, we are asked to either give a subset of \(\mathcal S\) summing up to a number between \((1-\eps)T\) and \((1+\eps)T\)
            or correctly state that there is no subset that sums up precisely to \(T\). 
    \end{itemize}

For \textsc{Subset Sum} or, e.g., \textsc{Knapsack}, optimization target approximation 
    is the most widely studied notion and in fact provides stronger guarantees. However, our framework can also capture many problems that do not have natural formulations as optimization problems, and there constraint approximation is more natural.
    For example, consider the \textsc{Equitable Coloring} problem: given a graph $G$ and integer $k$, is there a proper \(k\)-coloring of a given graph with all color classes having equal size \(\pm 1\)?
    Here, the reasonable way to approximate the problem is to relax the constraint from ``equal size \(\pm 1\)'' to ``roughly equal size''.
    \begin{itemize}
        \item Constraint approximation: 
            Either return a \(k\)-coloring where the ratio between the sizes of all color classes is at most \(1+\eps\),
            or determine that there is no \(k\)-coloring where all color classes have equal size \(\pm 1\). 
    \end{itemize}
    We organize this section based on the types of considered problems, starting with fundamental number-based problems and working our way up to graph problems parameterized by treewidth and cliquewidth, comparing the exact results and constraint approximation results derived via our meta-theorem to previous work.
	
	\subsection{Basic Number-Based Problems}
	Our framework can capture a set of classical number-based problems that are not typically defined on graphs.
    We encode them using an edgeless graph that trivially has cliquewidth 1 and formulas that typically have constant length.

\paragraph{Subset Sum.}
    Given a multiset $\mathcal{S}=\{s_1,\dots,s_n\}$ of integers and a target integer $T$, 
	to capture \textsc{Subset Sum} we simply construct an edgeless graph in which we introduce a vertex \(v_i\) with weight \(w(v_i) := s_i\) for each \(i \in [n]\). The problem can then be encoded by the \blockMSOone-constraint
    \[
        \phi_\textsc{SubS} (X) := w(X) \leq T \land w(X) \geq T,
    \]
    with an empty optimization target $\target=0$.
	Obviously \(\cw(G),\tw(G),|\phi_\textsc{SubS}| \le \mathcal{O}(1)\).
	
	For any \(\eps > 0\), the \((1+ \eps)\)-tightening of \(\phi_\textsc{SubS}(X)\) is the unsatisfiable formula \((1 + \eps) w(X) \leq \frac{1}{1 + \eps} T \land \frac{1}{1+\eps} w(X) \geq (1 + \eps) T\), which means that conservative solutions do not exist and are not meaningful in this context. 
    However, the \((1+ \eps)\)-loosening \(\phi_\textsc{SubS}(X)\) amounts to satisfying \( \frac{1}{1+\eps} w(X) \leq (1 + \eps) T \land (1 + \eps) w(X) \geq \frac{1}{1 + \eps} T\).
    Remember that a \((1+\eps)\)-eager solution still satisfies the \((1+\eps)\)-loosening of a given constraint.
    By rescaling $\eps$ according to Observation~\ref{obs:nicerGuarantees}, we can also guarantee that an eager solution satisfies \( w(X) \leq (1 + \eps) T \land w(X) \geq (1 - \eps)T \).
    Hence, we can either give a subset of \(\mathcal S\) summing up to a number between \((1-\eps)T\) and \((1+\eps)T\)
    or correctly state that there is no subset summing up to precisely \(T\). 
    By assuming w.l.o.g. that all numbers in \(\mathcal S\) are upper bounded by \(T\),
    \Cref{thm:mainapproxcw} gives a classical FPTAS for \textsc{Subset Sum}, running in time
    \[
        \left(\frac{\log(T) \cdot n}{\eps}\right)^{\bigoh(1)}.
    \]
    This result as well as its variant for optimization are of course well known~\cite{KellererMPS03}.
    Moreover, we can apply \Cref{thm:xp} to reprove the well-known fact that \textsc{Subset Sum} is polynomial-time solvable
    when all numbers are encoded in unary.

\paragraph{Knapsack.}	
	As a second example, consider the \textsc{Knapsack} problem: given a set $V$ of items, a capacity $T \in \N$, a score function $w_1:V\rightarrow \Nat$ and a size function $w_2:V\rightarrow \Nat$, find a subset $S\subseteq V$ maximizing $\sum_{s\in S}w_1(s)$, subject to satisfying $\sum_{s\in S}w_2(s)\leq T$. The graph used to capture \textsc{Knapsack} will be just as simple as for \textsc{Subset Sum}: for each item in $V$ we simply create a corresponding isolated vertex, and we equip the vertex set of the graph with the weight functions $w_1$ and $w_2$. We use the simple \blockMSOone-formula
    \[
        \phi_\textsc{Knap}(X) :=  w_2(X) \leq T
    \]
    and the optimization target $\target(X)=w_1(X)$.

    Here, a \((1+\eps)\)-eager solution results in a set of items guaranteed to achieve a score that is at least as good as the optimal solution to the original \textsc{Knapsack} instance, albeit it only satisfies the \((1+\eps)\)-loosened constraint $\frac{1}{1+\eps} w_2(X) \leq (1+\eps)T$ and hence may ``overshoot'' the knapsack by using up to $T \cdot (1+\eps)^2$ capacity; Observation~\ref{obs:nicerGuarantees} allows us to reduce this capacity upper bound to $T \cdot (1+\eps)$. 
	On the other hand, a \((1+\eps)\)-conservative solution to the same query results in a set of items that use capacity at most $T$ and hence ``fit'' in the original knapsack, but may not archive the optimal score. By again polishing the capacity bound via Observation~\ref{obs:nicerGuarantees}, we can ensure the returned score is at least as good as any solution satisfying the reduced capacity bound $(1-\eps)T$.
    In line with the discussion at the beginning of the section, we remark that this differs from the typical approximation setting for \textsc{Knapsack} where one seeks a solution fitting into the knapsack but whose score is within a factor \(1+\eps\) of the optimal score.
	
	Theorem~\ref{thm:mainapproxcw} gives us an FPTAS to compute both a $(1+\eps)$-eager and a $(1+\eps)$-conservative solution for an instance $\mathcal{I}$ of \textsc{Knapsack} in time \((\frac{\mathcal{I}}{\eps})^{\bigoh(1)},\) even when all numbers are encoded in binary.

    We remark that by adding edges to the graph-encoding described above, one may also encode additional graph-like dependencies in problem instances. 
    For example, one could only allow to place an element into the knapsack if its dependencies are also present, or one could add conflicts which forbid certain elements to be selected simultaneously. Naturally, the same also applies to the other problems described in this subsection, and the running times would then also depend on the cliquewidth of the graph representation.
	
\paragraph{Multidimensional Subset Sum.}
As an example that is slightly more involved than the previous two, let us consider the following variant of \textsc{Multidimensional Subset Sum}: given a  multiset $\mathcal{S}$ of $k$-dimensional vectors over $\mathbb{N}$ and a $k$-dimensional target vector $T$, all encoded in unary, decide whether there exists a subset $S\subseteq \mathcal{S}$ such that $\sum_{s\in S}s=T$. Like many others~\cite{GanianOR21,GanianOR19}, this variant of \textsc{Multidimensional Subset Sum} is \W[1]-hard when parameterized by $k$~\cite{GanianKO21}.

Similarly to normal \textsc{Subset Sum}, we create a vertex $v_i$ for each vector $s_i \in S$ and introduce $k$ weight functions $w_1,\dots,w_k$ such that for each $j\in [k]$ and vertex $v_i$, $w_j(v_i)=s_i[j]$. Our \blockMSO$_1$-query consists of the constraint
\[
    \phi_\textsc{MSubS}(X) := \bigwedge_{j\in [k]} \big( T[j]\geq w_j(X) \land T[j]\leq w_j(X)\big)
\]
and the empty optimization target $\target(X)=0$.

With the same arguments as for (one-dimensional) \textsc{Subset Sum},
our framework approximates the problem by either giving a subset of \(\mathcal S\) summing up to a \(k\)-vector between \((1-\eps)T\) and \((1+\eps)T\)
or correctly stating that there is no subset summing up to precisely \(T\). 
Given an instance \(\mathcal I\) of \textsc{Multidimensional Subset Sum}, with all numbers encoded in unary,
\Cref{thm:mainapproxcw} and \Cref{thm:xp} give us an FPTAS and exact algorithm whose running times, respectively, are 
\[
   \left(\frac{1}{\eps}\right)^{f(k)} \cdot |\mathcal{I}|^{\bigoh(1)}
   \quad
   \quad
   \quad
   \textnormal{and}
   \quad
   \quad
   \quad
   \quad
   \quad
   |\mathcal{I}|^{f(k)},
\]
for some computable function \(f\).
Notice that unlike in the previous two examples, here the length of the formula $\phi_\textsc{MSubS}(X)$ depends on $k$, giving both algorithms an exponential dependence on \(k\). That is to be expected, since without parameterizing by \(k\) this problem is known to be \NP-hard,
even when numbers are provided in unary~\cite{GanianKO21}.
We remark that other known variants of \textsc{Multidimensional Subset Sum} can be captured by our formalism in a similar way~\cite{GanianOR21,GanianOR19}.

	\subsection{Equitability}
	\label{sub:equit}
	A well-studied group of problems which are typically not fixed-parameter tractable when parameterized by treewidth and cliquewidth ask for a partitioning of the vertices of a graph that is \emph{equitable}, i.e., where each part has the same size ($\pm 1$). Prominent examples of such problems include \textsc{Equitable Coloring}~\cite{BodlaenderF05,LuoSSY10} and \textsc{Equitable Connected Partition}~\cite{EncisoFGKRS09,AnFKX20}, albeit several other variants have been studied as well~\cite{EsperetLM15,KaulMW22}.
	
	For \textsc{Equitable Coloring}, we are given a graph $G$ and an integer $k$ and are asked whether there exists a proper $k$-coloring of $G$ such that the difference between the size of each pair of color classes is at most $1$. The definition of \textsc{Equitable Connected Partition} is analogous, with the sole distinction being that we seek a partitioning of $G$ into $k$ connected components instead of a $k$-coloring. The two problems are known to be \W[1]-hard when parameterized by treewidth plus $k$~\cite{BodlaenderF05,EncisoFGKRS09}.
	
	Both problems can straightforwardly be encoded by \blockMSOone-formulas \(\phi_\textsc{EC}\) and \(\phi_\textsc{ECP}\)
    with free set variables $X_1,\dots,X_k$, defined as a conjunction over the following constraints:
	\begin{enumerate}[label=(\roman*)]
		\item \label{item:equitablesizes} \(k^2\) weight comparisons \(|X_i| \leq |X_j| + 1\) for all combinations of \(i,j \in [k]\).
		
	\end{enumerate}
	Depending on the problem we consider,
	we also conjoin an \MSO$_1$-formula which ensures that \(X_1, \dotsc, X_k\) partition \(V(G)\) and each \(X_i\) with \(i \in [k]\) is
	\begin{enumerate}[label=(ii)]
		\item a set of independent vertices in \(G\) (for \textsc{Equitable Coloring}), or \\
		      connected in \(G\) (for \textsc{Equitable Connected Partition}).
	\end{enumerate}

	By applying \Cref{thm:xp} we immediately obtain the following.
		\begin{corollary}
		\textsc{Equitable Coloring} and \textsc{Equitable Connected Partition} are in \XP\ parameterized by $k$ plus the cliquewidth of the input graph.
	\end{corollary}
	Naturally, we can combine the equitability of a partition of the vertex set with other \MSO-expressible conditions for the partition sets. For example, it is easy to encode \textsc{Equitable Partition into Induced Forests}~\cite{EsperetLM15} using a \blockMSOone-query.

    \paragraph*{Approximation.}
Turning to approximation, for a fixed \(\eps > 0\), \((1+\eps)\)-loosening \(\phi_\textsc{EC}\) or \(\phi_\textsc{ECP}\) amounts to replacing each constraint \(|X_i| \leq |X_j| + 1\) from \ref{item:equitablesizes} by \(\frac{1}{1 + \eps}|X_i| \leq (1 + \eps)(|X_j| + 1)\). Similarly, \((1+\eps)\)-tightening \(\phi_\textsc{EC}\) or \(\phi_\textsc{ECP}\) amounts to replacing each constraint \(|X_i| \leq |X_j| + 1\) from \Cref{item:equitablesizes} by \((1 + \eps)|X_i| \leq \frac{1}{1 + \eps}(|X_j| + 1)\). 
The tightened formula cannot be satisfied for any but the smallest values of $\eps$ (meaning that we obtain no guarantee for a conservative solution). 
But the loosened formula can be tweaked into the nicer form of \(|X_i| \leq (1+\eps)|X_j|\)
by rescaling \(\eps\) as in \Cref{obs:nicerGuarantees} and once again\footnote{
If \(a \le (b+1)(1+\eps/2)\) and \(b\) is sufficiently large, then also \(a \le b (1+\eps)\).
}
to get rid of the ``\(+1\)''.
This allows us to invoke Theorem~\ref{thm:mainapproxcw} and obtain the following scheme of fixed-parameter approximation algorithms for the corresponding equitable problems.

Even though we do not explicitly state it from now on, all these algorithms are FPTAS in the sense that \(\frac{1}{\eps}\) only appears polynomially,
albeit the exponent of this polynomial depends on the cliquewidth and the formula. 
	
	\begin{corollary}
	\label{cor:equitable}
There is a fixed-parameter algorithm that takes as input a graph $G$ of cliquewidth $\ell$, an integer $k$ and $0<\eps\leq 0.5$, is parameterized by $\ell+k+\frac{1}{\eps}$, and
\begin{enumerate}
\item either produces a $k$-coloring such that the ratio of the sizes of any two color classes is at most \(1+\eps\),
		or correctly decides that no equitable $k$-coloring exists;
\item either produces a partitioning into \(k\) connected components such that the ratio of the sizes of any two components is at most \(1+\eps\),
		or correctly decides that no equitable partitioning into \(k\) connected components exists.
\end{enumerate}
	\end{corollary}

We remark that the result matches the fixed-parameter approximation algorithm for \textsc{Equitable Coloring} by Lampis~\cite{Lampis14} (with the same parameterization, albeit a slower running time, as can be expected due to the use of \MSO\ logic). 

By slightly tweaking the thresholds in the formulas, we can even obtain stronger guarantees.
We state these guarantees only for \textsc{Equitable Coloring}, but they can be derived similarly for other problems mentioned in this section.
\begin{corollary}
There is a fixed-parameter algorithm that takes as input a graph $G$ of cliquewidth $\ell$, an integer $k$ and $0<\eps\leq 0.5$, is parameterized by $\ell+k+\frac{1}{\eps}$, and
either produces a $k$-coloring such that the ratio of the sizes of any two color classes is at most \(1+2\eps\),
or correctly decides that there is no $k$-coloring such that the ratio of the sizes of any two color classes is at most~\(1+\eps\).
\end{corollary}
	
    \paragraph{Running Time Independent of Number of Parts.}
	For \textsc{Equitable Connected Partition}, we can also avoid the involvement of \(k\) as a parameter at the cost of using an edge set variable, placing us in \blockMSOtwo. Hence, we only allow treewidth rather than cliquewidth as the width parameter.
	We achieve this via a formula that encodes the existence of an edge set whose removal results in all connected components having size \(|V(G)|/k \pm 1\), a property which can be checked by universally quantifying over all sets. More precisely, we can capture this via the following \blockMSOtwo\ formula (where the statements marked via quotation marks can be expressed via standard \MSO$_2$ formulas): 
	\begin{align*}
	\phi_\textsc{ECP-2}(X):={}&{}\Bigl(\forall Y: (\text{``$G[Y]$ is not a connected component of $G-X$''}) \lor |Y|\geq \left\lfloor\frac{n}{k}\right\rfloor\Bigr) \\	
{}\land {} &\Bigl(\forall Y: (\text{``G[Y] is not a connected component of $G-X$''}) \lor 
	|Y|\leq \left\lceil\frac{n}{k}\right\rceil\Bigr)\\
{}\land{} & \Bigl(\forall Y: (\text{``$Y$ does not contain precisely one vertex from each} \\
                      & \quad\quad\quad~~ \text{connected component of $G-X$''}) \lor {} |Y|\leq k\Bigr)\\
{}\land{} & \Bigl(\forall Y: (\text{``$Y$ does not contain precisely one vertex from each} \\
                      & \quad\quad\quad~~ \text{connected component of $G-X$''}) \lor {} |Y|\geq k\Bigr),
	\end{align*}
where $X$ is an edge set variable, $Y$ is a vertex set variable and $n=|V(G)|$.
	This allows us to, e.g., obtain the known \XP-tractability of this problem when parameterized by treewidth alone~\cite{EncisoFGKRS09}.
	\begin{corollary}
		\textsc{Equitable Connected Partition} is in \XP\ parameterized by the treewidth of the input graph.
	\end{corollary}

    In the related \textsc{Balanced Partitioning} problem~\cite{BevernFSS15}, we are given a graph $G$ along with integers $d$ and $k$, and are asked to find a set $X$ of at most $d$ edges such that $G-X$ consists of $k$ equitable connected components. With the same argument as above, this problem can also be captured via a \blockMSOtwo-query of constant length.
    Of course, \Cref{thm:mainapproxtw} also yields fixed-parameter algorithms for approximating these problems.
	
	\subsection{Bounded Degree Vertex Deletion}
	In the \textsc{Bounded Degree Vertex Deletion} problem, we are given a graph $G$ and an integer $p$ and are asked to find a minimum-size vertex set $D$ such that $G-D$ has degree at most $p$; for simplicity, let us call such a set a $p$-\emph{BDVD set}. The problem is known to be \W[1]-hard but \XP\ when parameterized by treewidth~\cite{DessmarkJL93,BetzlerBNU12,GanianKO21} or cliquewidth~\cite{Lampis14}. There are also strong lower bounds for the approximability of the problem on general graphs~\cite{DemaineGKLLSVP19}.
	
	To capture \textsc{Bounded Degree Vertex Deletion}, we can simply work on the uncolored input graph $G$. The optimization target will be $\target(X)=-|X|$ (allowing us to minimize the size of $X$).
Let a vertex $v\in A$ for some $A\subseteq V(G)$ be \emph{$A$-universal} if it is adjacent to every vertex in $A\setminus \{v\}$. We then capture \textsc{Bounded Degree Vertex Deletion} via the \blockMSOone-formula
\[
    \phi_{\textsc{BDVD}}(X):=\forall A: \big(\neg (\text{``}X\cap A=\emptyset\text{'' } \wedge \text{ ``$A$ contains an $A$-universal vertex''})\big) \vee |A|\leq p+1.
\]
\begin{observation}
\label{obs:BDVDcorrect}
Every $p$-BDVD set $D$ satisfies $\phi_{\textsc{CVC}}(D)$, and at the same time every interpretation of $X$ satisfying $\phi_{\textsc{CVC}}(X)$ is a $p$-BDVD set.
\end{observation}

\begin{proof}
Consider a $p$-BDVD set $D$, and observe that the formula $\phi_{\textsc{BDVD}}(D)$ is automatically satisfied for every choice of $A$ unless $A$ contains an $A$-universal vertex. Since $G-D$ has degree at most $p$, if $A$ contains an $A$-universal vertex then its size must be at most $p+1$. 

On the other hand, consider a set $D$ such that $G\models \phi_{\textsc{CVC}}(D)$. If a vertex $v\in G-D$ were to have degree at least $p+1$, then $v$ would be an $A$-universal vertex for the set $A$ containing $v$ and all neighbors of $v$ in $G-D$. Such a set $A$ would be disjoint from $D$, would contain an $A$-universal vertex, and would have size at least $p+2$. In other words, the existence of such a vertex $v$ would contradict $G\models \phi_{\textsc{BDVD}}(D)$. 
This means that the maximum degree in $G-D$ is at most $p$, and in particular $D$ is a $p$-BDVD set of $G$.
\end{proof}

At this point, applying \Cref{thm:xp} yields the known~\cite{Lampis14} \XP-algorithm for the problem when parameterized by cliquewidth. 
\begin{corollary}
\textsc{Bounded Degree Vertex Deletion} is in \XP\ parameterized by the cliquewidth of the input graph.
\end{corollary}

Towards the implications of our machinery for approximation, let us again consider what it means to loosen and tighten $\phi_{\textsc{BDVD}}(D)$ for some fixed $\eps > 0$, directly applying Observation~\ref{obs:nicerGuarantees}. A solution for the $(1+\eps)$-tightening of $\phi_{\textsc{BDVD}}(X)$ amounts to a $(1-\eps)p$-BDVD set, while analogously a solution to the $(1+\eps)$-loosening corresponds to a $(1+\eps)p$-BDVD set.

Hence, the conservative solution $W^-$ for the value $\val^-$ that forms the first part of an answer to the approximate \blockMSOone-query is guaranteed to be a $p$-BDVD set, and is no larger than the smallest $(1-\eps)p$-BDVD set (i.e., we are guaranteed to be safely within the bound $p$, but could overshoot in terms of size). At the same time, the eager solution $W^+$ for the value $\val^+$ forming the second part of the answer is guaranteed to be a $(1+\eps)p$-BDVD set and is no larger than the smallest $p$-BDVD set (here we have strong guarantees in terms of size, but could end up with some vertices of slightly larger degree). In summary, we obtain the following.

	\begin{corollary}
	\label{cor:BDVD}
There is a fixed-parameter algorithm which takes as input a graph $G$ of cliquewidth $\ell$, an integer $p$ and $0< \eps\leq 0.5$, is parameterized by $\ell+\frac{1}{\eps}$, and outputs
		\begin{enumerate}
\item a $p$-BDVD set in $G$ that is no larger than the smallest $(1-\eps)p$-BDVD set, and
\item a $(1+\eps)p$-BDVD set in $G$ that is no larger than the smallest $p$-BDVD set.	
		\end{enumerate}
	\end{corollary}

We again remark that the result matches the fixed-parameter approximation algorithm by Lampis~\cite{Lampis14} (with the same parameterization, albeit a slower running time). 

	\subsection{Capacitated Problems}
	
	Several generalizations of classical graph problems that include vertex capacities can also be captured by our framework. As the two typical examples here, let us consider \textsc{Capacitated Dominating Set} (\textsc{CDS}) and \textsc{Capacitated Vertex Cover} (\textsc{CVC}). 
    Both of these problems, when parameterized by treewidth, are \W[1]-hard~\cite{DomLSV08} but are \XP and lend themselves to \FPT-approximation~\cite{Lampis14}. 
    These problems will further exemplify why weight terms involving (multiple) universally quantified variables are useful.
	
The input for the optimization variants of \textsc{CDS} and \textsc{CVC} consists of a graph $G$ and a capacity function $c:V(G)\rightarrow \Nat$ which we assume to be encoded in unary. For \textsc{CDS} we ask for a minimum-size dominating set $X$ with the following property: there is a domination assignment from each vertex $v\not \in X$ to a neighbor in $X$ such that for each $x\in X$, the number of vertices assigned to it is upper-bounded by $c(x)$. Similarly, for \textsc{CVC} we ask for a minimum-size vertex cover $C$ with the following property: there is a covering assignment from each edge $ab$ to a vertex in $\{a,b\}\cap X$ such that for each $x\in X$, the number of edges assigned to it is upper-bounded by $c(x)$. We call a dominating set (or vertex cover) satisfying this additional property \emph{capacitated}.

In order to encode the problems by a \blockMSOtwo\ formula, we first discuss a suitable representation of the instances as graphs with a weight function $w$. A natural encoding would be to set $w(v)=c(v)$.
However, since our formula may only contain a single universal weight term per block,
we apply a complementary approach instead: each vertex will receive the weight $w(v)=\lambda-c(v)$ where $\lambda=\max_{v\in V(G)}c(v)$ is the largest capacity of the instance.
We can now capture \textsc{CDS} by a \blockMSOtwo-query with maximization target $\target(X)=-|X|$ and constraint-formula
\begin{align*}
    \phi_{\textsc{CDS}}(X) :={} &\exists F\colon \text{``each vertex outside $X$ is incident to precisely one edge from $F$''} \\
                             &{}\wedge{} \bigl(\forall x~ \forall A: \text{``\(x \not\in X\) or $A$ is not the neighborhood of \(x\) in the graph $(V(G),F)$''}\\
                             &~~~~{}\vee |A|+w(x)\leq \lambda \bigr),
\end{align*}
where $F$ is an edge set variable, \(x\) is a vertex variable and $A$ is a set variable.
\begin{observation}
\label{obs:CDScorrect}
Every capacitated dominating set $D$ in the original graph satisfies $\phi_{\textsc{CDS}}(D)$ on the constructed weighted graph, and on the other hand every interpretation of $X$ satisfying $\phi_{\textsc{CDS}}(X)$ on the weighted graph is a capacitated dominating set in the original graph.
\end{observation}
\begin{proof}
Given $D$, we know there exists a domination assignment $\alpha$ that does not overload the given capacities. Let us consider the edge set $F$ consisting of the edge connecting each vertex outside $D$ to the vertex $x\in D$ it is assigned to. This clearly satisfies the condition on the first line. For the second line, whenever the first term in the disjunction is not satisfied, the set \(A\) consists of the neighbors of \(v \in D\) reachable via \(F\)-edges. We are guaranteed that the left weight term in the weight comparison sums up to at most $\lambda$ due to the fact that each $v\in D$ is only incident to at most $c(v)=\lambda-w(v)$ edges of $F$ (which means that $|A|\leq \lambda-w(v)$).

On the other hand, given an interpretation of $X$ that satisfies $\phi_{\textsc{CDS}}(X)$ on the weighted graph, we can construct a domination assignment $\alpha$ that assigns each vertex $v\in V(G)\setminus X$ to its unique partner in $F$. The argument bounding the maximum size of a ``valid'' interpretation of $A$ guarantees that $\alpha$ only assigns at most $c(x)$ vertices to each particular $v\in D$.
\end{proof}

In order to capture \textsc{CVC}, we make one additional technical change to the constructed weighted graph in order to allow the logic to define an ``assignment''---specifically, we subdivide each edge and color the newly created vertices with a single color, say, red.
By selecting one of the two edges incident to a red vertex, we can indicate which of the two endpoints covers the edge that the red vertex represents.
While this changes the graph $G$ itself, it is well known that edge subdivisions do not increase the treewidth of a graph (\Cref{obs:subdivision}). Now, we use the same target $\target(X)=-|X|$ along with the constraint-formula
\begin{align*}\phi_{\textsc{CVC}}(X) := {} &\exists F: \bigl(\forall x\in X: \neg \emph{red}(x)\bigr) \\
& {} \wedge \bigl(\text{``each red vertex is incident to precisely one edge from $F$''}\bigr) \\
&{}\wedge  \bigl(\forall x~ \forall A: \text{``\(x \not\in X\) or $A$ is not the neighborhood of \(x\) in the graph $(V(G),F)$''}\\
&~~~~{}\vee |A|+w(x)\leq \lambda \bigr).
\end{align*}
By repeating the argument used in the proof of Observation~\ref{obs:CDScorrect}, we obtain:

\begin{observation}
\label{obs:CVCcorrect}
Every capacitated vertex cover $C$ in the original graph satisfies $\phi_{\textsc{CVC}}(C)$ on the constructed weighted graph, and on the other hand every interpretation of $X$ satisfying $\phi_{\textsc{CVC}}(X)$ on the weighted graph is a capacitated vertex cover in the original graph.
\end{observation}

At this point, we are ready to state the consequences of our encodings. First, by applying \Cref{thm:xp} we immediately 
observe the known \XP-tractability of the problems~\cite{DomLSV08,Lampis14}.

\begin{corollary}
\textsc{Capacitated Dominating Set} and \textsc{Capacitated Vertex Cover} are \XP\ parameterized by the treewidth of the input graph.
\end{corollary}

Note that the tightened and loosened versions of $\phi_{\textsc{CDS}}(X)$ and $\phi_{\textsc{CVC}}(X)$
correspond by Observation~\ref{obs:nicerGuarantees} to additive changes of \(\pm\eps\lambda\) to the right-hand side of the weight comparisons.
By applying Theorem~\ref{thm:mainapproxcw} we obtain:

\begin{corollary}
    \label{cor:capacitated}
    There is a fixed-parameter algorithm which takes as input a graph $G$ of treewidth $\ell$, a function $c:V(G)\rightarrow \Nat$ and $0< \eps\leq 0.5$, is parameterized by $\ell+\frac{1}{\eps}$, and 
    \begin{enumerate}
        \item 		outputs a capacitated dominating set (or capacitated vertex cover) of $(G,c)$ whose size is upper-bounded by that of an optimal solution for the instance $(G, c^-)$ where $c^-(v):=c(v)-\eps \cdot (\max_{v\in V(G)}c(v))$, and
        \item outputs a capacitated dominating set (or capacitated vertex cover) of $(G,c^+)$ where $c^+(v):=c(v)+\eps\cdot (\max_{v\in V(G)}c(v))$, whose size is upper-bounded by that of an optimal solution for $(G,c)$.
    \end{enumerate}
\end{corollary}
One would have naturally hoped for a stronger result where, instead of an additive error of \(\pm\eps\lambda\), 
one has relative errors $c^-(v):=c(v)(1-\eps)$ and $c^+(v):=c(v)\cdot (1+\eps)$.
While we do not see how to obtain such a strengthening via our meta-theorem, the stronger statement with relative errors holds as well~\cite{Lampis14}.

	\subsection{Graph Motif}
	\textsc{Graph Motif} is a problem motivated by pattern-matching applications in bioinformatics~\cite{FellowsFHV11}. 
    As input, we are given a graph $G$ whose vertices are equipped with a color from a color set $C=[k]$ together with a \emph{pattern}, which can be seen as a mapping $M \colon C \to \N$. The question is then whether there exists an \emph{$M$-motif}, that is, a connected subgraph $S$ of $G$ with the following property: for each color $i$, the number of vertices in $S$ colored $i$ is $M(i)$.
Thus, we are not interested in the finer details of a structure, but only require to preserve the basic topological requirement of connectedness.
	
\textsc{Graph Motif} is known to be \XP-tractable when parameterized by treewidth plus $k$~\cite{FellowsFHV11}, and has also been studied under a range of other structural parameterizations~\cite{BonnetS17,Ganian15} but surprisingly not cliquewidth.
The query we will use to capture \textsc{Graph Motif} will have an empty optimization target $\target=0$ and will consist of the formula
\begin{align*}
\phi_{\textsc{GM}}(X):={}&\exists X_1,\dots, X_k: (\text{``$X$ is connected''})\\ 
& \wedge(\text{``$X_1,\dots,X_k$ is the partitioning of $X$ into colors $1,\dots,k$, respectively''})\\ 
&\wedge \bigl(\bigwedge_{i\in [k]}|X_i|\leq M(i) \wedge M(i)\leq |X_i|\bigr).
\end{align*}
It is easy to see that a set \(S\) is an \(M\)-motif if and only if it satisfies $\phi_{\textsc{GM}}(S)$. We remark that the length of the formula $\phi_{\textsc{GM}}(X)$ depends on $k$, and this is necessary: when $k$ is unrestricted, \textsc{Graph Motif} remains \NP-hard, even on instances of constant treewidth~\cite{BonnetS17}. Since $\phi_{\textsc{GM}}(X)$ is a \blockMSOone-formula, we can immediately extend the known \XP-tractability of the problem from treewidth to cliquewidth.

\begin{corollary}
\textsc{Graph Motif} is in \XP\ parameterized by $k$ plus the cliquewidth of the input graph.
\end{corollary}

Turning to approximation, here we see that the two inequalities encode an equality and so the \((1+\eps)\)-tightened formula is unsatisfiable for all $\eps > 0$. On the other hand, in order to satisfy the loosened formula, it suffices to have each $X_i$ comply with the inequalities $(1-\eps)M(i)\leq |X_i|\leq (1+\eps)M(i)$ (after an application of Observation~\ref{obs:nicerGuarantees}). We will refer to a set $X$ partitioned into colors satisfying these inequalities as a $(1+\eps)$-\emph{approximate} $M$-motif. 

If we invoke Theorem~\ref{thm:mainapproxcw} to ask for a $(1+\eps)$-approximate answer to the query $(\phi_{\textsc{GM}},0)$ and receive $\val^+=0$, then we will also receive a witness $\bar W^+$ which is a $(1+\eps)$-approximate $M$-motif. Alternatively, if we receive $\val^+=-\infty$ then this implies $\val(G,\phi,\target)=-\infty$, which excludes the existence of an $M$-motif. In summary:

\begin{corollary}
    There is a fixed-parameter algorithm which takes as input an instance $(G,C=[k],M)$ of \textsc{Graph Motif} and a number $0< \eps\leq 0.5$, is parameterized by \(\frac{1}{\eps}+k+\cw(G)\), and either computes a $(1+\eps)$-approximate $M$-motif or correctly determines that $G$ does not contain an $M$-motif.
\end{corollary}

\clearpage
\part*{Proofs}

\section{Exact Extended Feferman--Vaught}\label{sec:exactFV}

In this section, we explore an exact composition theorem in the style of Feferman--Vaught for special kinds of \blockMSO-formulas.
Later on, in \Cref{sec:lookup}, we will see that approximate information about these special \blockMSO-formulas is enough
to approximately evaluate all \blockMSO-formulas.

	\subsection{Definition of Table Formulas}\label{sec:defoftables}

    We start by defining the central data structure of all our proofs and algorithms. As an allusion the typical dynamic programming algorithms, we call it a \emph{table}.
    \begin{definition}[Table]
    Let $\sigma$ be a \CMSO$_1$-signature, $q \in \N$ be a quantifier rank, $\bar X,\bar Y$ be tuples of set variables, 
    $\terms_1$ be a set of integral weight terms over $\bar X$, and $\terms_2$ be a set of integral weight terms over $\bar X\bar Y$.
    We define a \emph{table} $T(\sigma,q,\bar X,\bar Y,\terms_1,\terms_2)$ to be the set of all \emph{table formulas} $\omega(\bar X)$ of the form
	\begin{multline*}
		\omega(\bar X) :=
		\Bigl(\bigwedge_{t \in \terms_1} t(\bar X) \le g_{t,\lle} \land t(\bar X) \ge g_{t,\gge} \Bigr) ~ \land~ \\
		\Bigl( \bigwedge_{t \in \terms_2} \bigwedge_{\psi \in \cmsoset} \forall \bar Y \psi(\bar X \bar Y) \lor \bigl( t(\bar X \bar Y) \le g_{t,\psi,\lle} \land t(\bar X \bar Y) \ge g_{t,\psi,\gge} \bigr) \Bigr),
	\end{multline*}
    where $g_{t_1,\lle}, g_{t_1,\gge}, g_{t_2,\psi,\lle}, g_{t_2,\psi,\gge} \in \N$ for all $t_1 \in \terms_1$, $t_2 \in \terms_2$, $\psi \in \cmsoset$.

    Observe that once a table is fixed, each table formula in that table is fully determined by the numerical constants it contains.
	For each such formula $\omega(\bar X)$, we call its numerical constants \emph{thresholds} and denote them by
    \begin{multicols}{2}
	\begin{itemize}
		\item $\thres_\omega(t,\lle) := g_{t,\lle}$,
		\item $\thres_\omega(t,\gge) := g_{t,\gge}$,
		\item $\thres_\omega(t,\psi,\lle) := g_{t,\psi,\lle}$,
		\item $\thres_\omega(t,\psi,\gge) := g_{t,\psi,\gge}$.
	\end{itemize}
    \end{multicols}
    \end{definition}
We note that while for each fixed $(\sigma,q,\bar X, \bar Y, \terms_1,\terms_2)$ the table $T(\sigma,q,\bar X,\bar Y,\terms_1,\terms_2)$ only contains formulas of bounded length, the size of the table itself is infinite since table formulas range over all possible numbers in the inequalities. A restriction of $T(\sigma,q,\bar X,\bar Y,\terms_1,\terms_2)$ to a bounded-size set for algorithmic purposes will be discussed later; for now we will proceed with this mathematically cleaner formulation.

\subsection{Composition}
We extend the classical Feferman--Vaught theorem to our notion of tables.
Remember that \(\FV(\psi)\) is the ``Feferman--Vaught set'' from \Cref{thm:fv}.
We extend this notion to formulas \(\omega(\bar X)\) from our tables,
defining ``extended Feferman--Vaught sets'' \(\FVplus(\omega)\).
The definition of these sets is motivated by the following idea:
If I want a subset of vertices in the union \(G_1 \oplus G_2\) of size, say, at least \(g\),
then I can guess a partition of \(g\) into numbers \(g_1+g_2=g\) and look for a set of size at least \(g_1\) in \(G_1\) and size at least \(g_2\) in \(G_2\).
The full argument, however, becomes more involved since table formulas additionally contain a univerally quantified part.
The following formulation of our theorem and the corresponding Feferman--Vaught set is not algorithmic, as it
deals with a table $T$ of infinite size. 
In the upcoming \Cref{sec:approxFV}, we will make the result finitary, approximate and algorithmic.
For brevity, we also use the notation $\negFV(\psi) := \{(\psi_1,\psi_2) \mid (\neg\psi_1,\neg\psi_2) \in \FV(\neg \psi)\}$.

	\begin{definition}[Extended Feferman--Vaught Set]
	\label{def:FVplus}
	Let $T:= T(q,c,\bar X,\bar Y,\terms_1,\terms_2)$ be a table and $\omega(\bar X)\in T$ a table formula. Let $\FVplus(\omega) \subseteq T \times T$ be the set of all tuples of table formulas $(\omega_1,\omega_2)$ such that

		\begin{itemize}
			\item 
            for all \(t \in \terms_1\), we have \\
            \(\thres_{\omega_1}(t,\lle) + \thres_{\omega_2}(t,\lle) \leq \thres_{\omega}(t,\lle)\), and \\
			\(\thres_{\omega_1}(t,\gge) + \thres_{\omega_2}(t,\gge) \geq \thres_{\omega}(t,\gge)\), and
        \item for all $t \in \terms_2$, $\psi \in \cmsoset$, and $(\psi_1,\psi_2) \in \negFV(\psi)$, we have \\
            \({\thres_{\omega_1}(t,\psi,\lle)} + \thres_{\omega_2}(t,\psi,\lle) \le \thres_{\omega}(t,\psi,\lle)\), and \\
            \(\thres_{\omega_1}(t,\psi,\gge) + \thres_{\omega_2}(t,\psi,\gge) \geq {\thres_{\omega}(t,\psi,\gge)}\).
		\end{itemize}
		\end{definition}
	
	\begin{theorem}\label{thm:base}
        Let $T := T(\sigma,q,\bar X,\bar Y,\terms_1,\terms_2)$. 
		For every $\omega(\bar X) \in T$, the set $\FVplus(\omega)$ satisfies the following property: for all graphs $G_1$, $G_2$ and $\bar W_1 \in \PP(V(G_1))^{|\bar X|}$, $\bar W_2 \in \PP(V(G_2))^{|\bar X|}$,
		$$
			G_1 \oplus G_2 \models \omega(\bar W_1 \cup \bar W_2) \iff 
			G_1 \models \omega_1(\bar W_1) \text{ and } G_2 \models \omega_2(\bar W_2) \text{ for some } (\omega_1,\omega_2) \in \FVplus(\omega).
		$$
	\end{theorem}
	\begin{proof}
		We will prove the theorem by deriving a series of equivalent boxed statements. The first statement is simply the left side of the equivalence in the theorem.
		
		\mybox{
			$$
			G_1 \oplus G_2 \models \omega(\bar W_1 \cup \bar W_2).
			$$
		}
		
		We now expand \(\omega\) according to the definition of table formulas, where $\bar U_i$ is the restriction of \( \bar U\) to \(G_i\).
        From here on, we divide the box into sub-boxes to emphasize a conjunction between sub-statements.
		
		\mybox{
			\begin{itemize}
				\item for all \(t \in \terms_1\),
				\begin{itemize}[noitemsep,topsep=0pt,leftmargin=2mm]
                \item \(G_1 \oplus G_2 \models t(\bar{W}_1 \cup \bar{W}_2) \le \thres_\omega(t,\lle)\) and \(G_1 \oplus G_2 \models t(\bar{W}_1 \cup \bar{W}_2) \geq \thres_\omega(t,\gge)\)
                \end{itemize}

                    \bigskip
\centerline{\rule{0.995\textwidth}{.2pt}}

				\item for all $t \in \terms_2$, $\psi \in \cmsoset$,
								\begin{itemize}[noitemsep,topsep=0pt,leftmargin=2mm]
                \item for all $\bar U_1\cup \bar U_2 \in \PP(V(G_1 \oplus G_2))^{|\bar Y|}$,
                				\begin{itemize}[noitemsep,topsep=0pt,leftmargin=2mm]
				\item $G_1 \oplus G_2 \models \psi((\bar W_1 \cup \bar W_2) (\bar U_1 \cup \bar U_2))$ or
				\item $G_1 \oplus G_2 \models t(\bar W(\bar U_1 \cup \bar U_2)) \le \thres_\omega(t,\psi,\lle)$ and \(G_1 \oplus G_2 \models t(\bar W(\bar U_1 \cup \bar U_2)) \geq \thres_\omega(t,\psi,\gge).\)
				\end{itemize} \end{itemize}
			\end{itemize}%
		}

		Our first few steps will target terms in \(\terms_2\), i.e., involve only the second half of each box.
        As all $\bar U_1 \cup \bar U_2 \in \PP(V(G_1 \oplus G_2))^{|\bar Y|}$
		that do not satisfy $G_1 \oplus G_2 \models \psi((\bar W_1 \cup \bar W_2) (\bar U_1 \cup \bar U_2))$
		must satisfy $G_1 \oplus G_2 \models t((\bar W_1 \cup \bar W_2) (\bar U_1 \cup \bar U_2)) \le \thres_\omega(t,\psi,\lle)$ and $G_1 \oplus G_2 \models t((\bar W_1 \cup \bar W_2) (\bar U_1 \cup \bar U_2)) \geq \thres_\omega(t,\psi,\gge)$,
		we can reformulate the latter two statements to reference the extremal values as follows.
		
		\mybox{
			\begin{itemize}
				\item for all \(t \in \terms_1\),
				\begin{itemize}[noitemsep,topsep=0pt,leftmargin=2mm]
                \item \(G_1 \oplus G_2 \models t(\bar{W}_1 \cup \bar{W}_2) \le \thres_\omega(t,\lle)\) and \(G_1 \oplus G_2 \models t(\bar{W}_1 \cup \bar{W}_2) \geq \thres_\omega(t,\gge)\)
                \end{itemize}

                    \bigskip
\centerline{\rule{0.995\textwidth}{.2pt}}

				\item for all $t \in \terms_2$, $\psi \in \cmsoset$,
								\begin{itemize}[noitemsep,topsep=0pt,leftmargin=2mm]
				\item $\displaystyle\max_{\{\bar U_1 \oplus \bar U_2~|~G_1\cup G_2 \models \neg \psi((\bar W_1 \cup \bar W_2) (\bar U_1 \oplus \bar U_2))\}} t^{G_1 \oplus G_2}((\bar W_1 \cup \bar W_2) (\bar U_1 \oplus \bar U_2)) \le \thres_\omega(t,\psi,\lle)$, and
                \item $\displaystyle\min_{\{\bar U_1 \oplus \bar U_2~|~G_1\cup G_2 \models \neg \psi((\bar W_1 \cup \bar W_2) (\bar U_1 \oplus \bar U_2))\}} t^{G_1 \oplus G_2}((\bar W_1 \cup \bar W_2) (\bar U_1 \oplus \bar U_2)) \ge \thres_\omega(t,\psi,\gge)$.
				\end{itemize} 
			\end{itemize}%
		}

        We now apply \Cref{thm:fv} to the formula $\neg\psi(\bar X\bar Y)$ to obtain
        for all $\bar U_1 \in \PP(V(G_1))^{|\bar Y|}$ and $\bar U_2 \in \PP(V(G_2))^{|\bar Y|}$,
        \begin{multline}\label{eq:asdf1}
            G_1 \oplus G_2 \models \neg\psi((\bar W_1 \cup \bar W_2) (\bar U_1 \cup \bar U_2)) \iff \\
            G_1 \models \neg\psi_1(\bar W_1 \bar U_1) \text{ and } G_2 \models \neg\psi_2(\bar W_2 \bar U_2) \text{ for some } (\psi_1,\psi_2) \in \negFV(\psi).
        \end{multline}
        Let $\operatorname{opt} \in \{\min,\max\}$ and $t \in \terms_2$.
        Optimizing $t^{G_1 \oplus G_2}((\bar W_1 \cup \bar W_2) (\bar U_1 \cup \bar U_2))$ over all $\bar U_1 \cup \bar U_2$ satisfying the left-hand side of (\ref{eq:asdf1}) yields
        \begin{equation}\label{eq:asdf2}
            \operatorname*{opt}\limits_{\{\bar U_1 \oplus \bar U_2~|~G_1\cup G_2 \models \neg \psi((\bar W_1 \cup \bar W_2) (\bar U_1 \oplus \bar U_2))\}} t^{G_1 \oplus G_2}((\bar W_1 \cup \bar W_2) (\bar U_1 \oplus \bar U_2)),
        \end{equation}
        and similarly, optimizing $t^{G_1}(\bar W_1 \bar U_1) + t^{G_2}(\bar W_2 \bar U_2)$ over all $\bar U_1, \bar U_2$ satisfying the right-hand side of (\ref{eq:asdf1}) yields
        \begin{equation}\label{eq:asdf3}
            \operatorname*{opt}_{(\psi_1,\psi_2) \in \negFV(\psi)} \Bigl(
            \operatorname*{opt}_{\{\bar U_1 ~|~G_1 \models \neg \psi_1(\bar W_1 \bar U_1)\}} t^{G_1}(\bar W_1 \bar U_1) +
            \operatorname*{opt}_{\{\bar U_2 ~|~G_2 \models \neg \psi_2(\bar W_2 \bar U_2)\}} t^{G_2}(\bar W_2 \bar U_2) \Bigr).
        \end{equation}
        Since $t^{G_1 \oplus G_2}((\bar W_1 \cup \bar W_2) (\bar U_1 \cup \bar U_2)) = t^{G_1 \oplus G_2}(\bar W_1 \bar U_1) + t(\bar W_2 \bar U_2)$ for all \(\bar U_1 \in \PP(V(G_1))^{|\bar Y|}\) and \(\bar U_2 \in \PP(V(G_2))^{|\bar Y|}\),
        and the left and right side of (\ref{eq:asdf1}) are equivalent,
        we conclude that (\ref{eq:asdf2}) and (\ref{eq:asdf3}) are equal.
        Hence, the previous and following boxes are equivalent.

		\mybox{
			\begin{itemize}
				\item for all \(t \in \terms_1\),
				\begin{itemize}[noitemsep,topsep=0pt,leftmargin=2mm]
                \item \(G_1 \oplus G_2 \models t(\bar{W}_1 \cup \bar{W}_2) \le \thres_\omega(t,\lle)\) and \(G_1 \oplus G_2 \models t(\bar{W}_1 \cup \bar{W}_2) \geq \thres_\omega(t,\gge)\)
                \end{itemize}

                    \bigskip
                    \centerline{\rule{0.995\textwidth}{.2pt}}

				\item for all $t \in \terms_2$, $\psi \in \cmsoset$,
								\begin{itemize}[noitemsep,topsep=0pt,leftmargin=2mm]
                \item \hspace{-6.5mm}$\displaystyle
                    \max_{(\psi_1,\psi_2) \in \negFV(\psi)} \Bigl(
                    \max_{\{\bar U_1 ~|~G_1 \models \neg \psi_1(\bar W_1 \bar U_1)\}} t^{G_1}(\bar W_1 \bar U_1) + 
                    \max_{\{\bar U_2 ~|~G_2 \models \neg \psi_2(\bar W_2 \bar U_2)\}} t^{G_2}(\bar W_2 \bar U_2) \Bigr)
                    \le \thres_\omega(t,\psi,\lle), 
                    $
                \item \hspace{-6.5mm}$\displaystyle
                    \min_{(\psi_1,\psi_2) \in \negFV(\psi)} \Bigl(
                    \min_{\{\bar U_1 ~|~G_1 \models \neg \psi_1(\bar W_1 \bar U_1)\}} t^{G_1}(\bar W_1 \bar U_1) + 
                    \min_{\{\bar U_2 ~|~G_2 \models \neg \psi_2(\bar W_2  \bar U_2)\}} t^{G_2}(\bar W_2 \bar U_2) \Bigr)
                    \geq \thres_\omega(t,\psi,\gge).
                    $
                    \end{itemize}
			\end{itemize}
		}
		
        In the last two lines, we turn the maximization and minimization over $(\psi_1,\psi_2) \in \negFV(\psi)$ into a universal quantification.

		\mybox{
			\begin{itemize}
				\item for all \(t \in \terms_1\),
				\begin{itemize}[noitemsep,topsep=0pt,leftmargin=2mm]
                \item \(G_1 \oplus G_2 \models t(\bar{W}_1 \cup \bar{W}_2) \le \thres_\omega(t,\lle)\) and \(G_1 \oplus G_2 \models t(\bar{W}_1 \cup \bar{W}_2) \geq \thres_\omega(t,\gge)\)
                \end{itemize}

                    \bigskip
                    \centerline{\rule{0.995\textwidth}{.2pt}}

				\item for all $t \in \terms_2$, $\psi \in \cmsoset$, 
				\begin{itemize}[noitemsep,topsep=0pt,leftmargin=2mm]
                \item for all $(\psi_1,\psi_2) \in \negFV(\psi),$
				\begin{itemize}[noitemsep,topsep=0pt, leftmargin=2mm]               
                \item $\displaystyle
                    \max_{\{\bar U_1 ~|~G_1 \models \neg \psi_1(\bar W_1 \bar U_1)\}} t^{G_1}(\bar W_1 \bar U_1) + 
                    \max_{\{\bar U_2 ~|~G_2 \models \neg \psi_2(\bar W_2 \bar U_2)\}} t^{G_2}(\bar W_2 \bar U_2)
                    \le \thres_\omega(t,\psi,\lle), \text{ and } 
                    $
                \item $\displaystyle
                    \min_{\{\bar U_1 ~|~G_1 \models \neg \psi_1(\bar W_1 \bar U_1)\}} t^{G_1}(\bar W_1 \bar U_1) + 
                    \min_{\{\bar U_2 ~|~G_2 \models \neg \psi_2(\bar W_2 \bar U_2)\}} t^{G_2}(\bar W_2 \bar U_2)
                    \geq \thres_\omega(t,\psi,\gge)
                    $.
                    \end{itemize}
                \end{itemize}				                    
			\end{itemize}
		}

        Now, we change our perspective by rephrasing the quantification at the top of the second half.
        This subtle reformulation can be seen as the centerpiece of the proof.

		\mybox{
			\begin{itemize}
				\item for all \(t \in \terms_1\),
								\begin{itemize}[noitemsep,topsep=0pt,leftmargin=2mm]
                \item \(G_1 \oplus G_2 \models t(\bar{W}_1 \cup \bar{W}_2) \le \thres_\omega(t,\lle)\) and \(G_1 \oplus G_2 \models t(\bar{W}_1 \cup \bar{W}_2) \geq \thres_\omega(t,\gge)\)
                \end{itemize}

                    \bigskip
                    \centerline{\rule{0.995\textwidth}{.2pt}}

				\item for all $t \in \terms_2$, for all $\psi_1,\psi_2 \in \cmsoset$,
				\begin{itemize}[noitemsep,topsep=0pt, leftmargin=2mm]               				
				\item for all $\psi \in \cmsoset$ such that $(\psi_1,\psi_2) \in \negFV(\psi)$,
				\begin{itemize}[noitemsep,topsep=0pt, leftmargin=2mm]               				
                \item$\displaystyle
                    \max_{\{\bar U_1 ~|~G_1 \models \neg \psi_1(\bar W_1 \bar U_1)\}} t^{G_1}(\bar W_1 \bar U_1) + 
                    \max_{\{\bar U_2 ~|~G_2 \models \neg \psi_2(\bar W_2 \bar U_2)\}} t^{G_2}(\bar W_2 \bar U_2) 
                    \le \thres_\omega(t,\psi,\lle), \text{ and }
                    $
                \item$\displaystyle
                    \min_{\{\bar U_1 ~|~G_1 \models \neg \psi_1(\bar W_1 \bar U_1)\}} t^{G_1}(\bar W_1 \bar U_1) + 
                    \min_{\{\bar U_2 ~|~G_2 \models \neg \psi_2(\bar W_2 \bar U_2)\}} t^{G_2}(\bar W_2 \bar U_2) 
                    \geq \thres_\omega(t,\psi,\gge). \quad\quad~~
                    $
                    \end{itemize}
                    \end{itemize}
			\end{itemize}
		}

        Fix $t \in \terms_1$.
        Note that $t^{G_1 \oplus G_2}(\bar W_1 \oplus \bar W_2) = t^{G_1}(\bar W_1)+t^{G_2}(\bar W_2)$, and thus
        $t(\bar W_1 \bar W_2) \le \thres_\omega(t,\lle)$ if and only if
        there exist $a,b \in \N$, $0\leq a,b\leq \thres_\omega(t,\lle)$, with 
        $t^{G_1}(\bar{W}_1) \le a$, $t^{G_2}(\bar{W}_2) \le b$ and $a + b \le \thres_\omega(t,\lle)$.
        We can similarly use numbers $a'$ and $b'$ to split up $\thres_\omega(t,\gge)$: 
                $t^{G_1 \oplus G_2}(\bar W_1 \cup \bar W_2) \ge \thres_\omega(t,\gge)$ if and only if
        there exist $a',b' \in \N$, $0\leq a',b'\leq \thres_\omega(t,\gge)$, with
        $t^{G_1}(\bar{W}_1) \ge a'$, $t^{G_2}(\bar{W}_2) \ge b'$ and $a' + b' \ge \thres_\omega(t,\gge)$.
                This implies the equivalence of the first halves of the previous and following block.

		Moreover, let us fix $t \in \terms_2$ and $\psi_1,\psi_2 \in \cmsoset$. Since the left-side terms of the inequalities at the bottom of the previous box do not depend on the choice of $\psi$, we can view these as a sum of two numbers that are determined by the choice of $t$, $\psi_1$, $\psi_2$.
        For any two such numbers $M_1$ and $M_2$ we have the following:
        $$
        M_1 + M_2 \le \thres_\omega(t,\psi,\lle) \text{ holds for all $\psi \in \cmsoset$ with $(\psi_1,\psi_2) \in \negFV(\psi)$ }
        $$
        if and only if there exists $a,b \in \N$ such that for all $\psi \in \cmsoset$ with $(\psi_1,\psi_2) \in \negFV(\psi)$
        we have $a+b \le \thres_\omega(t,\psi,\lle)$ and
        $$
        M_1 \le a, \quad \quad  M_2 \le b.
        $$
        This way, we split each threshold $\thres_\omega(t,\psi,\lle)$ and $\thres_\omega(t,\psi,\gge)$ 
        into two summands corresponding to the terms taken on \(G_1\) and \(G_2\) respectively,
        proving the equivalence of the second halves of the previous and following block.
        We use subscripts on our $a$- and $b$-values to uniquely identity them.

		\mybox{
			\begin{itemize}
				\item for all \(t \in \terms_1\),
								\begin{itemize}[noitemsep,topsep=0pt,leftmargin=2mm]				
                \item there exist \(a_{t,\lle},b_{t,\lle},a_{t,\gge},b_{t,\gge} \in \N\) such that
								\begin{itemize}[noitemsep,topsep=0pt,leftmargin=2mm]				                
\item                     \(a_{t,\lle} + b_{t,\lle} \leq \thres_\omega(t,\lle)\), \quad \quad \quad \quad \quad \quad
                    \(a_{t,\gge} + b_{t,\gge} \geq \thres_\omega(t,\gge)\), and 
                \item \(t^{G_1}(\bar{W}_1) \le a_{t,\lle}\), \quad \quad \(t^{G_2}(\bar{W}_2) \le b_{t,\lle}\), \quad \quad \(t^{G_1}(\bar{W}_1) \geq a_{t,\gge}\), \quad \quad \(t^{G_2}(\bar{W}_2) \geq b_{t,\gge}\).
                \end{itemize}
                \end{itemize}

                    \bigskip
                    \centerline{\rule{0.995\textwidth}{.2pt}}
                    
				\item for all $t \in \terms_2$, for all $\psi_1,\psi_2 \in \cmsoset$,
				\begin{itemize}[noitemsep,topsep=0pt, leftmargin=2mm]               					
                \item 
                    there exist \(a_{t,\psi_1,\lle},b_{t,\psi_2,\lle},a_{t,\psi_1,\gge},b_{t,\psi_2,\gge} \in \N\) such that for all $\psi \in \cmsoset$ 
                    
                    satisfying $(\psi_1,\psi_2) \in \negFV(\psi)$ we have
				\begin{itemize}[noitemsep,topsep=0pt, leftmargin=2mm]               				                    
                    \item $a_{t,\psi_1,\lle}+b_{t,\psi_2,\lle} \le \thres_\omega(t,\psi,\lle)$, \quad \quad \quad \quad \quad \quad
                    $a_{t,\psi_1,\gge}+b_{t,\psi_2,\gge} \ge \thres_\omega(t,\psi,\gge)$, and
                \item$\displaystyle
                    \max_{\{\bar U_1 ~|~G_1 \models \neg \psi_1(\bar W_1 \bar U_1)\}} t^{G_1}(\bar W_1 \bar U_1) \le a_{t,\psi_1,\lle}, \quad
                    \max_{\{\bar U_2 ~|~G_2 \models \neg \psi_2(\bar W_2 \bar U_2)\}} t^{G_2}(\bar W_2 \bar U_2) \le b_{t,\psi_2,\lle},
                    $ and
                \item$\displaystyle
                    \min_{\{\bar U_1 ~|~G_1 \models \neg \psi_1(\bar W_1 \bar U_1)\}} t^{G_1}(\bar W_1 \bar U_1) \ge a_{t,\psi_1,\gge}, \quad
                    \min_{\{\bar U_2 ~|~G_2 \models \neg \psi_2(\bar W_2 \bar U_2)\}} t^{G_2}(\bar W_2 \bar U_2) \ge b_{t,\psi_2,\gge}. $
                    \end{itemize}
                    \end{itemize}
			\end{itemize}
		}

Next, we will swap the quantification at the beginning of the statement by existentially quantifying over the whole set of $a$-values (i.e., values of the form $a_{\circ,\circ}$) and $b$-values (values of the form $b_{\circ,\circ,\circ}$). The conditions are then checked over all possible choices of elements from $\terms_1$, $\terms_2$, and $\cmsosetx$. At the same time, we reshuffle the conditions in a way which will be useful later on: the second and fourth sub-box contains the conditions which were previously in the top sub-box, while the third and fifth sub-box contains the conditions which were in the bottom sub-box.

		\mybox{
			\begin{itemize}
                \item there exist \(a_{t^\exists,\lle},b_{t^\exists,\lle},a_{t^\exists,\gge},b_{t^\exists,\gge} \in \N\) for all choices of $t^\exists\in \terms_1$
								\begin{itemize}[noitemsep,topsep=0pt,leftmargin=2mm]				                
\item there exist \(a_{t^\forall,\psi_1,\lle},b_{t^\forall,\psi_2,\lle},a_{t^\forall,\psi_1,\gge},b_{t^\forall,\psi_2,\gge} \in \N\) for all choices of $t^\forall \in \terms_2$ 

and of $\psi'_1,\psi'_2 \in \cmsoset$

                    \smallskip
                    \centerline{\rule{0.995\textwidth}{.2pt}}
                     \smallskip

								\begin{itemize}[noitemsep,topsep=0pt,leftmargin=2mm]				                
\item for all $t\in \terms_1$
\begin{itemize}[noitemsep,leftmargin=2mm]
\item \(a_{t,\lle} + b_{t,\lle} \leq \thres_\omega(t,\lle)\), \quad \quad \quad \quad \quad \quad   
\(a_{t,\gge} + b_{t,\gge} \geq \thres_\omega(t,\gge)\), and
\end{itemize}
                    \smallskip
                    \centerline{\rule{0.995\textwidth}{.2pt}}
                     \smallskip

\item   for all $t'\in \terms_2$ and for all $\psi_1,\psi_2 \in \cmsoset$
\begin{itemize}[noitemsep,leftmargin=2mm]
\item $a_{t',\psi_1,\lle}+b_{t',\psi_2,\lle} \le \thres_\omega(t',\psi,\lle)$, \quad \quad \quad \quad \quad \quad
                    $a_{t',\psi_1,\gge}+b_{t',\psi_2,\gge} \ge \thres_\omega(t',\psi,\gge)$, and
\end{itemize}                    
                    \smallskip
                    \centerline{\rule{0.995\textwidth}{.2pt}}
                     \smallskip

\item for all $t\in \terms_1$
\begin{itemize}[noitemsep,leftmargin=2mm]
\item \(t^{G_1}(\bar{W}_1) \le a_{t,\lle}\), \quad \(t^{G_2}(\bar{W}_2) \le b_{t,\lle}\), \quad \(t^{G_1}(\bar{W}_1) \geq a_{t,\gge}\), \quad \(t^{G_2}(\bar{W}_2) \geq b_{t,\gge}\), and
\end{itemize}
                    \smallskip
                    \centerline{\rule{0.995\textwidth}{.2pt}}
                     \smallskip

\item   for all $t'\in \terms_2$ and for all $\psi_1,\psi_2 \in \cmsoset$
\begin{itemize}[noitemsep,leftmargin=2mm]
                \item$\displaystyle
                    \max_{\{\bar U_1 ~|~G_1 \models \neg \psi_1(\bar W_1 \bar U_1)\}} {t'}^{G_1}(\bar W_1 \bar U_1) \le a_{t',\psi_1,\lle}, \quad
                    \max_{\{\bar U_2 ~|~G_2 \models \neg \psi_2(\bar W_2 \bar U_2)\}} {t'}^{G_2}(\bar W_2 \bar U_2) \le b_{t',\psi_2,\lle},
                    $ and
                \item$\displaystyle
                    \min_{\{\bar U_1 ~|~G_1 \models \neg \psi_1(\bar W_1 \bar U_1)\}} {t'}^{G_1}(\bar W_1 \bar U_1) \ge a_{t',\psi_1,\gge}, \quad
                    \min_{\{\bar U_2 ~|~G_2 \models \neg \psi_2(\bar W_2 \bar U_2)\}} {t'}^{G_2}(\bar W_2 \bar U_2) \ge b_{t',\psi_2,\gge}. $
\end{itemize}                    
\end{itemize}
\end{itemize}
\end{itemize}
}

For our next step, it will be useful to recall that each table formula in $T$ includes (1) values $g_{t^\exists,\lle}$ and $g_{t^\exists,\gge}$ for each $t^\exists\in \terms_1$, and also (2) values $g_{t^\forall,\psi,\lle}$ and $g_{t^\forall,\psi,\gge}$ for each $t^\forall\in \terms_2$ and each formula $\psi \in \cmsoset$. Since $T$ is the set of all table formulas, the existence of the numbers specified in the first sub-box is equivalent to the existence of two table formulas: one whose $g_{\circ,\circ}$ and $g_{\circ,\circ,\circ}$ numbers are precisely the $a$-values, and another where these are the $b$-values. We reformulate the statement in the first three sub-boxes accordingly.

	\mybox{
			\begin{itemize}
                \item there exist formulas $\omega_1,\omega_2 \in T$ such that                
								\begin{itemize}[noitemsep,topsep=0pt,leftmargin=2mm]				                
\item for all $t\in \terms_1$
\begin{itemize}[noitemsep,leftmargin=2mm]
\item \(g_{\omega_1}(t,\lle) + g_{\omega_2}(t,\lle) \leq \thres_\omega(t,\lle)\), \quad \quad \quad \quad \quad \quad   
\(g_{\omega_1}(t,\gge) + g_{\omega_2}(t,\gge) \geq \thres_\omega(t,\gge)\), and
\end{itemize}

\item   for all $t'\in \terms_2$ and for all $\psi_1,\psi_2 \in \cmsoset$
\begin{itemize}[noitemsep,leftmargin=2mm]
\item $g_{\omega_1}(t',\psi_1,\lle)+g_{\omega_2}(t',\psi_2,\lle) \le \thres_\omega(t',\psi,\lle)$, and
\item      $g_{\omega_1}(t',\psi_1,\gge)+g_{\omega_2}(t',\psi_2,\gge) \ge \thres_\omega(t',\psi,\gge)$, and
\end{itemize}                    
                    \smallskip
                    \centerline{\rule{0.995\textwidth}{.2pt}}
                     \smallskip

\item for all $t\in \terms_1$
\begin{itemize}[noitemsep,leftmargin=2mm]
\item \(t^{G_1}(\bar{W}_1) \le g_{\omega_1}(t,\lle)\), \quad \(t^{G_2}(\bar{W}_2) \le g_{\omega_2}(t,\lle)\), \quad \(t^{G_1}(\bar{W}_1) \geq g_{\omega_1}(t,\gge)\), \quad \(t^{G_2}(\bar{W}_2) \geq g_{\omega_2}(t,\gge)\), and
\end{itemize}
                    \smallskip
                    \centerline{\rule{0.995\textwidth}{.2pt}}
                     \smallskip

\item   for all $t'\in \terms_2$ and for all $\psi_1,\psi_2 \in \cmsoset$
\begin{itemize}[noitemsep,leftmargin=2mm]
                \item$\displaystyle
                    \max_{\{\bar U_1 ~|~G_1 \models \neg \psi_1(\bar W_1 \bar U_1)\}} {t'}^{G_1}(\bar W_1 \bar U_1) \le g_{\omega_1}(t',\psi_1,\lle),$ and
\item                    $\displaystyle\max_{\{\bar U_2 ~|~G_2 \models \neg \psi_2(\bar W_2 \bar U_2)\}} {t'}^{G_2}(\bar W_2 \bar U_2) \le g_{\omega_2}(t',\psi_2,\lle),
                    $ and
                \item$\displaystyle
                    \min_{\{\bar U_1 ~|~G_1 \models \neg \psi_1(\bar W_1 \bar U_1)\}} {t'}^{G_1}(\bar W_1 \bar U_1) \ge g_{\omega_1}(t',\psi_1,\gge),$ and
\item                    $\displaystyle\min_{\{\bar U_2 ~|~G_2 \models \neg \psi_2(\bar W_2 \bar U_2)\}} {t'}^2(\bar W_2 \bar U_2) \ge g_{\omega_2}(t',\psi_2,\gge). $
\end{itemize}                    
\end{itemize}
\end{itemize}
}

        Checking the definition of $\FVplus(\omega)$, we see
        that any two formulas $\omega_1,\omega_2 \in T$ satisfy the requirements stated in the box above
        if and only if $(\omega_1,\omega_2)\in \FVplus(\omega)$.
        We substitute the first sub-box accordingly.

\mybox{
			\begin{itemize}
                \item there exist $(\omega_1,\omega_2) \in \FVplus(\omega)$ such that
                
                    \smallskip
                    \centerline{\rule{0.995\textwidth}{.2pt}}
\begin{itemize}[noitemsep,leftmargin=2mm]
\item for all $t\in \terms_1$
\begin{itemize}[noitemsep,leftmargin=2mm]
\item \(t^{G_1}(\bar{W}_1) \le g_{\omega_1}(t,\lle)\), \quad \(t^{G_2}(\bar{W}_2) \le g_{\omega_2}(t,\lle)\), \quad \(t^{G_1}(\bar{W}_1) \geq g_{\omega_1}(t,\gge)\), \quad \(t^{G_2}(\bar{W}_2) \geq g_{\omega_2}(t,\gge)\), and
\end{itemize}
                    \smallskip
                    \centerline{\rule{0.995\textwidth}{.2pt}}
                     \smallskip
\item   for all $t'\in \terms_2$ and for all $\psi_1,\psi_2 \in \cmsoset$
\begin{itemize}[noitemsep,leftmargin=2mm]
                \item$\displaystyle
                    \max_{\{\bar U_1 ~|~G_1 \models \neg \psi_1(\bar W_1 \bar U_1)\}} {t'}^{G_1}(\bar W_1 \bar U_1) \le g_{\omega_1}(t',\psi_1,\lle),$ and
\item                    $\displaystyle\max_{\{\bar U_2 ~|~G_2 \models \neg \psi_2(\bar W_2 \bar U_2)\}} {t'}^{G_2}(\bar W_2 \bar U_2) \le g_{\omega_2}(t',\psi_2,\lle),
                    $ and
                \item$\displaystyle
                    \min_{\{\bar U_1 ~|~G_1 \models \neg \psi_1(\bar W_1 \bar U_1)\}} {t'}^{G_1}(\bar W_1 \bar U_1) \ge g_{\omega_1}(t',\psi_1,\gge),$ and
\item                    $\displaystyle\min_{\{\bar U_2 ~|~G_2 \models \neg \psi_2(\bar W_2 \bar U_2)\}} {t'}^{G_2}(\bar W_2 \bar U_2) \ge g_{\omega_2}(t',\psi_2,\gge). $
\end{itemize}       
\end{itemize}             
\end{itemize}
}

All that we need to do now is regroup the statements into those depending on $G_1$ and those depending on $G_2$.

\mybox{
			\begin{itemize}
                \item there exist $(\omega_1,\omega_2) \in \FVplus(\omega)$ such that
                
                    \smallskip
                    \centerline{\rule{0.995\textwidth}{.2pt}}
                    \begin{itemize}[noitemsep,leftmargin=2mm]
\item for all $t\in \terms_1$
\begin{itemize}[noitemsep,leftmargin=2mm]
\item \(G_1 \models t(\bar{W}_1) \le g_{\omega_1}(t,\lle)\), \quad \(G_1 \models t(\bar{W}_1) \geq g_{\omega_1}(t,\gge)\), and
\end{itemize}
\item   for all $t'\in \terms_2$ and for all $\psi_1,\psi_2 \in \cmsoset$
\begin{itemize}[noitemsep,leftmargin=2mm]
                \item$\displaystyle
                    \max_{\{\bar U_1 ~|~G_1 \models \neg \psi_1(\bar W_1 \bar U_1)\}} {t'}^{G_1}(\bar W_1 \bar U_1) \le g_{\omega_1}(t',\psi_1,\lle),$ and
                \item$\displaystyle
                    \min_{\{\bar U_1 ~|~G_1 \models \neg \psi_1(\bar W_1 \bar U_1)\}} {t'}^{G_2}(\bar W_1 \bar U_1) \ge g_{\omega_1}(t',\psi_1,\gge),$ and
\end{itemize}       
\end{itemize}
                    \smallskip
                    \centerline{\rule{0.995\textwidth}{.2pt}}
                     \smallskip
\begin{itemize}[noitemsep,leftmargin=2mm]
\item for all $t\in \terms_1$
\begin{itemize}[noitemsep,leftmargin=2mm]
\item \(t(\bar{W}_2) \le g_{\omega_2}(t,\lle)\), \quad \(t(\bar{W}_2) \geq g_{\omega_2}(t,\gge)\), and
\end{itemize}
\item   for all $t'\in \terms_2$ and for all $\psi_1,\psi_2 \in \cmsoset$
\begin{itemize}[noitemsep,leftmargin=2mm]
\item                    $\displaystyle\max_{\{\bar U_2 ~|~G_2 \models \neg \psi_2(\bar W_2 \bar U_2)\}} t'(\bar W_2 \bar U_2) \le g_{\omega_2}(t',\psi_2,\lle),$ and

\item                    $\displaystyle\min_{\{\bar U_2 ~|~G_2 \models \neg \psi_2(\bar W_2 \bar U_2)\}} t'(\bar W_2 \bar U_2) \ge g_{\omega_2}(t',\psi_2,\gge). $
\end{itemize}
\end{itemize}                    
\end{itemize}             
}

        This now finally matches the conjunction of the separate evaluation of the table formulas
        $\omega_1$ and $\omega_2$ on $G_1$ and $G_2$, respectively.

		\mybox{
			\begin{itemize}
                \item there exist $(\omega_1,\omega_2) \in \FVplus(\omega)$ such that
\begin{itemize}[noitemsep,leftmargin=2mm]
                    
                \item $G_1 \models \omega_1(\bar W_1)$ and

                \item $G_2 \models \omega_2(\bar W_2)$.
\end{itemize}
			\end{itemize}
		}

	\end{proof}
	
		\subsection{Granular Composition}
		
		Assuming, as we later will for our original input, that all thresholds and coefficients occurring in weight comparisons of our formula are integers, (a finitary restriction of) \Cref{thm:fv} could be used to compute the required dynamic programming records for the exact evaluation of \blockMSOone-formulas at disjoint union nodes,
        leading to an \XP-algorithm.
		However, the \((1+\eps)\)-approximate evaluation of \blockMSOone-queries will require the consideration of non-integral table formulas.
		This is where \emph{granularity} becomes important.
        Remember that a number \(q \in \Q\) has \emph{granularity} \gran\ if it can be written as $\frac{a}{\text{\gran}}$ where \(a,\gran \in \N\).
		\begin{definition}[Granular Table]\label{def:granulartable}
        Let $\sigma$ be a \CMSO$_1$-signature, $q \in \N$ be a quantifier rank, $\bar X,\bar Y$ be tuples of set variables, 
        \(\gran \in \N\) be a granularity,
        $\terms_1$ and \(\terms_2\) be a set of \gran-granular weight terms over $\bar X$, and $\bar X\bar Y$, respectively.
			We define a \emph{\gran-granular table} $T_{\text{\gran}}(\sigma,q,\bar X,\bar Y,\terms_1,\terms_2)$ to be the set of all \emph{\gran-granular table formulas} $\omega(\bar X)$ of the form
			\begin{multline*}
				\omega(\bar X) :=
				\Bigl(\bigwedge_{t \in \terms_1} t(\bar X) \le g_{t,\lle} \land t(\bar X) \ge g_{t,\gge} \Bigr) ~ \land~ \\
				\Bigl( \bigwedge_{t \in \terms_2} \bigwedge_{\psi \in \cmsoset} \forall \bar Y \psi(\bar X \bar Y) \lor \bigl( t(\bar X \bar Y) \le g_{t,\psi,\lle} \land t(\bar X \bar Y) \ge g_{t,\psi,\gge} \bigr) \Bigr),
			\end{multline*}
			where $g_{t_1,\lle}, g_{t_1,\gge}, g_{t_2,\psi,\lle}, g_{t_2,\psi,\gge}$ are numbers with granularity \gran for all $t_1 \in \terms_1$, $t_2 \in \terms_2$ and $\psi \in \cmsoset$.
		\end{definition}
		Note that in the case of granularity \(\gran = 1\), \(T_{\text{\gran}}(\sigma,q,\bar X,\bar Y,\terms_1,\terms_2) = T(\sigma,q,\bar X,\bar Y,\terms_1,\terms_2)\).
		We also extend the definition of the Feferman--Vaught set for granular tables.
		\begin{definition}[Granular Extended Feferman--Vaught Set]
		\label{def:granularFV}
			Let $T:= T_{\text{\gran}}(\sigma,q,\bar X,\bar Y,\terms_1,\terms_2)$ be a \gran-granular table and $\omega(\bar X)\in T$ be a \gran-granular table formula. Then let $\FVplus_{\text{\gran}}(\omega) \subseteq T \times T$ be the set of all tuples $(\omega_1,\omega_2)$ such that
		\begin{itemize}
			\item 
            for all \(t \in \terms_1\), we have \\
            \(\thres_{\omega_1}(t,\lle) + \thres_{\omega_2}(t,\lle) \leq \thres_{\omega}(t,\lle)\), and \\
			\(\thres_{\omega_1}(t,\gge) + \thres_{\omega_2}(t,\gge) \geq \thres_{\omega}(t,\gge)\), and
        \item for all $t \in \terms_2$, $\psi \in \cmsoset$, and $(\psi_1,\psi_2) \in \negFV(\psi)$, we have \\
            \({\thres_{\omega_1}(t,\psi,\lle)} + \thres_{\omega_2}(t,\psi,\lle) \le \thres_{\omega}(t,\psi,\lle)\), and \\
            \(\thres_{\omega_1}(t,\psi,\gge) + \thres_{\omega_2}(t,\psi,\gge) \geq {\thres_{\omega}(t,\psi,\gge)}\).
		\end{itemize}
		\end{definition}
		By scaling all thresholds and coefficients of weight comparisons by a factor of \(\gran\), applying \Cref{thm:base}, and scaling back by \(\frac{1}{\gran}\), we can immediately formulate the following generalization of the theorem.

		\begin{corollary}
		\label{cor:granularFV}
			Let $T := T_{\text{\gran}}(\sigma,q,\bar X,\bar Y,\terms_1,\terms_2)$. 
			For every $\omega(\bar X) \in T$, the set $\FVplus_\gran(\omega)$ satisfies the following property: for all graphs $G_1$, $G_2$ and $\bar W_1 \in \PP(V(G_1))^{|\bar X|}$, $\bar W_2 \in \PP(V(G_2))^{|\bar X|}$,
			$$
			G_1 \oplus G_2 \models \omega(\bar W_1 \cup \bar W_2) \iff 
			G_1 \models \omega_1(\bar W_1) \text{ and } G_2 \models \omega_2(\bar W_2) \text{ for some } (\omega_1,\omega_2) \in \FVplus_\gran(\omega).
			$$
		\end{corollary}
	
We conclude this section with a useful observation that will later allow us to bound the thresholds pairs of formulas from extended Feferman--Vaught sets.

\begin{observation}
\label{obs:cutoff}
Let $T := T_{\text{\gran}}(\sigma,q,\bar X,\bar Y,\terms_1,\terms_2)$, $\omega(\bar X)\in T$ and $(\omega_1,\omega_2)\in \FVplus_\gran(\omega)$. For $i\in \{1,2\}$, let $\omega_i'$ be the granular table formula obtained from $\omega_i$ by upper-bounding each threshold in $\omega_i$ by the respective threshold in $\omega$, that is, $\thres_{\omega'_i}(t,\psi,\gge) := \min(\thres_{\omega_i}(t,\psi,\gge), \thres_{\omega}(t,\psi,\gge))$ (and analogously for $\thres_{\omega'_i}(t,\gge)$). Then
\begin{enumerate}
\item $(\omega'_1,\omega'_2)\in \FVplus_\gran(\omega)$, and
\item if $G_i\models \omega_i(\bar W_i)$, then also $G_i\models \omega_i'(\bar W_i)$.
\end{enumerate}
\end{observation}

\begin{proof}
The first point follows immediately from Definition~\ref{def:granularFV}. The second point follows from the definition of table formulas and the fact that each weight comparison in $\omega'_i$ is no harder to satisfy than the corresponding weight comparison in $\omega_i$.
\end{proof}

\section{Approximate Extended Feferman--Vaught}\label{sec:approxFV}

While the definition of table formulas serves as a crucial concept underlying our approach, their infinite nature naturally precludes their direct use in algorithms. Here, we will define an approximate and finite version of the table formulas which will be the objects used in our algorithms.

\subsection{Definition of Approximate Tables}\label{sec:approxtabd}

Intuitively, our aim will be to use thresholds that are non-negative rationals whose denominators will always be the granularity $\gran$ introduced earlier, and whose nominators will be numbers up to some non-zero limit $N \in \N$ with an ``\emph{accuracy}'' \(\alpha := 1 + \frac{1}{\alphabot}\) for \(b \in \N\).
When computing  $(1 + \eps)$-approximations to \blockMSOone-queries,
the number \(b\) will later be roughly equal to \(\log(|V(G)|)/\eps\).
To this end, we define 
\[
\roundedN[\alpha,N,\gran] := \Bigl([0,N] \cap \bigl\{\lceil \alpha^j \rceil/\gran \mid j \in \N\bigr\}\Bigr) 
\cup \Bigl([0,N] \cap \bigl\{\lfloor \alpha^j \rfloor/\gran \mid j \in \N\bigr\}\Bigr) \cup \{0,N\}.
\]
\begin{observation}\label{obs:tildeNsize}
    $
    |\roundedN[\alpha,N,\gran]| \le \bigoh\bigl(\log_{\alpha}(\gran \cdot N)\bigr).
    $
\end{observation}

For each rational number with denominator $\gran$, we can find a sufficiently close upper and lower bound in $\roundedN$. More precisely and stated for general values of our parameters:

    \begin{observation}\label{obs:sandwich}
        Let $\alpha>1$ be a rational number and $N,\gran\in \N$.        
        For every rational number $0 \le m \le N$ with granularity $\gran$, 
        there exists $\tilde m \in \roundedN[\alpha,N,\gran]$ 
        such that  $m \le \tilde m \le m \cdot \alpha$.
    \end{observation}
    \begin{proof}    
Let $j$ be the smallest integer such that $\frac{\alpha^j}{\gran}\geq m$.
        Then $m \le \frac{\alpha^j}{\gran} \le m \cdot \alpha$.
        If $m \cdot \alpha > N$, we set $\tilde m := N \in \roundedN[\alpha,N,\gran]$.
        Otherwise, $\frac{\alpha^j}{\gran} \le N$ and thus 
        $\tilde m := \frac{\lfloor \alpha^j \rfloor}{\gran} \in \roundedN[\alpha,N,\gran]$.
        Since the nominator of $m$ is an integer and its denominator is $\gran$, $m \le \frac{\alpha^j}{\gran}$ implies         $m \le \frac{\lfloor \alpha^j \rfloor}{\gran}$. We conclude that  $m \le \tilde m \le m \cdot \alpha$, as desired.
    \end{proof}
 
We can now proceed to the definition of the approximate tables. 

\begin{definition}[Approximate Table]
    Let $\sigma$ be a \CMSO$_1$-signature, $q \in \N$ be a quantifier rank, $\bar X,\bar Y$ be tuples of set variables, 
    $\terms_1$ be a set of weight terms over $\bar X$, and $\terms_2$ be a set of weight terms over $\bar X\bar Y$.
    Let us further fix a granularity $\gran\in \N$, a non-zero limit $N \in \N$ and an accuracy of $\alpha > 1$.

    The \emph{approximate table} $\tilde T_{\alpha,N,\gran}(\sigma,q,\bar X,\bar Y,\terms_1,\terms_2)$
    is defined as the subset of all approximate table formulas $\omega(\bar X) \in T_\gran(\sigma,q,\bar X,\bar Y,\terms_1,\terms_2)$ where the thresholds of all weight comparisons are in $\roundedN$. In other words, a table formula $\omega(\bar X) \in T_\gran(\sigma,q,\bar X,\bar Y,\terms_1,\terms_2)$ is in $\tilde T_{\alpha,N,\gran}(\sigma,q,\bar X,\bar Y,\terms_1,\terms_2)$ if and only if 
    \[
        \thres_\omega(t_1,\lle), \thres_\omega(t_1,\gge), \thres_\omega(t_2,\psi,\lle), \thres_\omega(t_2,\psi,\gge) \in \roundedN,
    \]
    for all $t_1 \in \terms_1$, $t_2 \in \terms_2$, $\psi \in \cmsoset$.

\end{definition}

Crucially, the size of these approximate tables is by definition bounded as follows.
\begin{observation}
    \label{obs:approxSize}
    $
    |\tilde T_{\alpha,N,\gran}(\sigma,q,\bar X,\bar Y,\terms_1,\terms_2)| \le |\roundedN|^{|\terms_1| + |\terms_2|\cdot|\cmsoset|}.
    $
\end{observation}

	\subsection{Composition}	
Similarly to restricting the thresholds of table formulas which we will consider to be values from \(\roundedN\), we can do the same for the pairs of formulas in their Feferman--Vaught sets.
Note that for us to be able to restrict to formulas with thresholds of a
certain granularity \(\gran\) in the granular Feferman--Vaught set of an
approximate table formula \(\omega\), by \Cref{def:granulartable}, the granularity of all terms in
\(\omega\) must be at most \(\gran\).
\begin{definition}[Approximate Extended Feferman--Vaught Set]
    \label{def:approxfv}
    Consider a table formula $\omega(\bar X)\in \tilde T_{\alpha,N,\gran'}(\sigma,q,\bar X,\bar Y,\terms_1,\terms_2)$.
    For \(\alpha > 1\), \(N \in \N\) and \(\gamma \in \N\) with \(\gamma' \leq \gamma\), we define the \emph{rounded Feferman--Vaught set}
    \[\roundedFVplus_{\alpha,N,\gran} (\omega) := \FVplus_{\gran}(\omega) \cap \bigl(\tilde{T}_{\alpha,N,\gran}(\sigma,q,\bar X,\bar Y,\terms_1,\terms_2)\bigr)^2,\]
    that is, as the Feferman--Vaught pairs of table formulas where all thresholds are in \(\roundedN[\alpha,N,\gran]\).
\end{definition}

Below we establish an approximate version of \Cref{thm:base} (and \Cref{cor:granularFV}) for approximate table formulas. 
It will be useful to recall the definition of
\(\alpha\)-undersatisfying ($\umodels$) and \(\alpha\)-oversatisfying
($\omodels$) a formula \(\phi\): 
Intuitively, the former means that \(\phi\) itself may not be satisfied but the $\alpha$-loosened version \(\phi_\alpha\) is, 
while the latter means that not only \(\phi\), but even its \(\alpha\)-tightened version \(\phi^\alpha\) is satisfied.
The following statement roughly says: If we can evaluate formulas on \(G_1\) and \(G_2\) with an ``accuracy'' of \((1+\frac{1}{b})^s\),
then on the union \(G_1 \cup G_2\), we can evaluate formulas with the slightly less precise ``accuracy'' of \((1+\frac{1}{b})^{s+1}\),
where one should think of \(s\) as the current depth of the cliquewidth expression.

As will become apparent in the proof, we will need to consider
the rounded Feferman--Vaught set of a larger granularity---in particular,
the granularity will depend on the accuracy of the approximate table
formula.

\begin{theorem}\label{thm:baseApprox2}
    Let $\alphabot \in \N$ and fix an arbitrary \(s \in \mathbb{N}\) and an approximate table $\tilde T := \tilde T_{1 + \frac{1}{\alphabot},N,\gran}(\sigma,q,\bar X,\bar Y,\terms_1,\terms_2)$.
    For every $\omega(\bar X) \in \tilde{T}$, the set $\roundedFVplus(\omega) := \roundedFVplus_{1 + \frac{1}{\alphabot},N,\gran\cdot (\alphabot(\alphabot+1))^{4s+4}}(\omega)$
    satisfies the following property. 
    For all graphs $G_1$, $G_2$ 
    and all $\bar W_1 \in \PP(V(G_1))^{|\bar X|}$, $\bar W_2 \in \PP(V(G_2))^{|\bar X|}$
    \begin{multline*}
        G_1 \cup G_2 \omodels[(1 + \frac{1}{\alphabot})^{s + 1}] \omega(\bar{W}_1 \cup \bar{W}_2) \implies\\
        G_1 \omodels[(1 + \frac{1}{\alphabot})^{s}] \omega_1(\bar{W}_1) \mbox{ and } G_2 \omodels[(1 + \frac{1}{\alphabot})^{s}] \omega_2(\bar{W}_2) \mbox{ for some } (\omega_1,\omega_2) \in \roundedFVplus(\omega),
    \end{multline*}
    \begin{multline*}
        G_1 \cup G_2 \umodels[(1 + \frac{1}{\alphabot})^s] \omega(\bar{W}_1 \cup \bar{W}_2) \Longleftarrow \\
        G_1 \umodels[(1 + \frac{1}{\alphabot})^{s}] \omega_1(\bar{W}_1) \mbox{ and } G_2 \umodels[(1 + \frac{1}{\alphabot})^{s}] \omega_2(\bar{W}_2) \mbox{ for some } (\omega_1,\omega_2) \in \roundedFVplus(\omega).
    \end{multline*}
\end{theorem}
\begin{proof}~
    \paragraph{First Implication.}
    Let $\omega(\bar X) \in \tilde T$, and assume $G_1 \cup G_2 \omodels[(1 + \frac{1}{\alphabot})^{s + 1}] \omega(\bar{W}_1 \cup \bar{W}_2)$.
Recalling the notation from Subsection~\ref{sub:approx}, the latter is equivalent to $G_1 \cup G_2 \models \omega^{(1 + \frac{1}{\alphabot})^{s + 1}}(\bar{W}_1 \cup \bar{W}_2)$, where $\omega^{(1 + \frac{1}{\alphabot})^{s + 1}}$ is a formula obtained from $\omega$ by multiplying and/or dividing all weight terms by $(\frac{b+1}{b})^{s + 1}$. By shifting these factors to the sides of weight comparisons containing the threshold, we see that 
every threshold of the form $\thres_{\omega^{(\frac{b+1}{b})^{s + 1}}}(\circ,\smallsquare,\lle)$ is obtained as $\thres_{\omega}(\circ,\smallsquare,\lle)\cdot (\frac{b}{b+1})^{2s + 2}$, and
every threshold of the form $\thres_{\omega^{(\frac{b+1}{b})^{s + 1}}}(\circ,\smallsquare,\gge)$ is obtained as $\thres_{\omega}(\circ,\smallsquare,\gge)\cdot (\frac{b+1}{b})^{2s + 2}$, for every relevant choice of $\circ$ and $\smallsquare$. The analogous statement also holds for thresholds of the form $\thres_{\omega^{(\frac{b+1}{b})^{s + 1}}}(\circ,\gge)$ and $\thres_{\omega^{(\frac{b+1}{b})^{s + 1}}}(\circ,\lle)$, respectively.

Given the above, we observe that every threshold in the granular table formula $\omega^{(\frac{b+1}{b})^{s + 1}}$ has granularity either $\gran\cdot b^{2s+2}$ or $\gran\cdot (b+1)^{2s+2}$.
Thus, $\omega^{(\frac{b+1}{b})^{s + 1}}$ is contained in the granular table $T_{\gran\cdot (b(b+1))^{2s+2}}(\sigma,q,\bar X,\bar Y,\terms_1,\terms_2)$.
Since we assume $G_1 \cup G_2 \models \omega^{(1 + \frac{1}{\alphabot})^{s + 1}}(\bar{W}_1 \cup \bar{W}_2)$,
we can now apply the granular version of the Feferman--Vaught Theorem given in Corollary~\ref{cor:granularFV} to obtain a pair of formulas $(\delta'_1,\delta'_2)$ of granularity $\gran\cdot (b(b+1))^{2s + 2}$ such that $G_1\models \delta'_1(\bar W_1)$ and $G_2\models \delta'_2(\bar W_2)$. We follow this up by applying Observation~\ref{obs:cutoff} to obtain a new pair of formulas $(\delta_1,\delta_2)$ with the same properties, but where additionally each threshold is upper-bounded by the corresponding threshold in $\omega^{(\frac{b+1}{b})^{s + 1}}$.
In particular, for each $i\in \{1,2\}$, we have $\thres_{\delta_i}(\circ,\smallsquare,\lle) \leq \thres_{\omega}(\circ,\smallsquare,\lle)\cdot (\frac{b}{b+1})^{2s + 2}$ and analogously $\thres_{\delta_i}(\circ,\smallsquare,\gge) \leq \thres_{\omega}(\circ,\smallsquare,\gge)\cdot (\frac{b+1}{b})^{2s + 2}$ (and the same holds for thresholds of the form $\thres_{\delta_i}(\circ,\lle)$ and $\thres_{\delta_i}(\circ,\gge)$).

Our next task will be to, intuitively, ``loosen'' the formulas $\delta_1$ and $\delta_2$. In particular, for $i\in \{1,2\}$ we define the granular table formula $\zeta_i$ by setting the thresholds (for each relevant choice of $\circ$ and \smallsquare) as 
\begin{itemize}
\item $\thres_{\zeta_i}(\circ,\smallsquare,\lle):= \thres_{\delta_i}(\circ,\smallsquare,\lle) \cdot (\frac{b+1}{b})^{2s + 2}$, and
\item $\thres_{\zeta_i}(\circ,\smallsquare,\gge):= \thres_{\delta_i}(\circ,\smallsquare,\gge) \cdot (\frac{b}{b+1})^{2s + 2}$,
\end{itemize}
where the thresholds of the form $\thres_{\omega_i}(\circ,\lle)$ and $\thres_{\omega_i}(\circ,\gge)$ are once again---as well as in the remainder of the proof---treated as degenerate special cases of the triples considered above.
Note that it is easier to satisfy \(\zeta_i\) than \(\delta_i\).
Since the granularity of all thresholds in $\delta_1$ and $\delta_2$ was $\gran\cdot (b(b+1))^{2s + 2}$, the granularity of all thresholds in $\zeta_1$ and $\zeta_2$ can be uniformly set to $\gran\cdot (b(b+1))^{4s + 4}$. 
Moreover, we observe that because each threshold in $\omega$ was upper-bounded by $N$ and  the bounds on the thresholds in $\delta_i$ guarantee that each threshold in $\zeta_i$, $i\in \{1,2\}$ is upper-bounded by $N$ (that is, the applied multiplications ``cancel out'').

As our final step towards the construction of the sought-after pair of formulas, we apply Observation~\ref{obs:sandwich}
to identify, for each $\thres_{\zeta_i}(\circ,\smallsquare,\lle)$ and $\thres_{\zeta_i}(\circ,\smallsquare,\gge)$, $i\in \{1,2\}$, rational numbers $\thres_{\omega_i}(\circ,\smallsquare,\lle), \thres_{\omega_i}(\circ,\smallsquare,\gge) \in \roundedN[\frac{b+1}{b},N,\gran\cdot (b(b+1))^{4s + 4}]$ such that 
\begin{itemize}
\item $\thres_{\zeta_i}(\circ,\smallsquare,\lle)\cdot \frac{b}{b+1}\leq \thres_{\omega_i}(\circ,\smallsquare,\lle) \leq 
\thres_{\zeta_i}(\circ,\smallsquare,\lle)$, and 
\item $\thres_{\zeta_i}(\circ,\smallsquare,\gge)
\leq \thres_{\omega_i}(\circ,\smallsquare,\gge) \leq 
\thres_{\zeta_i}(\circ,\smallsquare,\gge)\cdot \frac{b+1}{b}$.
\end{itemize}

Let $\omega_i$, $i\in \{1,2\}$ be the granular table formula defined by the thresholds $\thres_{\omega_i}(\circ,\smallsquare,\lle)$ and $\thres_{\omega_i}(\circ,\smallsquare,\gge)$ identified in the previous paragraph.
In particular, each $\omega_i$ is an approximate table formula in $\tilde{T}_{1 + \frac{1}{\alphabot},N,\gran\cdot (b(b+1))^{4s + 4}}(\sigma,q,\bar X,\bar Y,\terms_1,\terms_2)$, and may be intuitively viewed as a ``slight tightening of $\zeta_i$ with rounded numbers''.
To complete the proof, we will now argue that (A) $(\omega_1,\omega_2) \in \roundedFVplus(\omega)$ and (B) $G_i \models \omega_i^{(\frac{b+1}{b})^s}(\bar{W}_i)$ for $i\in \{1,2\}$. Before we proceed with the formal arguments, we remark that---on an intuitive level---(A) holds because $\omega_i$ is a tightening of $\zeta_i$ whereas $(\zeta_1,\zeta_2)$ is in the Feferman--Vaught set of $\omega$ due to how it was obtained from $(\delta_1,\delta_2)$, while (B) holds because $\omega_i$ is only a ``slight'' tightening of $\zeta_i$.

Towards establishing (A), we recall from Definition~\ref{def:approxfv} that 
$\roundedFVplus(\omega) = \FVplus_{\gran\cdot (b(b+1))^{4s + 4}}(\omega) \cap \bigl(\tilde{T}_{1 + \frac{1}{\alphabot},N,\gran \cdot (b(b+1))^{4s + 4}}(\sigma,q,\bar X,\bar Y,\terms_1,\terms_2)\bigr)^2$. Hence, it suffices to prove $(\omega_1,\omega_2)\in \FVplus_{\gran\cdot (b(b+1))^{4s + 4}}(\omega)$. By Definition~\ref{def:FVplus}, this is equivalent to showing
\begin{itemize}
\item $\thres_{\omega_1}(\circ, \smallsquare, \lle)+ \thres_{\omega_2}(\circ, \smallsquare, \lle) \leq \thres_{\omega}(\circ, \smallsquare, \lle)$, and 
\item $\thres_{\omega_1}(\circ, \smallsquare, \gge)+ \thres_{\omega_2}(\circ, \smallsquare, \gge) \geq \thres_{\omega}(\circ, \smallsquare, \gge)$,
\end{itemize}
once again for all suitable choices of $\circ$ and $\smallsquare$. 
We start by establishing the first inequality.
Since $(\delta_1,\delta_2)$ was obtained from $\omega^{(\frac{b+1}{b})^{s + 1}}$, by applying Corollary~\ref{cor:granularFV} we have
\[
    \thres_{\delta_1}(\circ,\smallsquare,\lle) + \thres_{\delta_2}(\circ,\smallsquare,\lle) \leq  \thres_{\omega}(\circ,\smallsquare,\lle) \cdot \left(\frac{b}{b+1}\right)^{2s + 2}.
\]
Moreover, by the definition of $\zeta_i$, $\thres_{\zeta_i}(\circ,\smallsquare,\lle)= \thres_{\delta_i}(\circ,\smallsquare,\lle) \cdot (\frac{b+1}{b})^{2s + 2}$.
Therefore,
\[
    \thres_{\zeta_1}(\circ,\smallsquare,\lle) \cdot \left(\frac{b}{b+1}\right)^{2s + 2} + \thres_{\zeta_2}(\circ,\smallsquare,\lle) \cdot \left(\frac{b}{b+1}\right)^{2s + 2} \leq  \thres_{\omega}(\circ,\smallsquare,\lle) \cdot \left(\frac{b}{b+1}\right)^{2s + 2},
\]
\[
    \thres_{\zeta_1}(\circ,\smallsquare,\lle)+ \thres_{\zeta_2}(\circ,\smallsquare,\lle) \leq  \thres_{\omega}(\circ,\smallsquare,\lle).
\]
Recall that $\thres_{\omega_i}(\circ, \smallsquare, \lle) \leq \thres_{\zeta_i}(\circ,\smallsquare,\lle)$.
It hence follows that
\[
    \thres_{\omega_1}(\circ,\smallsquare,\lle) + \thres_{\omega_2}(\circ,\smallsquare,\lle) \leq \thres_{\omega}(\circ,\smallsquare,\lle).
\]

For the second inequality, we recall that 
$\thres_{\zeta_i}(\circ,\smallsquare,\gge)= \thres_{\delta_i}(\circ,\smallsquare,\gge) \cdot (\frac{b}{b+1})^{2s + 2}$. Recalling that $(\delta_1,\delta_2)$ was obtained from $\omega^{(\frac{b+1}{b})^{2s + 2}}$ by applying Corollary~\ref{cor:granularFV}, we obtain similarly
	\begin{align*}
\thres_{\delta_1}(\circ,\smallsquare,\gge) + \thres_{\delta_2}(\circ,\smallsquare,\gge) & \geq \thres_{\omega}(\circ,\smallsquare,\gge) \cdot \left(\frac{b+1}{b}\right)^{2s + 2},\\
\thres_{\zeta_1}(\circ,\smallsquare,\gge) \cdot \left(\frac{b+1}{b}\right)^{2s + 2} + \thres_{\zeta_2}(\circ,\smallsquare,\ge) \cdot \left(\frac{b+1}{b}\right)^{2s + 2} & \geq \thres_{\omega}(\circ,\smallsquare,\gge) \cdot \left(\frac{b+1}{b}\right)^{2s + 2},\\
\thres_{\zeta_1}(\circ,\smallsquare,\ge) + \thres_{\zeta_2}(\circ,\smallsquare,\ge) & \geq \thres_{\omega}(\circ,\smallsquare,\gge).
\end{align*}
Since $\thres_{\omega_i}(\circ, \smallsquare, \gge) \geq \thres_{\zeta_i}(\circ,\smallsquare,\gge)$, we conclude that
\[
    \thres_{\omega_1}(\circ,\smallsquare,\lle) + \thres_{\omega_2}(\circ,\smallsquare,\lle) \geq \thres_{\omega}(\circ,\smallsquare,\lle).
\]
This establishes (A).

Next, we turn towards (B). We begin by recalling that $G_i\models \delta_i(\bar W_i)$. 
Here, our aim will be to show that for every variable assignment to $\bar X\bar Y$ in $G_i$, the subformula 
$$t(\bar X\bar Y)\leq \thres_{\delta_i}(\circ,\smallsquare,\lle) \wedge 
t(\bar X\bar Y)\geq \thres_{\delta_i}(\circ,\smallsquare,\gge)$$
of $\delta_i$ implies the corresponding subformula 
$$t(\bar X\bar Y)\leq \thres_{\omega_i^{(\frac{b+1}{b})^s}}(\circ,\smallsquare,\lle) \wedge 
t(\bar X\bar Y)\geq \thres_{\omega_i^{(\frac{b+1}{b})^s}}(\circ,\smallsquare,\gge)$$
of $\omega_i^{(\frac{b+1}{b})^s}$, where we once again assume all weight comparisons in $\omega_i^{(\frac{b+1}{b})^s}$ to have the factors containing $b$ shifted to the threshold. Given the fact that these subformulas (1) only occur positively in a table formula and (2) form the only difference between $\omega_i^{(\frac{b+1}{b})^s}$ and $\delta_i$, showing the above implication would yield $G_i\models \omega_i^{(\frac{b+1}{b})^s}(\bar W_i)$.

By the construction of $\zeta_i$, we have $\thres_{\delta_i}(\circ,\smallsquare,\lle)= \thres_{\zeta_i}(\circ,\smallsquare,\lle) \cdot (\frac{b}{b+1})^{2s + 2}$. 
Moreover, since $\thres_{\zeta_i}(\circ,\smallsquare,\lle)\cdot \frac{b}{b+1}\leq \thres_{\omega_i}(\circ,\smallsquare,\lle) \leq \thres_{\zeta_i}(\circ,\smallsquare,\lle)$, we have that 
$$
\thres_{\zeta_i}(\circ,\smallsquare,\lle)\cdot \left(\frac{b}{b+1}\right)^{2s+2}\leq \thres_{\omega_i}(\circ,\smallsquare,\lle)\cdot \left(\frac{b}{b+1}\right)^{2s+1} \leq \thres_{\zeta_i}(\circ,\smallsquare,\lle) \cdot \left(\frac{b}{b+1}\right)^{2s+1}.
$$
The first inequality guarantees that whenever $t(\bar X\bar Y)\leq \thres_{\delta_i}(\circ,\smallsquare,\lle)=\thres_{\zeta_i}(\circ,\smallsquare,\lle)\cdot (\frac{b}{b+1})^{2s+2}$ holds, also $t(\bar X\bar Y)\leq \thres_{\omega_i}(\circ,\smallsquare,\lle)\cdot \left(\frac{b}{b+1}\right)^{2s+1}$ holds. In particular,
this guarantees $t(\bar X\bar Y)\leq \thres_{\omega_i^{(\frac{b+1}{b})^s}}(\circ,\smallsquare,\lle)= \thres_{\omega_i}(\circ,\smallsquare,\lle)\cdot (\frac{b}{b+1})^{2s}$.

As for the other inequality in the subformula, there we have $\thres_{\delta_i}(\circ,\smallsquare,\gge)= \thres_{\zeta_i}(\circ,\smallsquare,\gge) \cdot (\frac{b+1}{b})^{2s + 2}$. 
Moreover, $\thres_{\zeta_i}(\circ,\smallsquare,\gge) \leq
\thres_{\omega_i}(\circ,\smallsquare,\gge) \leq 
\thres_{\zeta_i}(\circ,\smallsquare,\gge)\cdot \frac{b+1}{b}$ can be rewritten as
$$
\thres_{\zeta_i}(\circ,\smallsquare,\gge)\cdot \left(\frac{b+1}{b}\right)^{2s+1} \leq
\thres_{\omega_i}(\circ,\smallsquare,\gge) \cdot \left(\frac{b+1}{b}\right)^{2s+1} \leq 
\thres_{\zeta_i}(\circ,\smallsquare,\gge)\cdot \left(\frac{b+1}{b}\right)^{2s+2}.
$$
Here, the second inequality guarantees that whenever $t(\bar X\bar Y)\geq \thres_{\delta_i}(\circ,\smallsquare,\gge)=\thres_{\zeta_i}(\circ,\smallsquare,\lle) \cdot (\frac{b+1}{b})^{2s + 2}$ holds, so does $t(\bar X\bar Y)\geq 
\thres_{\omega_i}(\circ,\smallsquare,\gge) \cdot (\frac{b+1}{b})^{2s+1} \geq
\thres_{\omega_i}(\circ,\smallsquare,\gge) \cdot (\frac{b+1}{b})^{2s}=
\thres_{\omega_i^{(\frac{b+1}{b})^s}}(\circ,\smallsquare,\gge)
$.
This completes the argument for why $G_i\models \delta_i(\bar W_i)$ implies $G_i\models \omega_i^{(\frac{b+1}{b})^s}(\bar W_i)$. Hence, property (B) holds as well and the first implication holds.

\paragraph{Second Implication.}
Consider some $(\omega_1, \omega_2) \in \roundedFVplus(\omega)$ such that 
$G_i\models (\omega_i)_{(1 + \frac{1}{\alphabot})^s}(\bar W_i)$, $i\in \{1,2\}$. Note that $(\omega_1, \omega_2) \in \roundedFVplus(\omega)$ implies $((\omega_1)_{(1 + \frac{1}{\alphabot})^s}, (\omega_2)_{(1 + \frac{1}{\alphabot})^s}) \in \FVplus_{\gran'}(\omega_{(1 + \frac{1}{\alphabot})^s})$ for some $\gran'$, since all the inequalities in Definition~\ref{def:FVplus} hold. The proof now follows directly from the backward implication of Corollary~\ref{cor:granularFV}.
\end{proof}
	
	\section{Table Computation}	
	\label{sec:tablecomp}
	Our results from the previous section will allow us to aggregate records via dynamic programming along a cliquewidth-expression.
    Intuitively, our records consist of a mapping which assigns, to each table formula, a witness in the graph that (1) is guaranteed to at least satisfy the loosened formula, and (2) is at least as ``good'' in terms of the optimization target as the optimal witness for the tightened formula.
    By applying the following theorem (twice, to get a conservative and an eager estimate) with tables derived from the input query, we will obtain our approximate solutions.
    
    \begin{theorem}\label{thm:computeTable}
        Given \begin{itemize}
        \item $\tilde T := \tilde T_{1 + \frac{1}{\alphabot},N,\gran}(\sigma,q,\bar X,\bar Y,\terms_1,\terms_2)$ for positive $\alphabot, N, \gran\in \N$,
        \item an optimization target \(\target(\bar X)\), and 
        \item a graph $G$ with a corresponding \(k\)-expression $\chi$ of depth \(d\),
        \end{itemize}
        one can compute a function $\witness_G \colon \tilde T \to \PP(V(G))^{|\bar{X}|} \cup \{\bot\}$
    	such that for all \(\omega \in \tilde T\),
		\begin{itemize}
            \item \(\textnormal{oversatisfy}(G, \omega, \target, (1 + \frac{1}{\alphabot})^d) \le \target(\witness_G(\omega))\),
            \item if \(\witness_G(\omega) \neq \bot\) then \(G \umodels[(1 + \frac{1}{\alphabot})^d] \omega(\witness_G(\omega))\),
		\end{itemize}
        where we define $\target(\bot) := -\infty$.
        Moreover, there exists a computable function $f(|\sigma|,k,q,|\bar X\bar Y|)$ such that the running time of the above algorithm is bounded by
        $$
        |\chi| \cdot f(|\sigma|,k,q,|\bar X\bar Y|) \cdot |\roundedN[1 + \frac{1}{\alphabot},N,\gran \cdot (b(b+1))^{4d^2+4d}]|^{3|\terms_1| + |\terms_2|\cdot f(|\sigma|,k,q,|\bar X\bar Y|)}.
        $$
	\end{theorem}
	\begin{proof}
        As it is common for this kind of proof, we consider the labels of the $k$-expression as additional colors $P_1,\dots,P_k$ of the graph itself.
        Thus, if the signature \(\sigma\) originally allows the graph $G$ to have $c$ colors, we will treat it a graph with $c+k$ colors, and include the extra colors in an extended signature \(\sigma'\),
        that is, \(\sigma' := \sigma \cup \{P_1, \dotsc, P_k\}\).

        For \(\ell \le d\), we define a function \(\gran(\ell) := \gran \cdot (b(b+1))^{(4d+4)(d-\ell)}\).
        Given a graph $H$ with $c+k$ colors together with a corresponding \(k\)-expression of depth \(\ell \le d\),
        we compute a function $\witness_H \colon \tilde T_{1 + \frac{1}{\alphabot},N,\gran(\ell)}(\sigma',q,\bar X,\bar Y,\terms_1,\terms_2) \to \PP(V(H))^{|\bar{X}|} \cup \{\bot\}$
    	such that for \(\omega \in \tilde T\),
		\begin{itemize}
            \item \(\textnormal{oversatisfy}(H, \omega, \target, (1 + \frac{1}{\alphabot})^d) \le \target(\witness_H(\omega))\),
            \item if \(\witness_H(\omega) \neq \bot\) then \(H \umodels[(1 + \frac{1}{\alphabot})^d] \omega(\witness_H(\omega))\).
		\end{itemize}
        For \(\ell = d\), we have \(\gran(\ell) = \gran\).
        In this case, we can then simply restrict the domain of \(\witness_G\) to $\tilde T$ (that is, ignore formulas mentioning the \(k\) additional colors) to get our final output. Hence, for the rest of the proof we focus
        on obtaining \(\witness_H\) for this larger domain by induction on $\ell \le d$.
        In particular, as our inductive hypothesis we assume that the statement holds for all depths lower than the current depth $\ell$. 

        After covering the base case of \(H\) being a single-vertex graph,
        we assume \(H\) to be given by \(\rho_{i \to j}(H_1)\), \(\eta_{i,j}(H_1)\) or \(H_1 \oplus H_2\),
        where we assume the statement for \(k\)-expressions of smaller depth, in particular for \(H_1\) and \(H_2\), as inductive hypothesis.
        We will bound the running time of all cases at the end of the proof.
        
		\paragraph{Base Case.}
		For the base case, consider a single-vertex graph $H$.
		Here we 
        simply enumerate all formulas $\omega$ of the domain of \(\witness_H\) and compute the corresponding entry \(\witness_H(\omega)\) by brute force
        in time depending only on $\omega$.

		\paragraph{Relabeling.}
		Recall that \(\rho_{i \to j}\) changes all labels \(P_i\) to \(P_j\) in a \(k\)-expression of a graph.
		Consider a graph \(H = \rho_{i \to j}(H_1)\).
        Let $\omega$ be from the domain of \(\witness_H\) and let $\rho_{i \to j}(\omega)$ be the formula obtained from $\omega$ 
        by substituting all atoms $P_i(x)$ with ``false'' and substituting all atoms $P_j(x)$ with $P_j(x) \lor P_i(x)$.
        Then for every $\bar W \subseteq \PP(V(H))^{|\bar X|}$,
        $$
        H \models \omega(\bar W) \iff H_1 \models \rho_{i \to j}(\omega)(\bar W).
        $$
        Moreover, for each choice of $\ell \le d$, 
\begin{align*}
        &H \omodels[(1 + \frac{1}{\alphabot})^\ell] \omega(\bar W) \iff H_1 \omodels[(1 + \frac{1}{\alphabot})^\ell] \rho_{i \to j}(\omega)(\bar W)
        \quad \text{ and } \\
        &H \umodels[(1 + \frac{1}{\alphabot})^\ell] \omega(\bar W) \iff H_1 \umodels[(1 + \frac{1}{\alphabot})^\ell] \rho_{i \to j}(\omega)(\bar W).
\end{align*}
        This means $\textnormal{oversatisfy}(H, \omega, \target, (1 + \frac{1}{\alphabot})^\ell) = \textnormal{oversatisfy}(H_1, \rho_{i \to j}(\omega), \target, (1 + \frac{1}{\alphabot})^\ell)$.
        By the inductive hypothesis, we have already computed an entry
        $\witness_{H_1}(\rho_{i \to j}(\omega)) = \bar W$ such that
		\begin{itemize}
            \item \(\textnormal{oversatisfy}(H_1, \rho_{i \to j}(\omega), \target, (1 + \frac{1}{\alphabot})^{\ell-1}) \le \target(\bar W)\),
            \item if \(\bar W \neq \bot\) then \(H_1 \umodels[(1 + \frac{1}{\alphabot})^{\ell-1}] \rho_{i \to j}(\omega)(\bar{W})\).
		\end{itemize}
        Note that both statements still hold if we weaken them by replacing $(1 + \frac{1}{\alphabot})^{\ell-1}$ with $(1 + \frac{1}{\alphabot})^\ell$.
        By the above, we obtain:
		\begin{itemize}
            \item \(\textnormal{oversatisfy}(H, \omega, \target, (1 + \frac{1}{\alphabot})^\ell) = \textnormal{oversatisfy}(H_1, \rho_{i \to j}(\omega), \target, (1 + \frac{1}{\alphabot})^{\ell}) \le \target(\bar W)\),
            \item if \(\bar W \neq \bot\)  
                then \(H_1 \umodels[(1 + \frac{1}{\alphabot})^{\ell}] \rho_{i \to j}(\omega)(\bar{W})\), 
                which in turn implies \(H \umodels[(1 + \frac{1}{\alphabot})^\ell] \omega(\bar W)\).
		\end{itemize}
        We therefore obtain our entry $\witness_H(\omega)$ with all the required properties by setting
        $$
        \witness_H(\omega) := \bar W.
        $$
		
		\paragraph{Adding Edges.}
        Recall that the function $\eta_{i,j}$ adds edges between all vertices with label $P_i$ and all vertices with label $P_j$, with $i\neq j$.
		Consider a graph \(H = \rho_{i \to j}(H_1)\).
        Let $\eta_{i,j}(\omega)$ be the formula obtained from $\omega$ 
        by substituting all atoms $\relE(x,y)$ with $\relE(x,y) \lor (P_i(x) \land P_j(y)) \lor (P_i(x) \land P_j(y))$.
        We set
        $$
        \tabl_H(\omega) = \tabl_{H_1}(\eta_{i,j}(\omega)).
        $$
        The correctness of this follows using the same arguments as for the recoloring.
		
		\paragraph{Disjoint Union.}
		Consider a graph \(H = H_1 \oplus H_2\) and let $\ell \le d$ be the current depth of the $k$-expression tree.
        We invoke the inductive hypothesis for \(\ell-1\) to construct functions \(\witness_{H_1}\) and \(\witness_{H_2}\) such that for all $\omega' \in \tilde T_{1 + \frac{1}{\alphabot},N,\gran(\ell)}(\sigma',q,\bar X,\bar Y,\terms_1,\terms_2)$ and $i \in \{1,2\}$,
		\begin{itemize}
            \item \(\textnormal{oversatisfy}(H_i, \omega', \target, (1 + \frac{1}{\alphabot})^{\ell-1}) \le \target(\witness_{H_i}(\omega'))\),
            \item if \(\witness_{H_i}(\omega') \neq \bot\) then \(H_i \umodels[(1 + \frac{1}{\alphabot})^{\ell-1}] \omega'(\witness_{H_i}(\omega'))\).
		\end{itemize}
		Notice that \(\gran(\ell-1) = \gran(\ell) \cdot (b(b+1))^{(4d+4)}\).
        Thus, the change in the granularity of the formulas that we apply the induction hypothesis for matches the one incurred by considering the approximate Feferman--Vaught set of a formula in Theorem~\ref{thm:baseApprox2}.
        Fix $\omega \in \tilde T_{1 + \frac{1}{\alphabot},N,\gran(\ell)}(\sigma',q,\bar X,\bar Y,\terms_1,\terms_2)$ and choose $(\omega_1,\omega_2) \in \roundedFVplus_{1 + \frac{1}{\alphabot},N,\gran(\ell-1)}(\omega)$ such that
        \begin{equation}\label{eq:witness1}
        \target(\witness_{H_1}(\omega_1)) + \target(\witness_{H_2}(\omega_2)) = 
        \quad \quad \max_{\mathclap{(\omega_1',\omega_2') \in \roundedFVplus_{1 + \frac{1}{\alphabot},N,\gran(\ell-1)}(\omega)}} \quad \quad 
        \target(\witness_{H_1}(\omega_1')) + \target(\witness_{H_2}(\omega_2')), 
        \end{equation}
        where ties are broken arbitrarily.

        We define
        $$
            \witness_H(\omega) := \witness_{H_1}(\omega_1) \cup \witness_{H_2}(\omega_2),
        $$
        where $\bot$ propagates (as one would expect) by the rule $\bot \cup \witness_{H_2}(\omega_2) = \witness_{H_1}(\omega_1) \cup \bot = \bot$.
        This means
        \begin{equation}\label{eq:witness2}
            \target(\witness_H(\omega)) = \target(\witness_{H_1}(\omega_1)) + \target(\witness_{H_2}(\omega_2)).
        \end{equation}
		According to \Cref{thm:baseApprox2}, the set \(\roundedFVplus_{1 + \frac{1}{\alphabot},N,\gran(\ell-1)}(\omega)\) guarantees that for every $\bar W$,
		\begin{multline*}
            H \omodels[(1 + \frac{1}{\alphabot})^{\ell}] \omega(\bar{W}) \Longrightarrow {} 
            H_1 \omodels[(1 + \frac{1}{\alphabot})^{\ell-1}] \omega_1'(\bar{W}\cap V(H_1)) \text{ and }
			H_2 \omodels[(1 + \frac{1}{\alphabot})^{\ell-1}] \omega_2'(\bar{W}\cap V(H_2)) \\			
			\text{ for some } (\omega_1',\omega_2') \in \roundedFVplus_{1 + \frac{1}{\alphabot},N,\gran(\ell-1)}(\omega),
		\end{multline*}
		which implies
        \begin{multline*}
            \textnormal{oversatisfy}\left(H,\omega,\target,\left(1 + \frac{1}{\alphabot}\right)^\ell\right) \leq ~\max_{{(\omega_1',\omega_2') \in \roundedFVplus_{1 + \frac{1}{\alphabot},N,\gran(\ell-1)}(\omega)}}\quad \quad \\
        \left( \textnormal{oversatisfy}\left(H_1,\omega_1',\target,\left(1 + \frac{1}{\alphabot}\right)^{\ell-1}\right) + \textnormal{oversatisfy}\left(H_2,\omega_2',\target,\left(1 + \frac{1}{\alphabot}\right)^{\ell-1}\right)\right).
        \end{multline*}
        Since by the induction hypothesis, $\textnormal{oversatisfy}(H_i,\omega_i,\target,(1 + \frac{1}{\alphabot})^{\ell-1}) \le \target(\witness_{H_i}(\omega_i))$ for $i \in \{1,2\}$,
        this implies
        $$
            \textnormal{oversatisfy}\left(H,\omega,\target,\left(1 + \frac{1}{\alphabot}\right)^\ell\right) \leq \quad \quad \quad \quad
            \max_{\mathclap{(\omega_1',\omega_2') \in \roundedFVplus_{1 + \frac{1}{\alphabot},N,\gran(\ell-1)}(\omega)}} \quad \quad
            \left(\target(\witness_{H_1}(\omega_1')) + \target(\witness_{H_2}(\omega_2'))\right).
        $$
Recalling (\ref{eq:witness1}) and (\ref{eq:witness2}), the above inequality can be rewritten as
$$
            \textnormal{oversatisfy}\left(H,\omega,\target,\left(1 + \frac{1}{\alphabot}\right)^\ell\right) \leq \\
            \target(\witness_{H_1}(\omega_1)) + \target(\witness_{H_2}(\omega_2))
            = \target(\witness_H(\omega)).
$$

        Furthermore, if $\witness_H(\omega) \neq \bot$ then also $\witness_{H_1}(\omega_1) \neq \bot$ and $\witness_{H_2}(\omega_2) \neq \bot$.
        The induction hypothesis states that 
        \(H_1 \umodels[(1 + \frac{1}{\alphabot})^{\ell-1}] \omega_1(\witness_{H_1}(\omega_1))\) and 
        \(H_2 \umodels[(1+\frac{1}{\alphabot})^{\ell-1}] \omega_1(\witness_{H_1}(\omega_1))\).
        Since $(\omega_1,\omega_2) \in \roundedFVplus_{(1 + \frac{1}{\alphabot}),N,\gran(\ell-1)}(\omega)$,
        by the second implication of \Cref{thm:baseApprox2},
        \(H_1 \cup H_2 \umodels[(1 + \frac{1}{\alphabot})^{\ell-1}] \omega\bigl(\witness_{H_1}(\omega_1) \cup \witness_{H_1}(\omega_1)\bigr)\).
        In other words, \(H \umodels[(1 + \frac{1}{\alphabot})^{\ell-1}] \omega(\witness_{H}(\omega))\).
        Since increasing $\ell$ only weakens the statement,
        \(H \umodels[(1 + \frac{1}{\alphabot})^{\ell}] \omega(\witness_{H}(\omega))\) as desired.

        \paragraph{Running Time.}
        By \Cref{obs:msosizebound} and \Cref{thm:fv}, there exists a function $f_1(|\sigma|,k,q,|\bar X\bar Y|)$ such that one 
        can compute for every formula $\psi(\bar X) \in \textnormal{CMSO}_1(\sigma',q,|\bar X\bar Y|)$ the set $\FV(\psi)$
        in time at most $f_1(|\sigma|,k,q,|\bar X\bar Y|)$.
        Deciding whether $(\omega_1,\omega_2) \in \roundedFVplus_{1 + \frac{1}{\alphabot},N,\gran(\ell-1)}(\omega)$
        for some $\omega \in \tilde T_{1 + \frac{1}{\alphabot},N,\gran(\ell)}(\sigma',q,\bar X,\bar Y,\terms_1,\terms_2)$
        and \(\omega_1,\omega_2 \in \tilde T_{1 +
        \frac{1}{\alphabot},N,\gran(\ell-1)}(\sigma',q,\bar X,\bar
        Y,\terms_1,\terms_2)\)
        thus takes time $f_1(|\sigma|,k,q,|\bar X\bar Y|)$ plus the time required to perform the basic number checks arising from \Cref{def:FVplus} and \Cref{def:approxfv}.
        We can bound this by some function $f_2(|\sigma|,k,q,|\bar X\bar Y|)$.
        We can thus compute $\roundedFVplus_{1 + \frac{1}{\alphabot},N,\gran(\ell-1)}(\omega)$ for all 
        $\omega \in \tilde T_{1 + \frac{1}{\alphabot},N,\gran(\ell)}(\sigma',q,\bar X,\bar Y,\terms_1,\terms_2)$
        in time $f_2(|\sigma|,k,q,|\bar X\bar Y|) \cdot |T_{\alpha,N,\gran(\ell)}(|\sigma'|,q,\bar X,\bar Y,\terms_1,\terms_2)| \cdot |T_{\alpha,N,\gran(\ell-1)}(|\sigma'|,q,\bar X,\bar Y,\terms_1,\terms_2)|^2$.
        Note that \(\gran(\ell) \le (b(b+1))^{4d^2+4d}\) for all considered values \(\ell \le d\),
        and smaller values for the granularity can only decrease the size of our tables.
        All in all, there is a function $f_3(|\sigma|,k,q,|\bar X\bar Y|)$ such that the algorithm described in this proof computes for each labeled and colored graph $H$ obtained at some node of the $k$-expression, $\witness_H$ 
        in time
        $f_3(|\sigma|,k,q,|\bar X\bar Y|) \cdot |T_{1 + \frac{1}{\alphabot},N,\gran \cdot (b(b+1))^{4d^2+4d}}(\sigma',q,\bar X,\bar Y,\terms_1,\terms_2)|^3$ from
        \begin{enumerate}
            \item scratch (for single vertices), or
            \item the function \(\witness_{H_1}\) when $H$ is obtained from $H_1$ by relabeling or adding edges, or
            \item the functions \(\witness_{H_1}\) and \(\witness_{H_2}\) when $H$ is obtained as $H_1\cup H_2$.
        \end{enumerate}
       Indeed, for case $1.$, the computation is trivial, in case $2.$, the computation merely involves a ``reshuffling'' of \(\witness_{H_1}\), while case $3.$ can be handled by combining \(\witness_{H_1}\) and \(\witness_{H_2}\) according to the extended Feferman--Vaught sets---without requiring any access to the graph $H$ itself. 

       We can see the \(k\)-expression \(\chi\) as a tree consisting of \(\bigoh(|\chi|)\) nodes,
       and the above observation bounds the required processing time per node.
       Recalling \Cref{obs:approxSize} which bounds the number of approximate table formulas, we can bound the time to process a node of this tree by $f(|\sigma|,k,q,|\bar X\bar Y|) \cdot |\roundedN[1 + \frac{1}{\alphabot},N,\gran \cdot (b(b+1))^{4d^2+4d}]|^{3|\terms_1| + |\terms_2|\cdot f(|\sigma|,k,q,|\bar X\bar Y|)}$ for some computable function $f$. 
       Hence, 
the total running time of the whole algorithm can be bounded by
        \[|\chi| \cdot f(|\sigma|,k,q,|\bar X\bar Y|) \cdot |\roundedN[\alpha,N,\gran \cdot (b(b+1))^{4d^2+4d}]|^{3|\terms_1| + |\terms_2|\cdot f(|\sigma|,k,q,|\bar X\bar Y|)}. \qedhere
        \]
	\end{proof}

\section{Deriving the Answer from the Table}
\label{sec:lookup}

Theorem~\ref{thm:computeTable} allows us to compute a function $\witness_G$ which, for each table formula $\omega$, essentially provides all the information required to approximately evaluate $\omega$.
 However, the input query will in general not have the form of a table formula. The purpose of this section is ultimately to show how Theorem~\ref{thm:computeTable} can be used to obtain the necessary information about any formula that might occur in a \blockMSOone-query.
 
For this it will be useful to express the constraint of a general \blockMSOone-query, at least up to some loosening or tightening, using only approximate table formulas. This is achieved by Theorem~\ref{thm:almostlookup} below; similarly as in Theorem~\ref{thm:baseApprox2}, this will incur an increased granularity, expressed as a function of the accuracy \((1 + \frac{1}{\alphabot})^d\) which is linked to the answer we receive from Theorem~\ref{thm:computeTable}. The final summarizing statement of this section is given by \Cref{thm:finallookup}.

\begin{definition}\label{def:parameterizedblocks}
Given a set of weight terms $\terms_1$ over $\bar X$, and a set of weight terms $\terms_2$ over $\bar X\bar Y$,
we denote by $\Block(\gamma,\sigma,q,\bar X,\bar Y,\terms_1,\terms_2)$
all block formulas
of the form $\forall \bar Y \phi(\bar X\bar Y)$,
where $\phi$ has a signature contained in  \(\sigma\), quantifier rank at most $q$,
all weight terms are non-negative with granularity \(\gran\),
the normal weight terms with free variables from $\bar X$ are contained in $\terms_1$,
and the possible additional weight term with free variables from $\bar X\bar Y$ is contained in $\terms_2$.
\end{definition}

{\sloppy
 \begin{theorem}
 	\label{thm:almostlookup}
 	Given a formula \(\phi(\bar X) \in \Block(\gamma,\sigma,q,\bar X,\bar Y,\terms_1,\terms_2)\) and \(\alphabot,d,N \in \N\), 
    we can compute in time $|\roundedN[(1 + \frac{1}{\alphabot})^d,\lceil(1 + \frac{1}{\alphabot})^{2d} \cdot (N + 1)\rceil,\gran\cdot (\alphabot(\alphabot+1))^{3d}]|^{\mathcal{O}(|\terms_1|^2)}\cdot|\phi|^{\mathcal{O}(1)}$ 
    a formula \(\tilde{\xi}(\bar{X})\) which is obtained via conjunctions and disjunctions from formulas in \(\tilde T_{(1 + \frac{1}{\alphabot})^d,\lceil(1 + \frac{1}{\alphabot})^{2d} \cdot (N+1)\rceil,\gran\cdot (\alphabot(\alphabot+1))^{3d}}(\sigma,q,\bar X,\bar Y,\terms_1,\terms_2)\), 
    such that for every graph \(G\) whose \((\terms_1 \cup \terms_2)\)-range is contained in \(\{0, \dotsc, N\}\) and for every $\bar W\in \PP(V(G))^{|\bar X|}$,
    \begin{align*}
        G \omodels[(1 + \frac{1}{\alphabot})^{2d}] \phi(\bar{W}) &~~\Longrightarrow~~ G \omodels[(1 + \frac{1}{\alphabot})^d] \tilde{\xi}(\bar{W}), \\
        G \umodels[(1 + \frac{1}{\alphabot})^{2d}] \phi(\bar{W}) &~~\Longleftarrow~~ G \umodels[(1 + \frac{1}{\alphabot})^d] \tilde{\xi}(\bar{W}).
    \end{align*}
 \end{theorem}
 }
 \begin{proof}
	$\phi(\bar X)$ is a block formula and thus has the form $\forall \bar{Y} \phi'(\bar{X}\bar Y)$.
	The proof proceeds in two steps:
	\begin{enumerate}
		\item First, we describe a disjunction \(\xi\) of formulas in the granular table \(T_\gran(q,c,\bar X,\bar Y,\terms_1,\terms_2)\), each of whose lengths are bounded by \(\bigoh(|\phi|)\), such that for each $\bar W\in \PP(V(G))^{|\bar X|}$, \(G \models \phi(\bar{W})\) if any only if \(G \models \xi(\bar{W})\).
		\item Then, we modify the thresholds of each of these table formulas to ensure containment in \(\tilde T_{(1 + \frac{1}{\alphabot})^{d},\lceil (1 + \frac{1}{\alphabot})^{2d} \cdot (N + 1)\rceil ,\gran\cdot (b(b+1))^{3d}}(q,c,\bar X,\bar Y,\terms_1,\terms_2)\), at the cost of the ``imprecise equivalence'' between their disjunction \(\tilde{\xi}\) and $\phi(\bar X)$ described in the theorem statement.
	\end{enumerate}
	
	\paragraph{Step~1.}
	We begin by transforming each weight comparison using the strict comparison operators $<$ (or $>$) to equivalent weight comparisons using the non-strict comparison operators $\leq$ (or $\geq$, respectively); this is achieved by noticing that for two numbers \(p,q\) with granularity \(\gamma\), \(p < q\) if and only if \(p + \frac{1}{\gamma} \leq q\).
	Notice that a graph with \((\terms_1 \cup \terms_2)\)-range in \(\{0, \dots, N\}\) has a \((\terms_1' \cup \terms_2')\)-range in \(\{0, \dotsc, N + 1\}\), where \(\terms_1'\) and \(\terms_2'\) are weight terms resulting from the described rewriting.
Moreover, since $\phi$ is a block formula, it contains at most one weight comparison $t_2 \newbowtie t_1$ or $t_1 \newbowtie t_2$ where $t_2$ is a weight term involving at least one universally quantified variable and $\newbowtie~\in\{\leq, \geq\}$; if no such weight comparison exists, we w.l.o.g.\ assume
$t_2(\bar X \bar Y)=t_1(\bar X):=0$ and fix $\newbowtie$ arbitrarily.
We then syntactically rewrite every weight comparison not involving $t_2$ to fix the comparison operator to $\leq$, and the single weight comparison involving $t_2$ to the form $t_2 \newbowtie t_1$.

From now on, we consider $\phi$ to be the formula after completing the aforementioned modifications. 
Let $\notnewbowtie$ be the negation of $\newbowtie$, i.e., 
if \({\newbowtie}={\leq}\), then \({\notnewbowtie}={>}\) and if \({\newbowtie}={\geq}\) then \({\notnewbowtie}={<}\).
	
	Let $\comp(\phi,\bar{X})$ be the set of all weight comparisons over $\bar X$ alone (i.e., excluding the single weight comparison involving $\bar Y$).
	By complete specification of the possible evaluation of each element in $\comp(\phi,\bar{X})$, we obtain that \(\phi(\bar X)\) is logically equivalent to
	\begin{equation}
		\bigvee_{S \subseteq \comp(\phi,\bar{X})} \big( (\bigwedge_{t \leq t' \in S} t(\bar{X}) \leq t'(\bar{X})) \land (\quad \bigwedge_{\mathclap{t \leq t' \in \comp(\phi,\bar{X}) \setminus S}}\quad  t(\bar{X}) > t'(\bar{X})) \land (\forall \bar{Y} \chi_S(\bar{X}\bar{Y}))\big) \end{equation}
	where \(\chi_S(\bar{X}\bar{Y})\) arises from \(\phi'\) by replacing all weight comparisons in \(S\) by true and all weight comparisons in \(\comp(\phi,\bar{X}) \setminus S\) by false.
	Furthermore, by also completely specifying the evaluation of \(t_2 \newbowtie t_1\) this is equivalent to
	\begin{align}
	\begin{split}
		\bigvee_{S \subseteq \comp(\phi,\bar{X})} & \bigl(\bigwedge_{t \leq t' \in S} t(\bar{X}) \leq t'(\bar{X})\bigr) \land \bigl(\quad \bigwedge_{\mathclap{t \leq t' \in \comp(\phi,\bar{X})}} \quad t(\bar{X}) > t'(\bar{X})\bigr)\\
		{}\land{} & \forall \bar{Y} \big((t_2(\bar{X}\bar{Y}) \newbowtie t_1(\bar{X}) \rightarrow \chi_S^{t_2 \newbowtie t_1}(\bar{X}\bar{Y})) \land (t_2(\bar{X}\bar{Y}) \notnewbowtie t_1(\bar{X}) \rightarrow \chi_S^{t_2 \notnewbowtie t_1}(\bar{X}\bar{Y}))\big)
		\end{split}
	\end{align}
	where \(\chi_S^{t_2 \newbowtie t_1}(\bar{X}\bar{Y})\) arises from \(\chi_S(\bar{X}\bar{Y})\) by replacing \(t_2 \newbowtie t_1\) with true, and \(\chi_S^{t_2 \notnewbowtie t_1}(\bar{X}\bar{Y})\) arises from \(\chi_S(\bar{X}\bar{Y})\) by replacing \(t_2 \newbowtie t_1\) with false.
	Notice that \(\chi_S^{t_2 \newbowtie t_1}(\bar{X}\bar{Y}), \chi_S^{t_2 \notnewbowtie t_1}(\bar{X}\bar{Y}) \in \CMSO_1(\sigma,q,\bar{X}\bar{Y})\). Next, we rewrite the implications as disjunctions to obtain equivalence~to
	\begin{align}
	\label{rewrittenphibeforeobs}
	\begin{split}
		\bigvee_{S \subseteq \comp(\phi,\bar{X})} & \bigl(\bigwedge_{t \leq t' \in S} t(\bar{X}) \leq t'(\bar{X}) \bigr)\land \bigl(\quad \bigwedge_{\mathclap{t \leq t' \in \comp(\phi,\bar{X}) \setminus S}}\quad t(\bar{X}) > t'(\bar{X})\bigr)\\
		{}\land{} & \forall \bar{Y} \big((t_2(\bar{X}\bar{Y}) \notnewbowtie t_1(\bar{X}) \lor \chi_S^{t_2 \newbowtie t_1}(\bar{X}\bar{Y})) \land (t_2(\bar{X}\bar{Y}) \newbowtie t_1(\bar{X}) \lor \chi_S^{t_2 \notnewbowtie t_1}(\bar{X}\bar{Y}))\big).
		\end{split}
	\end{align}
	
Notice that this formula, while equivalent to $\phi(\bar X)$, contains weight comparisons which were not in $\phi(\bar X)$---specifically weight comparisons of the form \(t > t'\) and \(t_2 \notnewbowtie t_1\). These would behave unfavorably with respect to the tightening and loosening that will take place in the second step of the proof. Luckily, we show that the formula can be simplified further in a way which omits these new weight comparisons. 
We will do so by using the following general observation.

	For binary variables $A_0, A_1, B$ such that $A_0$ implies $A_1$,
	\begin{align}
	\label{generalobsone}
	\begin{split}
		& (\neg B \lor A_1) \land (B \lor A_0) \\
		{} \equiv {}& (\neg B \lor A_1) \land \Bigl((B \lor A_1) \land (B \lor A_0)\Bigr) \\
		{} \equiv {}& \Bigl((\neg B \lor A_1) \land (B \lor A_1)\Bigr) \land (B \lor A_0) \\
		{} \equiv {}& A_1 \land (B \lor A_0).
		\end{split}
	\end{align}
	To simplify the formula in~(\ref{rewrittenphibeforeobs}), we will now apply the observation from~(\ref{generalobsone}) for \(B := t_2(\bar{X}\bar{Y}) \newbowtie t_1 (\bar X)\), \(A_0 := \chi_S^{t_1 \notnewbowtie t_2}(\bar{X}\bar{Y})\) and \(A_1 := \chi_S^{t_2 \newbowtie t_1}(\bar{X}\bar{Y})\). This is possible because \(t_2(\bar{X}\bar{Y}) \newbowtie t_1(\bar{X})\) occurs as a positive atom in \(\phi'\), which guarantees that $\chi_S^{t_1 \notnewbowtie t_2}(\bar{X}\bar{Y})$ implies $\chi_S^{t_2 \newbowtie t_1}(\bar{X}\bar{Y})$. Applying~(\ref{generalobsone}) yields the equivalent formulation of~(\ref{rewrittenphibeforeobs})
	\begin{align}
	\label{phiafterobs}
	\begin{split}
		\bigvee_{S \subseteq \comp(\phi,\bar{X})} & \big((\bigwedge_{t \leq t' \in S} t(\bar{X}) \leq t'(\bar{X}) )\land (\quad \bigwedge_{\mathclap{t \leq t' \in \comp(\phi,\bar{X}) \setminus S}} \quad t(\bar{X}) > t'(\bar{X}))\\
		{} \land {} & \forall \bar{Y} \bigl(\chi_S^{t_2 \newbowtie t_1}(\bar{X}\bar{Y}) \land \bigl(t_2(\bar{X}\bar{Y}) \newbowtie t_1(\bar{X}) \lor \chi_S^{t_2 \notnewbowtie t_1}(\bar{X}\bar{Y})\bigr)\bigr).
		\end{split}
	\end{align}

Our next task will be to get rid of the second half of the first line in~(\ref{phiafterobs}), i.e., to establish equivalence to
	\begin{align}
	\label{phicleaned}
	\begin{split}
		\bigvee_{S \subseteq \comp(\phi,\bar{X})} & \big((\bigwedge_{t \leq t' \in S} t(\bar{X}) \leq t'(\bar{X}))\\
		{} \land {} & \forall \bar{Y} (\chi_S^{t_2 \newbowtie t_1}(\bar{X}\bar{Y}) \land (t_2(\bar{X}\bar{Y}) \newbowtie t_1(\bar{X}) \lor \chi_S^{t_2 \notnewbowtie t_1}(\bar{X}\bar{Y}))\big).
		\end{split}
	\end{align}
	
To argue this equivalence, we notice that the formula in~(\ref{phiafterobs}) immediately implies the formula in~(\ref{phicleaned}) since for each choice of $S$, the latter consists of a conjunction over a subset of the atoms. For the opposite direction, assume that an arbitrary graph $H$ models the formula in~(\ref{phicleaned}) for some instantiation $\bar W$ of $\bar X$. This means that there exists some $S_1\subseteq \comp(\phi,\bar{X})$ such that all weight comparisons in $S_1$ are satisfied by $\bar W$ in $H$, and the subformula on the second row holds. Now let us define $S_2\supseteq S_1$ as the set of all weight comparisons in $\comp(\phi,\bar{X})$ which are satisfied by $\bar W$ in $H$. By definition, for this choice of $S:=S_2$ we are guaranteed that all atoms on the first row of the formula in~(\ref{phiafterobs}) are satisfied. Moreover, since all weight comparisons occur as positive atoms in $\phi(\bar X)$ and since $S_2\supseteq S_1$, the definition of \(\chi_S^{t_2 \newbowtie t_1}(\bar{X}\bar{Y})\) ensures that \(\chi_{S_1}^{t_2 \newbowtie t_1}(\bar{X}\bar{Y})\) implies \(\chi_{S_2}^{t_2 \newbowtie t_1}(\bar{X}\bar{Y})\). Hence, $H$ also models the formula in~(\ref{phiafterobs}) when interpreting $\bar X$ as $\bar W$, establishing the desired equivalence.

	\newcommand{\ubar}[1]{\underline{#1}}
	For our next step, we need to ensure that all weight comparisons only have free variables on one side. While the intuitive way to achieve this would be to simply move all the terms involving free variables to one side (and this would be fine as far as exact solutions as considered), doing so will not be feasible in the approximate setting where our aim is to preserve over- and undersatisfaction. Indeed, for an illustrative example consider that $G\umodels[\circ] |X_1|\leq |X_2|$ (for some $X_1, X_2\subseteq V(G)$ and some $\circ>1$) is fundamentally different from $G\umodels[\circ] |X_1|-|X_2|\leq 0$, since the former may hold even if $|X_1|<|X_2|$ and the latter does not.
	
	To deal with this in a way which is compatible with under- and oversatisfaction, we express each weight comparison between terms as a conjunction of weight comparisons where one side has no free variables, i.e., is a number. 
	Specifically, for every variable instantiation, each weight comparison \(t \leq t' \in \comp(\phi,\bar{X})\) can be equivalently replaced by \(t \leq g(t \leq t') \land g(t \leq t') \leq t'\) for some integer $0\leq g(t \leq t')\leq N+1$. 
	In particular, if for $a,b\in \Nat$ we use $Q_{a,b}$ to denote the set of all numbers $q$ of granularity $b\in \Nat$ such that $0\leq q\leq a$, then $g(t \leq t')\in Q_{N+1,\gran}$.	

	Moreover, the same also holds for \(t_2(\bar{X}\bar{Y}) \newbowtie t_1(\bar{X})\) since we can assume \(g(t_2 \newbowtie t_1) = t_1(\bar{X})\), irrespective of $t_2(\bar{X}\bar{Y})$ or the instantiation of $Y$.
    The weight comparison $g(t_2\newbowtie t_1) \newbowtie t_1(\bar{X})$ may thus occur outside the quantification of $Y$. 
	We will now rewrite the formula from~(\ref{phicleaned}) by adding a global disjunction over all values of these numbers $g$. To avoid any confusion, we remark that this indeed increases the total length of the formula by a factor of $\bigoh(N)$, but the length of each individual subformula forming an atom of this disjunction remains upper-bounded by $\bigoh(|\phi|)$.
	
	\newcommand\numberthis{\addtocounter{equation}{1}\tag{\theequation}}
		\begin{align*}
	\label{phibig}
		\bigvee_{S \subseteq \comp(\phi,\bar{X})} & \Bigg(\bigvee_{g : S \cup \{t_2 \newbowtie t_1\} \to Q_{N+1,\gran}} \Big(\bigwedge_{t \leq t' \in S}  \big(t(\bar{X}) \leq g(t \leq t') \land g(t \leq t') \leq t'(\bar{X})\big) \numberthis\\
				{}\land{} & g(t_2\newbowtie t_1) \newbowtie t_1(\bar{X}) \land
				\forall \bar{Y} \big(\chi_S^{t_2 \newbowtie t_1}(\bar{X}\bar{Y}) \land \big((t_2(\bar{X}\bar{Y}) \newbowtie g(t_2\newbowtie t_1)  \lor \chi_S^{t_2 \notnewbowtie t_1}(\bar{X}\bar{Y})\big)\big)\Big)\Bigg).
	\end{align*}

To truly transform~(\ref{phibig}) into a disjunction of table formulas, for \(g: S \cup \{t_2 \newbowtie t_1\} \to Q_{N+1,\gran}\) and \(t \in \terms_1\), we introduce a set of thresholds that will apply to individual weight terms occurring in~(\ref{phibig}) rather than weight comparisons.
For this we will want to choose the strictest threshold for an individual weight term implied by a comparison involving it, which in the case of \(t_1\) may include a threshold introduced for \(t_1 \bowtie t_2\).
Moreover, in case that no upper or lower threshold for some term is implied by our thresholds for weight comparisons, we want to choose a threshold to extreme that it is trivially satisfied.
For this, we use the upper or lower thresholds implied by the bounded \((\terms_1 \cup \terms_2)\)-range of \(G\).
For \(t \in \terms_1\) and \(\leq\), this means we can use as threshold
\(g^g_{t,\leq}\) the minimum (corresponding to the strictest) \(g\)-value of a weight comparison in which \(t\) is on the left, and the trivially satisfied upper threshold for \(t\) is \(N + 1\).
Hence, we can formalize our above intentions by defining compactly
\[g^g_{t,\leq} :=
\min \Bigl(\{g(t \leq t') \mid t \leq t' \in S\} \cup 
\{g(t_2 \newbowtie t_1) \mid {\newbowtie}={\geq} \land t = t_1\} \cup \{N + 1\}\Bigr).
\]

Analogously, for each weight term $t(\bar X)$ occurring in a weight comparison of the form $g(t' \leq t) \leq t(\bar X)  \in S$, we will set its threshold $g^g_{t,\geq}$ as the largest value of $g(t' \leq t)$ over all choices of $t'$, while also including the special case where $\newbowtie=\leq$ and $t_1(\bar X)=t(\bar X)$. Formally,
	\[g^g_{t,\geq} :=
		\max \Bigl(\{g(t' \leq t) \mid t' \leq t \in S\} \cup \{g(t_2 \newbowtie t_1) \mid {\newbowtie}={\leq} \land t=t_1\} \cup \{0\}\Bigr).
    \]
	
When used in a table formula, the above thresholds capture all weight comparisons occurring in the formula~(\ref{phibig}) apart from those occurring in the part of the formula quantified by $\bar Y$. To deal with these weight comparisons, we will use the table formula thresholds of the form $g_{\circ,\psi,\circ}$ (over all choices of \(\psi(\bar X \bar Y) \in \CMSO_1(\sigma,q,\bar{X}\bar{Y})\)). We again provide the intuition over how this works before giving the formal definition of the thresholds.

Since formula~(\ref{phibig}) requires $\chi_S^{t_2 \newbowtie t_1}(\bar{X}\bar{Y})$ to hold for all choices of $\bar Y$, for the case where $\psi(\bar X \bar Y)=\chi_S^{t_2 \newbowtie t_1}(\bar{X}\bar{Y})$ we set the thresholds in a way which make the weight comparisons unsatisfiable, hence forcing $\psi(\bar X \bar Y)$ to be satisfied. Moreover, to encode 
$\big((t_2(\bar{X}\bar{Y}) \newbowtie g(t_2\newbowtie t_1)  \lor \chi_S^{t_2 \notnewbowtie t_1}(\bar{X}\bar{Y})\big)$
we set one of the thresholds (depending on whether ${\newbowtie} = {\leq}$ or ${\newbowtie}={\geq}$) for $\psi(\bar X\bar Y)=\chi_S^{t_2 \notnewbowtie t_1}(\bar{X}\bar{Y})$ to the value of $g(t_2\newbowtie t_1)$; the other threshold is not needed and hence is set in a way which makes the corresponding weight comparison trivially satisfied. For all other choices of $\psi(\bar X\bar Y)$, we do not care whether it is satisfied or not and hence set the corresponding thresholds in a way which makes them both trivially satisfied. 

To formalize the previous paragraph, for each \(\psi(\bar X\bar Y) \in \CMSO_1(\sigma,q,\bar{X}\bar{Y})\) we set
	\[g^g_{t_2,\psi,\leq} := \begin{cases}
        g(t_2 \newbowtie t_1) & \mbox{ if } {\newbowtie}={\leq} \text{ and } \psi = \chi_S^{t_2 \notnewbowtie t_1}\\
		0 & \mbox{ if } \psi = \chi_S^{t_2 \newbowtie t_1}\\
		N + 1 & \mbox{ otherwise,}
	\end{cases}\]
	and
	\[g^g_{t_2,\psi,\geq} := \begin{cases}
        g(t_2 \newbowtie t_1) & \mbox{ if } {\newbowtie}={\geq} \text{ and } \psi = \chi_S^{t_2 \notnewbowtie t_1}\\
	N + 1 & \mbox{ if } \psi = \chi_S^{t_2 \newbowtie t_1}\\
	0 & \mbox{ otherwise.}
	\end{cases}\]

	Overall this allows us to rewrite formula~(\ref{phibig}) as a disjunction of length-$\bigoh(|\phi|)$ table formulas
	\begin{align*}
	\label{phifinal}
	\xi(\bar X)=\bigvee_{\substack{S \subseteq \comp(\phi,\bar{X})\\
			g: S \cup \{t_2 \newbowtie t_1\} \to Q_{N+1,\gran}}}	
		& \Bigl(\bigwedge_{t \in \terms_1} t(\bar{X}) \leq g^g_{t,\leq} \land t(\bar{X}) \geq g^g_{t,\geq}\Bigr) \numberthis\\
		{}\land{} & \Bigl(\bigwedge_{\psi \in \CMSO_1(\sigma,q,\bar{X}\bar{Y})} \forall \bar{Y}  \begin{aligned}[t] & \psi(\bar{X}\bar{Y})\\
			{}\lor{} & (t_2(\bar{X}\bar{Y}) \leq g^g_{t_2,\psi,\leq} \land t_2(\bar{X}\bar{Y}) \geq g^g_{t_2,\psi,\geq})\Bigr).
		\end{aligned}
	\end{align*}
	
    Summarizing, we have that for every $\bar W \in \PP(V(G))^{|\bar X|}$, \(G \models \phi(\bar W)\) if and only if ${G\models\xi(\bar W)}$, completing our targeted first step.
	
	\paragraph{Step~2.}
	Now, consider some function $h:S\cup \{t_2 \newbowtie t_1\} \to Q_{N+1,\gran \cdot (b(b+1))^{2d}}$. Intuitively, such an $h$ can be viewed as a choice of $g$ from the previous step, but with the additional option of allowing for higher granularity (a property which will be useful later).
Our aim will be to alter $h$ in a way which only uses numbers from a rounded set $\roundedN[1+\frac{1}{b},\circ,\verysmallsquare]$ where $\circ$ will be a slightly larger upper bound than $N+1$ and the granularity $\smallsquare$ will be slightly larger than $\gran \cdot (b(b+1))^{2d}$. 
	
	In particular, by \Cref{obs:sandwich}, for each \(t \leq t' \in \comp(\phi,\bar{X}) \cup \{t_2 \newbowtie t_1\}\), we find a number 
	\[\tilde h_\geq(t \leq t') \in \left[\left(\frac{b}{b + 1}\right)^d \cdot h(t \leq t'), h(t \leq t')\right] \cap \roundedN[(1 + \frac{1}{\alphabot})^d,\lceil(1 + \frac{1}{\alphabot})^d \cdot (N + 1)\rceil,\gran \cdot b^{2d}(b+1)^{3d}]\]
	and set \(\tilde h_\leq(t \leq t') := \min\{a\in \roundedN[(1 + \frac{1}{\alphabot})^d,\lceil(1 + \frac{1}{\alphabot})^d \cdot (N + 1)\rceil,\gran \cdot b^{3d}(b+1)^{3d}]~|~a\geq \tilde h_\geq(t \leq t') \cdot \left(\frac{\alphabot+1}{\alphabot}\right)^d\}\). Observe that by applying \Cref{obs:sandwich} on the interval $[\tilde h_\geq(t \leq t') \cdot \left(\frac{\alphabot+1}{\alphabot}\right)^{d},\tilde h_\geq(t \leq t') \cdot (\frac{\alphabot+1}{\alphabot})^{2d}]$ (where $h_\geq(t \leq t') \cdot \left(\frac{\alphabot+1}{\alphabot}\right)^{d}$ has granularity $\gran \cdot b^{3d}(b+1)^{3d}$), we obtain a guarantee that $\tilde h_\leq(t \leq t')\leq \tilde h_\geq(t \leq t') \cdot (\frac{\alphabot+1}{\alphabot})^{2d}$. In particular, this ensures
	\[\tilde h_\leq(t \leq t') \in \left[h(t \leq t'), \Bigl(\frac{\alphabot + 1}{\alphabot}\Bigr)^{2d} \cdot h(t \leq t')\right] \cap \roundedN[(1 + \frac{1}{\alphabot})^d,\lceil(1 + \frac{1}{\alphabot})^d \cdot (N + 1)\rceil,\gran \cdot b^{3d}(b+1)^{3d}].\]

	For convenience, we can treat all of these numbers as simply having granularity $\gran \cdot (b(b+1))^{3d}$. 
	We use these to define ``rounded'' analogues to the thresholds obtained in Step~1, with another slight increase of the upper bound to $(1+\frac{1}{b})^{2d}\cdot (N+1)$,
\begin{align*}	
	\tilde g^h_{t,\leq} =&
		\min \big(\{\tilde h_\leq(t \leq t') \mid t \leq t' \in S\} \cup 
				\{\tilde h_\leq(t_2 \newbowtie t_1) \mid {\newbowtie}={\geq} \land t = t_1\}
\cup  \{(N + 1) \Bigl(1+\frac{1}{b}\Bigr)^{2d}\}\big),\\
	\tilde g^h_{t,\geq} =&
		\max \big(\{\tilde h_\geq(t' \leq t) \mid t' \leq t \in S\} \cup 
		\{\tilde h_\geq(t_2 \newbowtie t_1) \mid {\newbowtie}={\leq} \land t=t_1\}
		\cup \{0\}\big), \hfill~\quad
\end{align*}		
and for each \(\psi \in \CMSO_1(\sigma,q,\bar{X}\bar{Y})\) we define
	\[\tilde g^h_{t_2,\psi,\leq} = \begin{cases}
        \tilde h_\leq(t_1 \newbowtie t_2) & \mbox{ if } {\newbowtie} = {\leq} \text{ and } \psi = \chi_S^{t_2 \notnewbowtie t_1}\\
		0 & \mbox{ if } \psi = \chi_S^{t_2 \newbowtie t_1}\\
		(1 + \frac{1}{\alphabot})^{2d} \cdot (N + 1) & \mbox{ otherwise,} 
	\end{cases}\]
	and
	\[\tilde g^h_{t_2,\psi,\geq} = \begin{cases}
        \tilde h_\geq(t_1 \newbowtie t_2) & \mbox{ if } {\newbowtie} ={\geq} \text{ and } \psi = \chi_S^{t_2 \notnewbowtie t_1}\\
		(1 + \frac{1}{\alphabot})^{2d} \cdot (N + 1) & \mbox{ if } \psi = \chi_S^{t_2 \newbowtie t_1}\\
		0 & \mbox{ otherwise.}
	\end{cases}\]
	We can use these ``rounded'' thresholds to obtain a disjunction of formulas in the approximate table \(\tilde T_{(1 + \frac{1}{\alphabot})^d,\lceil(1 + \frac{1}{\alphabot})^{2d} \cdot (N + 1)\rceil,\gran \cdot (b(b+1))^{3d}}(q,c,\bar X,\bar Y,\terms_1,\terms_2)\) in an analogous way as in the first step of the proof, in particular giving rise to the formula
	
	\[\bigvee_{\substack{S \subseteq \comp(\phi,\bar{X})\\
			h \colon S \cup \{t_2 \newbowtie t_1\} \to Q_{N+1,\gran \cdot (b(b+1))^{2d}}}} \begin{aligned}[t]
		& \Bigl(\bigwedge_{t \in \terms_1} t(\bar{X}) \leq \tilde g^h_{t,\leq} \land t(\bar{X}) \geq \tilde g^h_{t,\geq}\Bigr)\\
		{}\land{} & \Bigl(\bigwedge_{\mathclap{\psi \in \CMSO_1(\sigma,q,\bar{X}\bar{Y})}}\quad\quad \forall \bar{Y}  \begin{aligned}[t] & \psi(\bar{X}\bar{Y})\\
			{}\lor{} & \bigl(t_2(\bar{X}\bar{Y}) \leq \tilde g^h_{t_2,\psi,\leq} \land t_2(\bar{X}\bar{Y}) \geq \tilde g^h_{t_2,\psi,\geq}\bigr)\Bigr).
		\end{aligned}
	\end{aligned}\label{phibeforez} \numberthis
	\]

Observe that the disjunction at the very beginning of the formula above suggests that it has a length of at least $\big((N+1)\cdot \gran \cdot (b(b+1))^{2d}\big)^{|\terms_1|}$, and hence it could not be constructed in the running time claimed in the theorem statement. Luckily, many of the choices of $h$ lead to the same values of $\tilde h_\geq$ (and hence also of $\tilde h_\leq$), and we can thus change the range of the disjunction by directly considering the values of $\tilde h_\geq$ to obtain the following formula $\tilde \xi$
	\[\label{phiwithz}\numberthis
	\bigvee_{\substack{z \colon (\comp(\phi,\bar{X}) \cup \{t_2 \newbowtie t_1\})~\to \{\bot\} \cup \\ \roundedN[(1 + \frac{1}{\alphabot})^d,\lceil(1 + \frac{1}{\alphabot})^{2d} \cdot (N + 1)\rceil,\gran \cdot b^{2d}(b+1)^{3d}]\\[3pt] \text{such that } z(t_2 \newbowtie t_1)\neq \bot}} \begin{aligned}[t]
		& \Bigl(\bigwedge_{t \in \terms_1} t(\bar{X}) \leq \tilde g^z_{t,\leq} \land t(\bar{X}) \geq \tilde g^z_{t,\geq}\Bigr)\\
		{}\land{} & \Bigl(\bigwedge_{\mathclap{\psi \in \CMSO_1(\sigma,q,\bar{X}\bar{Y})}}\quad\quad \forall \bar{Y}  \begin{aligned}[t] & \psi(\bar{X}\bar{Y})\\
			{}\lor{} & \bigl(t_2(\bar{X}\bar{Y}) \leq \tilde g^z_{t_2,\psi,\leq} \land t_2(\bar{X}\bar{Y}) \geq \tilde g^z_{t_2,\psi,\geq}\bigr)\Bigr),
		\end{aligned}
	\end{aligned}
	\]
	where the thresholds are obtained as follows.
	\begin{itemize}
	\item $\tilde g^z_{t,\geq}$ is obtained analogously to $\tilde g^h_{t,\geq}$ but where $\tilde h_\geq$ is replaced by $z$, i.e.,
	\end{itemize}
	\[\tilde g^z_{t,\geq} =
		\max \Bigl(\{z(t' \leq t) \mid z(t' \leq t)\neq \bot\} \cup  
		\{z(t_2 \newbowtie t_1) \mid {\newbowtie}={\leq} \land t = t_1\}
		\cup \{0\}\Bigr). \hfill~\quad\]

	\begin{itemize}
	\item $\tilde g^z_{t,\leq}$ is obtained analogously to $\tilde g^h_{t,\leq}$ but where $\tilde h_\leq$ is replaced by \\$z^\uparrow(t\leq t'):=\min\{a\in \roundedN[(1 + \frac{1}{\alphabot})^d,\lceil(1 + \frac{1}{\alphabot})^d \cdot (N + 1)\rceil,\gran \cdot b^{3d}(b+1)^{3d}]~|~a\geq \tilde z(t\leq t') \cdot \left(\frac{\alphabot+1}{\alphabot}\right)^d\}$, i.e.,
	\end{itemize}	
	\[\tilde g^z_{t,\leq} =
		\min \biggl(\Bigl\{z^\uparrow(t \leq t') \Bigm| z(t \leq t')\neq \bot\Bigr\}~\cup\] 
		 \[
	\Bigl\{z^\uparrow(t_2 \newbowtie t_1) \Bigm| {\newbowtie}={\geq} \land t=t_1\Bigr\}	 
\cup  \Bigl\{(N + 1)\cdot \Bigl(1+\frac{1}{b}\Bigr)^{2d}\Bigr\}\biggr).\]			

	\begin{itemize}
	\item $\tilde g^z_{t_2,\psi,\leq}$, $\tilde g^z_{t_2,\psi,\geq}$ are obtained in the same way as $\tilde g^h_{t_2,\psi,\leq}$, $\tilde g^h_{t_2,\psi,\geq}$, respectively, when interpreting $\tilde h_\geq(\circ)$ as $z(\circ)$, $\tilde h_\leq(\circ)$ as $z^\uparrow(\circ)$, and $S$ as the subset of $\comp(\phi,\bar{X})$ which is not mapped by $z$ to $\bot$.
\end{itemize}
	Notice that for every alternative in the formula~(\ref{phibeforez}) defined by a choice of $S$ and $h$, we can identify a choice of $z$ in the formula~(\ref{phiwithz}) that leads to the same thresholds; in other words, (\ref{phiwithz}) includes all alternatives from (\ref{phibeforez}). At the same time, for every alternative in (\ref{phiwithz}) defined by a choice of $z$, we can identify a choice of $S$ and $h$ in (\ref{phibeforez}) which leads to the same thresholds; in other words, (\ref{phibeforez}) includes all alternatives from (\ref{phiwithz}).
    This means (\ref{phibeforez}) and (\ref{phiwithz}) are equivalent.
    
	We now move on to show the desired properties of \(\tilde \xi\).
    For this we only consider the case that \({\newbowtie}={\leq}\), as \({\newbowtie}={\geq}\) can be handled symmetrically.
	Let \(G\) be a graph whose signature matches \(\phi\) such that all weight comparisons in \(\phi\) contain only \(\leq\) and \(\geq\) whose \((\terms_1 \cup \terms_2)\)-range is contained in \(\{0, \dotsc, N + 1\}\).

	\paragraph{First Implication.}
	If \(G \omodels[(1 + \frac{1}{\alphabot})^{2d}] \phi(\bar{X})\), let this be witnessed by \(\bar{W} \in \PP(V(G))^{|\bar{X}|}\).
	We can carry out the same construction as in Step~1 up to including step (\ref{phibig}) for \(\phi^{(1 + \frac{1}{\alphabot})^{2d}}\) rather than \(\phi\) to write \(G \omodels[(1 + \frac{1}{\alphabot})^{2d}] \phi(\bar{W})\) equivalently as 
	\begin{align*}
		G \models & \bigvee_{\substack{S \subseteq \comp(\phi,\bar{X})\\
		h \colon S \cup \{t_2 \newbowtie t_1\}~\to~Q_{(N+1)\cdot (1+\frac{1}{b})^{2d},\gran \cdot (b(b+1))^{2d}}}}\\ & \bigwedge_{t \leq t' \in S} \left(\frac{\alphabot + 1}{\alphabot}\right)^{2d} \cdot t(\bar{W}) \leq h(t \leq t') \land h(t \leq t') \leq \left(\frac{\alphabot}{\alphabot + 1}\right)^{2d} \cdot t'(\bar{W})\\
		&\land h(t_2 \leq t_1) \leq \left(\frac{\alphabot}{\alphabot + 1}\right)^{2d} \cdot t_1(\bar{W})\\
		& \land \forall \bar{Y} \chi_S^{t_2 \leq t_1}(\bar{W}\bar{Y}) \land \left(\left(\frac{\alphabot + 1}{\alphabot}\right)^{2d} \cdot t_2(\bar{W}\bar{Y}) \leq h(t_2 \leq t_1) \lor \chi_S^{t_2 > \cdot t_1}(\bar{W}\bar{Y})\right).
	\end{align*}
	Here, we used the fact that each weight comparison \(\left(\frac{\alphabot + 1}{\alphabot}\right)^{2d} \cdot t \leq \left(\frac{\alphabot}{\alphabot + 1}\right)^{2d} \cdot t'\) in \(\phi^{\left(\frac{\alphabot + 1}{\alphabot}\right)^{2d}}\) corresponds directly to a weight comparison \(t \leq t'\) in \(\phi\).
	
	Now, let us consider the choice of $h: S \cup \{t_2 \leq t_1\} \to Q_{(N+1)\cdot (1+\frac{1}{b})^{2d},\gran \cdot (b(b+1))^{2d}}$ and \(S \subseteq \comp(\phi,\bar{X})\) such that the corresponding alternative of the above disjunction is satisfied.
	This yields
	\begin{align} \label{phifirstimplication}
	\begin{split}
		G \models & \bigwedge_{t \leq t' \in S} \left(\frac{\alphabot + 1}{\alphabot}\right)^{2d} \cdot t(\bar{W}) \leq h(t \leq t') \land h(t \leq t') \leq \left(\frac{\alphabot}{\alphabot + 1}\right)^{2d} \cdot t'(\bar{W})\\
		& \land h(t_2 \leq t_1) \leq \left(\frac{\alphabot}{\alphabot + 1}\right)^{2d} \cdot t_1(\bar{W})\\
		& \land \forall \bar{Y} \chi_S^{t_2 \leq t_1}(\bar{W}\bar{Y}) \land \left(\left(\frac{\alphabot + 1}{\alphabot}\right)^{2d} \cdot t_2(\bar{W}\bar{Y}) \leq h(t_2 \leq t_1) \lor \chi_S^{t_2 > t_1}(\bar{W}\bar{Y})\right).
		\end{split}
	\end{align}

	To this end, we first rewrite the formula in~(\ref{phifirstimplication}) (for the aforementioned choice of $h$ and $S$) by multiplying both sides of each inequality with \(\left(\frac{\alphabot + 1}{\alphabot}\right)^d\) or \(\left(\frac{\alphabot}{\alphabot + 1}\right)^d\) to obtain
	\begin{align}
	\label{phirightform}
	\begin{split}
		G \models & \bigwedge_{t \leq t' \in S} \left(\frac{\alphabot + 1}{\alphabot}\right)^d \cdot t(\bar{W}) \leq \left(\frac{\alphabot}{\alphabot + 1}\right)^d \cdot h(t \leq t') \land \left(\frac{\alphabot + 1}{\alphabot}\right)^d \cdot h(t \leq t') \leq \left(\frac{\alphabot}{\alphabot + 1}\right)^d \cdot t'(\bar{W})\\
		& \land \left(\frac{\alphabot + 1}{\alphabot}\right)^d \cdot h(t_2 \leq t_1) \leq \left(\frac{\alphabot}{\alphabot + 1}\right)^d \cdot t_1(\bar{W})\\
		& \land \forall \bar{Y} \chi_S^{t_2 \leq t_1}(\bar{W}\bar{Y}) \land \left(\left(\frac{\alphabot + 1}{\alphabot}\right)^d \cdot t_2(\bar{W}\bar{Y}) \leq \left(\frac{\alphabot}{\alphabot + 1}\right)^d \cdot h(t_2 \leq t_1) \lor \chi_S^{t_2 > t_1}(\bar{W}\bar{Y})\right).
		\end{split}
	\end{align}

	Recalling the definitions of $\tilde h_\leq(t \leq t')$ and $\tilde h_\geq(t \leq t')$ from the second step, this implies the following formula (in particular, each weight comparison below is implied by the corresponding weight comparison in~(\ref{phirightform}))
	\begin{align*}
		G \models & \bigwedge_{t \leq t' \in S} \left(\frac{\alphabot + 1}{\alphabot}\right)^d\hspace{-0.08cm} \cdot t(\bar{W}) \hspace{-0.05cm}\leq \hspace{-0.05cm} \left(\frac{\alphabot}{\alphabot + 1}\right)^d\hspace{-0.08cm} \cdot \tilde h_\leq(t \leq t') \land \left(\frac{\alphabot + 1}{\alphabot}\right)^d\hspace{-0.08cm} \cdot \tilde h_\geq(t \leq t') \hspace{-0.05cm}\leq \hspace{-0.05cm}\left(\frac{\alphabot}{\alphabot + 1}\right)^d \cdot t'(\bar{W})\\
		& \land \left(\frac{\alphabot + 1}{\alphabot}\right)^d \cdot \tilde h_\geq(t_2 \leq t_1) \leq \left(\frac{\alphabot}{\alphabot + 1}\right)^d \cdot t_1(\bar{W})\\
		& \land \forall \bar{Y} \chi_S^{t_2 \leq t_1}(\bar{W}\bar{Y}) \land \left(\left(\frac{\alphabot + 1}{\alphabot}\right)^d \cdot t_2(\bar{W}\bar{Y}) \leq \left(\frac{\alphabot}{\alphabot + 1}\right)^d \cdot \tilde h_\leq(t_2 \leq t_1) \lor \chi_S^{t_2 > t_1}(\bar{W}\bar{Y})\right).
	\end{align*}

	Observe that because of the \((\terms_1 \cup \terms_2)\)-range of \(G\) and the way in which we defined the upper bound for $g^h_{t,\leq}$,
	whenever a threshold of the form \(g^h_{t,\circ,\leq}\) is set to \((1 + \frac{1}{\alphabot})^{2d} \cdot (N + 1)\), \(G \models \left(\frac{\alphabot + 1}{\alphabot}\right)^d \cdot t \leq \left(\frac{\alphabot}{\alphabot + 1}\right)^d \cdot g^h_{t,\circ,\leq}\).
	At this point, we use the previous observation that for each choice of $S$ and $h$ in formula~(\ref{phibeforez}) there exists a corresponding choice of $z$ in formula~(\ref{phiwithz}) to establish that there indeed is a choice of $z$ (i.e., an alternative in $\tilde \xi$) such that
	\[G \models \begin{aligned}[t] & \biggl(\bigwedge_{t \in \terms_1} \left(\frac{\alphabot + 1}{\alphabot}\right)^d \cdot t(\bar{W}) \leq \left(\frac{\alphabot}{\alphabot + 1}\right)^d \cdot \tilde g^z_{t,\leq} \land \left(\frac{\alphabot}{\alphabot + 1}\right)^d \cdot t(\bar{W}) \geq \left(\frac{\alphabot + 1}{\alphabot}\right)^d \cdot \tilde g^z_{t,\geq}\biggr)\\
		{}\land{} & \biggl(\quad\bigwedge_{\mathclap{\psi \in \CMSO_1(\sigma,q,\bar{X}\bar{Y})}}\quad \forall \bar{Y}  \begin{aligned}[t] & \psi(\bar{X}\bar{Y})\\
			{}\lor{} & \biggl(\left(\frac{\alphabot + 1}{\alphabot}\right)^d \cdot t_2(\bar{W}\bar{Y}) \leq \left(\frac{\alphabot}{\alphabot + 1}\right)^d \cdot \tilde g^z_{t_2,\psi,\leq}  \\
&	\land 		 \left(\frac{\alphabot}{\alphabot + 1}\right)^d \cdot t_2(\bar{W}\bar{Y}) \geq \left(\frac{\alphabot + 1}{\alphabot}\right)^d \cdot \tilde g^z_{t_2,\psi,\geq}\biggr)\biggr).
		\end{aligned}
	\end{aligned}\]
	This completes the proof of the first implication.

	\paragraph{Second Implication.}
	Assume \(G \umodels[(1 + \frac{1}{\alphabot})^d] \tilde \xi(\bar W)\) for some \(\bar{W} \in \PP(V(G))^{|\bar{X}|}\), and consider a choice of $z$ which gives rise to an alternative of the disjunction in \(\tilde \xi\) that is satisfied by \(\bar{W}\) in \(G\).
	In particular, this means
	\begin{align}
		\label{lem:extract:approxtableunder}
		G \models \begin{aligned}[t]
		& \big(\bigwedge_{t \in \terms_1} \left(\frac{\alphabot}{\alphabot + 1}\right)^d t(\bar{W}) \leq \left(\frac{\alphabot + 1}{\alphabot}\right)^d \cdot \tilde g^z_{t,\leq} \land \left(\frac{\alphabot + 1}{\alphabot}\right)^d \cdot t(\bar{W}) \geq \left(\frac{\alphabot}{\alphabot + 1}\right)^d \cdot \tilde g^z_{t,\geq}\big) \\
		{}\land{} & \biggl(\bigwedge_{\mathclap{\psi \in \CMSO_1(\sigma,q,\bar{W}\bar{Y})}} \quad \forall \bar{Y}  \begin{aligned}[t] & \psi(\bar{W}\bar{Y})\\
			{}\lor{} & \biggl(\left(\frac{\alphabot}{\alphabot + 1}\right)^d t_2(\bar{W}\bar{Y}) \leq \left(\frac{\alphabot + 1}{\alphabot}\right)^d \cdot \tilde g^z_{t_2,\psi,\leq} \\
&			\land  \left(\frac{\alphabot + 1}{\alphabot}\right)^d \cdot t_2(\bar{W}\bar{Y}) \geq \left(\frac{\alphabot}{\alphabot + 1}\right)^d \cdot \tilde g^z_{t_2,\psi,\geq}\biggr)\biggr).
		\end{aligned}
	\end{aligned}
	\end{align}

	Once again, observe that because of the \((\terms_1 \cup \terms_2)\)-range of \(G\) and the way in which we defined the upper bound for $g^h_{t,\leq}$,
	whenever a threshold of the form \(g^h_{t,\cdot,\leq}\) is set to \((1 + \frac{1}{\alphabot})^{2d} \cdot (N + 1)\), \(G \models \left(\frac{\alphabot + 1}{\alphabot}\right)^d \cdot t \leq \frac{\alphabot}{\alphabot + 1} \cdot g^h_{t,\cdot,\leq}\).
We now use the fact that for each choice of $z$ in~(\ref{phiwithz}), we can identify a corresponding choice of $S$ and $h$ in~(\ref{phibeforez}) for which it must now hold that
\begin{align}
		\label{secondimplicationh}
		G \models \begin{aligned}[t]
		& \big(\bigwedge_{t \in \terms_1} \left(\frac{\alphabot}{\alphabot + 1}\right)^d t(\bar{W}) \leq \left(\frac{\alphabot + 1}{\alphabot}\right)^d \cdot \tilde g^h_{t,\leq} \land \left(\frac{\alphabot + 1}{\alphabot}\right)^d \cdot t(\bar{W}) \geq \left(\frac{\alphabot}{\alphabot + 1}\right)^d \cdot \tilde g^h_{t,\geq}\big) \\
		{}\land{} & \biggl(\bigwedge_{\mathclap{\psi \in \CMSO_1(\sigma,q,\bar{W}\bar{Y})}} \quad \forall \bar{Y}  \begin{aligned}[t] & \psi(\bar{W}\bar{Y})\\
			{}\lor{} & \biggl(\left(\frac{\alphabot}{\alphabot + 1}\right)^d t_2(\bar{W}\bar{Y}) \leq \left(\frac{\alphabot + 1}{\alphabot}\right)^d \cdot \tilde g^h_{t_2,\psi,\leq} \\
&			\land  \left(\frac{\alphabot + 1}{\alphabot}\right)^d \cdot t_2(\bar{W}\bar{Y}) \geq \left(\frac{\alphabot}{\alphabot + 1}\right)^d \cdot \tilde g^h_{t_2,\psi,\geq}\biggr)\biggr).
		\end{aligned}
	\end{aligned}
	\end{align}

This formula has the same form as the formula~(\ref{phifinal}) obtained at the end of the first step, and hence we can reverse the last transformation that introduced the thresholds to the form of formula~(\ref{phibig}) to a formula which directly uses the values of $\tilde h_\geq (t\leq t')$ and $\tilde h_\leq (t\leq t')$,
\begin{align*}
\label{phialmostdone}
G \models 
& \bigwedge_{t \leq t' \in S}  \big(\left(\frac{b}{b+1}\right)^d t(\bar{X}) \leq \left(\frac{b+1}{b}\right)^d\tilde h_\geq(t \leq t') \land \left(\frac{b}{b+1}\right)^d\tilde h_\leq (t \leq t') \leq \left(\frac{b+1}{b}\right)^d t'(\bar{X})\big) \numberthis\\
				{}\land{} & \left(\frac{b}{b+1}\right)^d\tilde h_\leq (t_2\leq t_1) \leq \left(\frac{b+1}{b}\right)^d t_1(\bar{X}) \\
				{}\land{} &
				\forall \bar{Y} \Big(\chi_S^{t_2 \leq t_1}(\bar{X}\bar{Y}) \land \big(\left(\frac{b}{b+1}\right)^dt_2(\bar{X}\bar{Y}) \leq \left(\frac{b+1}{b}\right)^d\tilde h_\geq(t_2\leq t_1)  \lor \chi_S^{t_2 \not\leq t_1}(\bar{X}\bar{Y})\big)\Big).
				\end{align*}

Recalling that \(\tilde h_\leq(t \leq t') \geq \left(\frac{\alphabot + 1}{\alphabot}\right)^d \cdot \tilde h_\geq(t \leq t')\) by the definition of $\tilde h_\leq(t \leq t')$, each pair of weight comparisons $\left(\frac{b}{b+1}\right)^d t(\bar{X}) \leq \left(\frac{b+1}{b}\right)^d\tilde h_\geq (t \leq t') \land \left(\frac{b}{b+1}\right)^d\tilde h_\leq (t \leq t') \leq \left(\frac{b+1}{b}\right)^d t'(\bar{X})$ occurring in~(\ref{phialmostdone}) imply 

\[
\left(\frac{b}{b+1}\right)^d t(\bar{X}) \leq \left(\frac{b+1}{b}\right)^d \tilde h_\geq(t \leq t') \land \tilde h_\geq (t \leq t') \leq \left(\frac{b+1}{b}\right)^d t'(\bar{X}),
\]
which in turn is equivalent to
\[
\Bigl(\frac{b}{b+1}\Bigr)^{1.5d}t(\bar{X}) \leq \Bigl(\frac{b+1}{b}\Bigr)^{0.5d}\tilde h_\geq(t \leq t') \land \Bigl(\frac{b+1}{b}\Bigr)^{0.5d}\tilde h_\geq (t \leq t') \leq \Bigl(\frac{b+1}{b}\Bigr)^{1.5d}t'(\bar{X}).
\]

This in turn guarantees the existence of a number $g(t \leq t')$ such that $\big(\frac{b}{b+1}\big)^{2d}t(\bar{X}) \leq g(t \leq t') \land g(t \leq t') \leq \big(\frac{b+1}{b}\big)^{2d}t'(\bar{X})$. From here on, we proceed by reversing each of the transformations from formula~(\ref{phibig}) to the beginning of the first step (observe that in the step from~(\ref{phibig}) to~(\ref{phicleaned}), the granularity of $g(t\leq t')$ does not matter---the existence of such a number is sufficient in this direction). The final formula we obtain is then $\phi_{(1 + \frac{1}{\alphabot})^{2d}}$, and in particular we conclude $G\umodels[(1 + \frac{1}{\alphabot})^{2d}] \phi$ as desired.

	\paragraph{Running Time.}
	We iterate over all of the $|\roundedN[(1 + \frac{1}{\alphabot})^d,\lceil(1 + \frac{1}{\alphabot})^{2d} \cdot (N + 1)\rceil,\gran\cdot (\alphabot(\alphabot+1))^{3d}]|^{\mathcal{O}(|\terms_1|^2)}\cdot|\phi|^{\mathcal{O}(1)}$
	choices of $z: (\comp(\phi,\bar{X}) \cup \{t_2 \newbowtie t_1\}) \to \{\bot\} \cup \roundedN[(1 + \frac{1}{\alphabot})^d,\lceil(1 + \frac{1}{\alphabot})^{2d} \cdot (N + 1)\rceil,\gran \cdot (b(b+1))^{3d}]$, $z(t_2 \newbowtie t_1)\neq \bot$. For each such $z$, we construct the appropriate approximate table formula to be included in the disjunction in time polynomial in \(|\phi|\).
\end{proof}

	Apart from potential existential quantification, \Cref{thm:almostlookup} lets us derive a Boolean combination of entries of our table, which we can then evaluate.
	However, when it comes to the construction of witnesses for a satisfying assignment, a simple disjunction will be more useful. We first handle this in the next lemma, for which it will be useful to recall that oversatisfying by $c$ is equivalent to undersatisfying by $\frac{1}{c}$.

	\begin{lemma}
		\label{lem:lookup2}
		Given a formula $\xi(\bar X)$ obtained via conjunctions and disjunctions from formulas in $\tilde T := \tilde T_{\alpha,N,\gran}(\sigma,q,\bar X,\bar Y, \terms_1,\terms_2)$, one can compute in time \(\mathcal{O}(|\xi|\cdot |\tilde T|^2)\) a set $\Omega \subseteq \tilde T$ such that for any graph \(G\), any \(\bar{W} \in \PP(V(G))^{|\bar{X}|}\) and any \(c > 0\)
		\[G \omodels[c] \xi(\bar{W}) \Longleftrightarrow \mbox{there is \(\omega \in \Omega\) such that }  G \omodels[c] \omega(\bar{W}).\]
	\end{lemma}
	\begin{proof}
        We proceed by structural induction on $\xi$.
		If $\xi \in \tilde T$, then the statement holds trivially by simply setting \(\Omega = \{\xi'\}\).
		If $\xi = \xi_1 \lor \xi_2$ or $\xi = \xi_1 \land \xi_2$,
		then by induction, we can compute sets $\Omega_1 \subseteq \tilde T$ and $\Omega_2 \subseteq \tilde T$
		such that for all $G$, $\bar W$ and $c>0$
		$$
		G \omodels [c] \xi_i(\bar W) \Longleftrightarrow \mbox{there is \(\omega \in \Omega_i\) such that } G \omodels[c] \omega(\bar{W}).
		$$
		
		Assume $\xi = \xi_1 \lor \xi_2$. In this case, we set $\Omega := \Omega_1 \cup \Omega_2$.
        Note that
		$G \omodels[c] \xi(\bar W)$
		is equivalent to $G \omodels[c] \xi_1(\bar W)$ or
		$G \omodels[c] \xi_2(\bar W)$. This holds if and only if there is \(\omega \in \Omega_1\) such that \(G \omodels[c] \omega(\bar{W})\) or there is \(\omega \in \Omega_2\) such that \(G \omodels[c] \omega(\bar{W})\).
		Equivalently, there exists \(\omega \in \Omega_1 \cup \Omega_2 = \Omega\) such that \(G \omodels[c] \omega(\bar{W})\), which justifies the correctness of our construction of $\Omega$ in this case.
		
		Assume $\xi = \xi_1 \land \xi_2$. In this case, we let $\Omega$ be the set of all formulas that can be obtained from tuples $(\omega_1,\omega_2) \in (\Omega_1 \times \Omega_2)$ as
		\begin{multline*}
			\bigwedge_{t \in \terms_1} 
			t(\bar{Z}\bar X) \le \min\bigl(\thres_{\omega_1}(t,\lle),\thres_{\omega_2}(t,\lle)\bigr) \land t(\bar{Z}\bar X) 
			\ge \max\bigl(\thres_{\omega_1}(t,\gge),\thres_{\omega_2}(t,\gge)\bigr) \land \\
			\bigwedge_{t \in \terms_2} \bigwedge_{\psi \in \cmsoset} \forall \bar Y \psi(\bar{Z}\bar X \bar Y) \lor \Bigl( 
			t(\bar{Z}\bar X\bar Y) \le \min\bigl(\thres_{\omega_1}(t,\psi,\lle),\thres_{\omega_2}(t,\psi\lle)\bigr) \land {} \\
			t(\bar{Z}\bar X\bar Y) \ge \max\bigl(\thres_{\omega_1}(t,\psi,\gge),\thres_{\omega_2}(t,\psi\gge)\bigr)
			\Bigr).
		\end{multline*}
		These formulas are contained in $\tilde T$ and equivalent to $\omega_1 \land \omega_2$. It is straightforward to verify that this equivalence remains invariant under \(c\)-tightened approximation because the construction commutes with constructing the \oversatformulatext[c].
        We start with
		$G \omodels[c]  \xi(\bar W)$
		if and only if $G \omodels[c] (\xi_1(\bar W) \wedge \xi_1(\bar W))$. 
		By choice of $\Omega_1,\Omega_2$, this is then equivalent to the existence of $\omega_1 \in \Omega_1$, $\omega_2 \in \Omega_2$ such that $G \omodels[c]  (\omega_1(\bar{W}) \wedge \omega_2(\bar{W}))$.
		Hence, we have equivalence to there being an $\omega\in \Omega$ such that $G \omodels[c] \omega(\bar{W})$, as desired.
		
		\paragraph{Running Time.}
        We can transform the formula in time \(\bigoh(|\xi|)\) into a tree with at most \(|\xi|\) many conjunctive or disjunctive nodes.
		Per node, the running time is dominated by the construction of \(\Omega\) in the case of a conjunction, requiring the combination of two approximate table formulas in time
		\(\mathcal{O}(|\tilde T|^2)\).
	\end{proof}

\begin{definition}\label{def:parameterizedblockmso}
Let $\text{\blockMSOone}(\gamma,\sigma,q,\bar X,\bar Y,\terms_1,\terms_2)$
be the set of all \blockMSOone-formulas obtainable via iterative conjunctions, disjunctions and existential quantification
from formulas in $\Block(\gamma,\sigma,q,\bar X,\bar Y,\terms_1,\terms_2)$ (as defined in \Cref{def:parameterizedblocks}).
\end{definition}

 \begin{theorem}
 	\label{thm:finallookup}
 	Given \(\phi(\bar X) \in \text{\blockMSOone}(\gamma,\sigma,q,\bar X,\bar Y,\terms_1,\terms_2)\) and \(\alphabot,d,N \in \N\), we can compute 
    a set $\Omega \subseteq \tilde T_{(1 + \frac{1}{\alphabot})^d,\lceil(1 + \frac{1}{\alphabot})^{2d} \cdot (N+1)\rceil,\gran\cdot (\alphabot(\alphabot+1))^{3d}}(\sigma,q,\bar{Z}\bar X,\bar Y, \terms_1,\terms_2)$, with \(|\bar{Z}| \leq |\phi|\), 
    such that for every graph \(G\) whose \((\terms_1 \cup \terms_2)\)-range is contained in \(\{0, \dotsc, N\}\) and for every $\bar W\in \PP(V(G))^{|\bar X|}$,
\begin{align*}
        G \omodels[(1 + \frac{1}{\alphabot})^{2d}] \phi(\bar{W}) & \implies  \mbox{\hspace{-.08cm}there is \(\omega \in \Omega\) and $\bar{A}\in \PP(V(G))^{|\bar{Z}|}$}
        \text{ such that }  G \omodels[(1+\frac{1}{b})^d] \omega(\bar{A}\bar{W}),\\
        G \umodels[(1 + \frac{1}{\alphabot})^{2d}] \phi(\bar{W}) & \Longleftarrow \mbox{
        there is \(\omega \in \Omega\) and $\bar{A}\in \PP(V(G))^{|\bar{Z}|}$}
        \text{ such that }  G \umodels[(1+\frac{1}{b})^d] \omega(\bar{A}\bar{W}).
\end{align*}
    The computation takes time 
    $|\roundedN[(1 + \frac{1}{\alphabot})^d,\lceil(1 + \frac{1}{\alphabot})^{2d} \cdot (N + 1)\rceil,\gran\cdot (\alphabot(\alphabot+1))^{3d}]|^{\bigoh(|\terms_1|^2 + |\terms_2|\cdot |\CMSO_1(\sigma,q,\bar Z\bar X\bar Y)|)}
    \cdot |\phi|^{\bigoh(1)}$.
 \end{theorem}
 \begin{proof}
    Note that \(\phi\) is obtained from block formulas via iterative conjunctions, disjunctions and existential quantification.
    As a first step, we remove existential quantification by rewriting \(\phi(\bar{X})\), leaving only conjunctions and disjunctions:
    First, we introduce a tuple of variable symbols \(\bar{Z}\), \(|\bar Z| \le |\phi|\), as existentially quantified variables which are pairwise distinct from entries in \(\bar{X}\) and unique for each of the existentially quantified vertex set variables in $\phi$.
    After this, we obtain \(\phi'(\bar{Z}\bar{X})\) by removing all existential quantification from \(\phi(\bar{X})\) and keeping the previously quantified variables free.
    Since existential quantification over a variable that is unique to its scope commutes with conjunction, disjunction and existential quantification over different variables, and this is invariant under scaling weight comparisons, we see that
    \[
        \text{ for all $c>0$: } G \omodels[c] \phi(\bar{W}) \Longleftrightarrow \textnormal{there is $\bar{A}\in \PP(V(G))^{|\bar{Z}|}$ such that } G \omodels[c]  \phi'(\bar{A}\bar{W}).
    \]
    It is therefore sufficient to construct \(\Omega\) such that
\begin{align*}
        G \omodels[(1 + \frac{1}{\alphabot})^{2d}] \phi'(\bar A\bar{W}) &\implies \mbox{\hspace{-.08cm}there is \(\omega \in \Omega\) such that }  G \omodels[(1+\frac{1}{b})^d] \omega(\bar A\bar{W}),\\
        G \umodels[(1 + \frac{1}{\alphabot})^{2d}] \phi'(\bar A\bar{W}) &\Longleftarrow \mbox{
        there is \(\omega \in \Omega\) such that }  G \umodels[(1+\frac{1}{b})^d] \omega(\bar A\bar{W}).
\end{align*}
    We substitute every block of \(\phi'(\bar Z\bar X)\) with the corresponding output of \Cref{thm:almostlookup}.
    Since conjunctions and disjunctions commute with over- and undersatisfaction, we obtain a formula \(\xi(\bar Z\bar X)\) such that
 	\begin{align*}
 	G \omodels[(1 + \frac{1}{\alphabot})^{2d}] \phi'(\bar A\bar{W}) & \implies G \omodels[(1 + \frac{1}{\alphabot})^d] {\xi}(\bar A\bar{W}),\\	
 	G \umodels[(1 + \frac{1}{\alphabot})^{2d}] \phi'(\bar A\bar{W}) & \Longleftarrow \hspace{0.23cm}G \umodels[(1 + \frac{1}{\alphabot})^d] {\xi}(\bar A\bar{W}).
 	\end{align*}
    We observe that \(\xi(\bar Z\bar X)\) is obtained via conjunctions and disjunctions from table formulas in
\(\tilde T := \tilde T_{(1 + \frac{1}{\alphabot})^d,\lceil(1 + \frac{1}{\alphabot})^{2d} \cdot N\rceil,\gran\cdot (\alphabot(\alphabot+1))^{3d}}(\sigma,q,\bar Z\bar X,\bar Y,\terms_1,\terms_2)\).
    Thus, we can apply \Cref{lem:lookup2} to obtain a set $\Omega \subseteq \tilde T$ such that for and any \(c > 0\),
    \[
        G \omodels[c] \xi(\bar A\bar{W}) ~ \Longleftrightarrow \mbox{ there is \(\omega \in \Omega\) such that } G \omodels[c] \omega(\bar A\bar{W}).
    \]
    This proves the correctness of the construction.
    The running time bound is obtained by combining the bound from \Cref{thm:almostlookup} (which is the source of the first additive term in the exponent) and the bound from \Cref{lem:lookup2} together with \Cref{obs:approxSize} (which is the source of the second additive term in the exponent).    
 \end{proof}

\section{Final Algorithm}
\label{sec:final}
Having laid the groundwork for both computing our dynamic programming table (\Cref{sec:tablecomp}) and extracting answers from it (\Cref{sec:lookup}), 
we can now combine both to finally answer \blockMSOone-queries.
As a first step, we simply combine \Cref{thm:computeTable} and \Cref{thm:finallookup}.
This will yield one ``half'' of an answer to a query.
Note that while \blockMSOone-queries technically only contain non-negative \emph{integers} as coefficients,
we will also speak of \emph{\gran-granular} \blockMSOone-queries, where we relax this assumption to non-negative numbers of granularity \gran.
Intuitively, \Cref{lem:asdfadsf} allows us to either correctly determine the tightening of a \blockMSOone\ formula cannot be satisfied, or provides a solution to the loosened formula that is not worse than any solution to the tightened formula.

\begin{lemma}\label{lem:asdfadsf}
    There is an algorithm that receives as input a \gran-granular \blockMSOone-query \((\phi(\bar X),\target)\)
    with \(\phi(\bar X) \in \text{\blockMSOone}(\gamma,\sigma,q,\bar X,\bar Y,\terms_1,\terms_2)\),
    a number \(\alphabot \in \N\),
    a graph $G$ whose $\terms_1 \cup \terms_2$-range is $\{0,\dots,N\}$, and a \(k\)-expression \(\chi\) of $G$ of depth \(d\),
    and either correctly outputs that
    $\textnormal{oversatisfy}(G,\phi,\target,(1+\frac{1}{\alphabot})^{2d}) = -\infty$,
    or outputs $\bar W \in \PP(V(G))^{|\bar{X}|}$ with
    \(G \umodels[(1+\frac{1}{\alphabot})^{2d}] \omega(\bar W)\)
    and \(\textnormal{oversatisfy}(G, \phi, \target, (1+\frac{1}{\alphabot})^{2d}) \le \target(\bar W)\).
    There exists a computable function $f$ such that the running time of this algorithm is bounded by
    $$
    |\chi| \cdot |\roundedN[1+\frac{1}{\alphabot},\lceil(1+\frac{1}{\alphabot})^2 \cdot (N+1)\rceil,\gran \cdot(b(b + 1))^{4d^2+7d}]|^{f(|\phi|+k)}.
    $$
\end{lemma}
\begin{proof}

    We define two approximate tables
    \begin{align*}
        \tilde T^d & := \tilde T_{(1+\frac{1}{\alphabot})^d,\lceil(1+\frac{1}{\alphabot})^{2d} \cdot (N + 1)\rceil,\gran \cdot(b(b+1))^{3d}}(\sigma,q,\bar Z \bar X,\bar Y,\terms_1,\terms_2), \\
        \tilde T & := \tilde T_{1+\frac{1}{\alphabot},\lceil(1+\frac{1}{\alphabot})^{2d} \cdot (N + 1)\rceil,\gran \cdot(b(b+1))^{3d}}(\sigma,q,\bar Z \bar X,\bar Y,\terms_1,\terms_2).
    \end{align*}
    Since one accuracy is a power of the other and all other parameters are the same, we have \(\tilde T^d \subseteq \tilde T\).
Invoking \Cref{thm:computeTable} gives us a function $\witness_G$ 
such that for all \(\omega \in \tilde T\),
\begin{equation}\label{eq:final1}
    \textnormal{oversatisfy}\left(G, \omega, \target, \left(1+\frac{1}{\alphabot}\right)^d\right) \le \target\bigl(\witness_G(\omega)\bigr),
\end{equation}
\begin{equation}\label{eq:final2}
    \text{if } \witness_G(\omega) \neq \bot \text{ then } G \umodels[(1+\frac{1}{\alphabot})^d] \omega\bigl(\witness_G(\omega)\bigr),
\end{equation}
where $\target(\bot)$ is defined as $-\infty$.

Since the $\terms_1 \cup \terms_2$-range of $G$ is contained in $\{0,\dots,N\}$,
we can invoke \Cref{thm:finallookup} to compute a set $\Omega \subseteq \tilde T^d$
such that for all $\bar W \in \PP(V(G))^{|\bar X|}$
\begin{equation}\label{eq:finalx}
        G \omodels[(1 + \frac{1}{\alphabot})^{2d}] \phi(\bar{W}) \Longrightarrow \mbox{
        there is \(\omega \in \Omega\), $\bar{A}\in \PP(V(G))^{|\bar{Z}|}$}
        \text{ such that }  G \omodels[(1+\frac{1}{b})^{d}] \omega(\bar{A}\bar{W}),
\end{equation}
\begin{equation}\label{eq:final3}
        G \umodels[(1 + \frac{1}{\alphabot})^{2d}] \phi(\bar{W}) \Longleftarrow \mbox{
        there is \(\omega \in \Omega\), $\bar{A}\in \PP(V(G))^{|\bar{Z}|}$}
        \text{ such that }  G \umodels[(1+\frac{1}{b})^d] \omega(\bar{A}\bar{W}).
\end{equation}

The optimization target \(\target\) does not weight \(\bar Z\) and thus optimizing over it amounts to a merely existential quantification over $\bar{A}\in \PP(V(G))^{|\bar{Z}|}$.
Hence, (\ref{eq:finalx}) yields
\begin{equation}\label{eq:final4}
\textnormal{oversatisfy}\left(G,\phi,\target,\left(1+\frac{1}{\alphabot}\right)^{2d}\right) \le
\max_{\omega \in \Omega}~~ \textnormal{oversatisfy}\left(G,\omega,\target,\left(1+\frac{1}{\alphabot}\right)^{d}\right).
\end{equation}

We have \(\tilde T^d \subseteq \tilde T\) and thus \(\Omega \subseteq \tilde T\), meaning that \(\witness_G(\omega)\) is defined for every \(\omega \in \Omega\).
Choose $\bar W \in \{\witness_G(\omega) \mid \omega \in \Omega\} \cup \{\bot\}$ such that
\begin{equation}\label{eq:final5}
    \target(\bar W) = \max_{\omega \in \Omega}~~ \target(\witness_G(\omega)).
\end{equation}
Combining (\ref{eq:final1}), (\ref{eq:final4}) and (\ref{eq:final5}) yields
$$
\textnormal{oversatisfy}\left(G,\phi,\target,\left(1+\frac{1}{\alphabot}\right)^{2d}\right) \le
\max_{\omega \in \Omega}~~ \textnormal{oversatisfy}\left(G,\omega,\target,\left(1+\frac{1}{\alphabot}\right)^{d}\right) 
\le \target(\bar W).
$$

If $\bar W = \bot$, then this means that $\target(\bar W) = - \infty$. We can therefore correctly return that 
$\textnormal{oversatisfy}(G,\phi,\target,(1+\frac{1}{\alphabot})^{2d}) = -\infty$.
If $\bar W \neq \bot$, then by definition of $\bar W$, there exists $\omega \in \Omega$ with $\bar W = \witness_G(\omega)$.
Hence, by (\ref{eq:final2}), 
$G \umodels[(1+\frac{1}{\alphabot})^d] \omega(\bar W)$.
Thus, (\ref{eq:final3}) implies
$G \umodels[(1+\frac{1}{\alphabot})^{2d}] \phi(\bar W)$.

\paragraph{Running Time.}
The running time of this procedure is dominated by the running time of \Cref{thm:computeTable} and \Cref{thm:finallookup}.
\Cref{thm:computeTable}, applied to the table \(\tilde T\), runs in time
$$
|\chi| \cdot f(|\sigma|,k,q,|\bar X\bar Y|) \cdot |\roundedN[1+\frac{1}{\alphabot},\lceil(1+\frac{1}{\alphabot})^2 \cdot (N+1)\rceil,\gran \cdot (b(b+1))^{4d^2+7d}]|^{3|\terms_1| + |\terms_2|\cdot f'(|\sigma|,k,q,|\bar X\bar Y|)},
$$
for some computable function $f'$. Moreover, \Cref{thm:finallookup} takes time 
\[
    |\roundedN[1 + \frac{1}{\alphabot},\lceil(1 + \frac{1}{\alphabot})^2 \cdot (N + 1)\rceil,\gran\cdot (\alphabot(\alphabot+1))^3]|^{\bigoh(|\terms_1|^2 + |\terms_2|\cdot |\CMSO_1(\sigma,q,\bar Z\bar X\bar Y)|)} \cdot |\phi|^{\bigoh(1)}.
\]
We can assume $|\sigma|,q,|\bar Z|,|\bar X|,|\bar Y|,|\terms_1|,|\terms_2|$ to be bounded in terms of $|\phi|$
and thus choose (with \Cref{obs:msosizebound}) a function $f$ bounding the total running time by
\[
|\chi| \cdot |\roundedN[1+\frac{1}{\alphabot},\lceil(1+\frac{1}{\alphabot})^2 \cdot (N+1)\rceil,\gran \cdot (b(b+1))^{4d^2+7d}]|^{f(|\phi|+k)}. \qedhere
\]
\end{proof}

We now work toward our targeted fixed-parameter running time, describing the running time in terms of our accuracy \(0 < \eps \leq 0.5\).

\begin{lemma}\label{lem:asdfadsfEps}
    There exists an algorithm that receives as input
    a \gran-granular \blockMSOone-query $(\phi(\bar X),\target)$, an accuracy \(0.5 \geq \eps > 0\), 
    a graph $G$, and
    a \(k\)-expression \(\chi\) of $G$ of depth $d$,
    and either
    correctly outputs that
    $\textnormal{oversatisfy}(G,\phi,\target,1+\eps) = -\infty$
    or outputs $\bar W \in \PP(V(G))^{|\bar{X}|}$ with
    \(G \umodels[1+\eps] \phi(\bar W)\)
    and \(\textnormal{oversatisfy}(G, \phi, \target, 1+\eps) \le \target(\bar W)\).
    There exists a computable function $f$ such that the running time of this algorithm is bounded by
    $$
   |\chi| \cdot \left(\frac{d \cdot \log(\gran+2) \cdot \log(M+|V(G)|)}{\eps}\right)^{f(|\phi|,k)},
    $$
    where $f$ is a computable function and $M$ is two plus the highest number that occurs in any weight term of $\phi$ or as any weight in $G$.
\end{lemma}
\begin{proof}
    Let $n = |V(G)|$ and define $N := \MM + |\phi|^2 \cdot \MM^2 \cdot n$.
    When $\terms$ are the size terms of $\phi$, then by \Cref{lem:boundTermsRange}, the $\terms$-range of $G$ is contained in $\{0,\dotsc,N\}$.
    Let \(\alphabot \in \N\) be minimal such that \(\frac{1}{\alphabot} \leq \frac{\eps}{5d}\).
    This implies \(\frac{1}{\alphabot} \geq \frac{\eps}{10d}\), as otherwise \(\alphabot/2\) or \((\alphabot + 1)/2\) (depending on the parity of \(\alphabot\)) would be a smaller valid choice.
    We apply \Cref{lem:asdfadsf} and obtain
    either the answer that
    $\textnormal{oversatisfy}(G,\phi,\target,(1 + \frac{1}{\alphabot})^{2d}) = -\infty$,
    or receive $\bar W \in \PP(V(G))^{|\bar{X}|}$ such that
    \(G \umodels[(1 + \frac{1}{\alphabot})^{2d}] \phi(\bar W)\)
    and \(\textnormal{oversatisfy}(G, \phi, \target, (1 + \frac{1}{\alphabot})^{2d}) \le \target(\bar W)\).

    We use the well-known bound $(1+\frac{x}{n})^n \le e^x$, and also $e^{x/2} \le 1 + x$ for $0 \le x \le 1/2$ (here it is important that \(\eps \leq 0.5\)) to obtain
    $$
    \left(1 + \frac{1}{\alphabot}\right)^{2d} \leq \left(1+\frac{\eps/2}{2d}\right)^{2d} \le e^{\eps/2} \le 1 + \eps.
    $$
    This means the algorithm either correctly outputs that
    $\textnormal{oversatisfy}(G,\phi,\target,1+\eps) = -\infty$,
    or yields $\bar W \in \PP(V(G))^{|\bar{X}|}$ with
    \(G \umodels[1+\eps] \phi(\bar W)\)
    and \(\textnormal{oversatisfy}(G, \phi, \target, 1+\eps) \le \target(\bar W)\).
    This is what we set out to do.

    Let us analyze the running time. \Cref{lem:asdfadsf} runs in time
    \[
        d \cdot |\chi| |\roundedN[1 + \frac{1}{\alphabot},\lceil(1 + \frac{1}{\alphabot})\cdot (N+1)\rceil,\gamma \cdot (\alphabot(\alphabot + 1)^{4d^2 + 7d})]|^{f'(|\phi|+k)}
    \]
    for some computable function $f'$.
    Since \(\eps \leq 0.5\) and $\log(1+x) \ge x/2$ for all $0 \le x \le 1/2$, 
    $$
    \log\left(1 + \frac{1}{\alphabot}\right) \geq \log\left(1+\frac{\eps}{10d}\right) \ge \frac{\eps}{20d}.
    $$
    Using $|\phi|\ge 1$ and \(\log(x^{{\bigoh(1)}})\le\bigoh(\log(x))\) for all \(x\), we bound
    \[
    \log(N) = \log\bigl(\MM + |\phi|^2 \cdot \MM^2 \cdot n\bigr) \le 
    \bigoh\bigl(|\phi| \cdot \log\bigl(\MM + n)\bigr).
    \]
    
    Hence, using \Cref{obs:tildeNsize}, \(1/\log(1+\frac{1}{b}) \le \frac{20d}{\eps}\) and \(b \le \frac{10d}{\eps}\),
    \begin{multline*}
	    |\roundedN[1 + \frac{1}{\alphabot},\lceil(1 + \frac{1}{\alphabot})^2 \cdot (N+1)\rceil,\gran \cdot (b(b+1))^{4d^2+7d}]| \le \mathcal{O}\left(\log\Biggl(\biggl(1 + \frac{1}{\alphabot}\biggr)^2 \cdot (N+1) \cdot \gran \cdot \bigl(b(b+1)\bigr)^{4d^2+7d}\Biggr)\right) \\
        \le \mathcal{O}\Biggl(\bigl(\log(N)+\log(\gran)+bd^2\bigr)/\log\biggl(1 + \frac{1}{\alphabot}\biggr)\Biggr) 
	    \le \mathcal{O}\Biggl(\frac{d^4 \cdot \log(\gran+2) \cdot |\phi| \cdot \log(\MM+n)}{\eps^2}\Biggr).
	\end{multline*}
    Therefore, there exists a computable function $f$ bounding the previously stated running time of \Cref{lem:asdfadsf} by
\[
    |\chi| \cdot
    \mathcal{O}\left(\frac{d \cdot \log(\gran+2) \cdot \log(\MM+n)}{\eps}\right)^{f(|\phi|+k)}.\qedhere
\]
\end{proof}

The previous lemma gives us ``half'' of an answer to a query, and does so with ``half'' the accuracy than one claims in our meta-theorems.
To prove our main result, we evoke it twice with slightly shifted input queries.

	\begin{theorem}[Base Meta-Theorem]
		\label{thm:maintheorem}
		Given a \blockMSOone-query \((\phi,\target)\), an accuracy \(0 < \eps \leq 0.5\),
		a graph $G$ with matching signature and a corresponding \(k\)-expression \(\chi\) of depth $d$,
		we can compute a \((1+\eps)\)-approximate answer to \((\phi,\target)\) on \(G\) in time 
		$$
		|\chi| \cdot \left(\frac{d \cdot \log(|V(G)|+\MM)}{\eps}\right)^{f(|\phi|,k)},
		$$
		where $f$ is a computable function
        and $\MM$ is the highest number that occurs in any weight term of $\phi$ or as any weight in $G$.
	\end{theorem}

\begin{proof}
    For $0 < \eps \leq 0.5$, we have $1+\eps \ge (1+\eps/3)^2$.
Thus, for all $p,q$ that are natural weight terms over the free variables of the query, and ${\vartriangleleft} \in \{\le,<\}$,
    $$
    G \models p \vartriangleleft q 
    \iff G \umodels[1+\eps/3] p \cdot (1+\eps/3) \vartriangleleft q / (1+\eps/3),
    $$
    $$
    G \omodels[1+\eps] p \vartriangleleft q \Longrightarrow
    G \omodels[(1+\eps/3)^2] p \vartriangleleft q 
    \iff G \omodels[1+\eps/3] p \cdot (1+\eps/3) \vartriangleleft q / (1+\eps/3).
    $$
    Thus, replacing every comparison $p \vartriangleleft q$ with $p \cdot (1+\eps/3) \vartriangleleft q / (1+\eps/3)$,
    we obtain from $\phi$ a \blockMSOone-formula $\phi^-$ such that
    for every $G$ and $\bar W \in \PP(V(G))^{|\bar X|}$,
    $$
    G \models \phi(\bar W) \iff G \umodels[1+\eps/3] \phi^-(\bar W),
    $$
    $$
    G \omodels[1+\eps] \phi(\bar W) \Longrightarrow G \omodels[1+\eps/3]  \phi^-(\bar W).
    $$

    By slightly rescaling \(\eps\), we can assume without loss of generality that \(\eps = 3/b\) for some \(b \in \N\).
    Thus, when going from $\phi$ to \(\phi^-\), all numbers are either scaled by \(\frac{b+1}{b}\) or by \(\frac{b}{b+1}\).
    Note that \(\phi\) has granularity \(1\) and thus \(\phi^-\) has granularity \(\gran := b(b+1)\).

    Applying \Cref{lem:asdfadsfEps} to $\phi^-$ with accuracy $\eps/3$ and granularity \(\gran\), we either get the answer that
    $\textnormal{oversatisfy}(G,\phi^-,\target,1+\eps/3) = -\infty$
    or we obtain $\bar W^- \in \PP(V(G))^{|\bar{X}|}$ with
    \(G \umodels[1+\eps/3] \phi^-(\bar W^-)\)
    and \(\textnormal{oversatisfy}(G, \phi^-, \target, 1+\eps/3) \le \target(\bar W^-)\).
    Since, $G \omodels[1+\eps] \phi(\bar W^-) \iff G \omodels[1+\eps/3] \phi^-(\bar W^-)$,
    we have
    $\textnormal{oversatisfy}(G,\phi,\target,1+\eps) = \textnormal{oversatisfy}(G,\phi^-,\target,1+\eps/3)$.
    Thus, by construction of $\phi^-$, 
    \begin{enumerate}[label=\alph*),series=sandwhich]
        \item\label{itm:aa} we either know that $\textnormal{oversatisfy}(G,\phi,\target,1+\eps) = -\infty$, or
        \item\label{itm:bb} we obtain $\bar W^-$ with \(G \models \phi(\bar W^-)\) and \(\textnormal{oversatisfy}(G,\phi,\target,1+\eps) \le \target(\bar W^-)\).
    \end{enumerate}
    Analogously to the beginning of the proof, we notice
    $$
    G \models p \vartriangleleft q 
    \iff G \omodels[1+\eps/3] p / (1+\eps/3) \vartriangleleft q \cdot (1+\eps/3),
    $$
    $$
    G \umodels[1+\eps] p \vartriangleleft q \Longleftarrow
    G \umodels[(1+\eps/3)^2] p \vartriangleleft q 
    \iff G \umodels[1+\eps/3] p / (1+\eps/3) \vartriangleleft q \cdot (1+\eps/3).
    $$
    Thus, as before, we obtain a \blockMSOone-formula $\phi^+$ such that
    for every $G$ and $\bar W \in \PP(V(G))^{|\bar X|}$,
    $$
    G \models \phi(\bar W) \iff G \omodels[1+\eps/3] \phi^+(\bar W),
    $$
    $$
    G \umodels[1+\eps] \phi(\bar W) \Longleftarrow G \umodels[1+\eps/3]  \phi^+(\bar W).
    $$
    Applying \Cref{lem:asdfadsfEps} to $\phi^+$ (again with accuracy $\eps/3$ and granularity $\gran$), we either get the answer that
    $\textnormal{oversatisfy}(G,\phi^+,\target,1+\eps/3) = -\infty$
    or we obtain $\bar W^+ \in \PP(V(G))^{|\bar{X}|}$ with
    \(G \umodels[1+\eps/3] \phi^+(\bar W^+)\)
    and \(\textnormal{oversatisfy}(G, \phi^+, \target, 1+\eps/3) \le \target(\bar W^+)\).
    By the same argument as above, 
    \begin{enumerate}[label=\alph*),resume=sandwhich]
        \item\label{itm:cc} we either know that $\val(G,\phi^+,\target) = -\infty$,
        \item\label{itm:dd} or we obtain $\bar W^+$ with \(G \umodels[1+\eps] \phi(\bar W^+)\) and \(\val(G, \phi^+, \target) \le \target(\bar W^+)\).
    \end{enumerate}

    \paragraph{Query Answer.}
    According to \Cref{def:approxAnswer},
    a \emph{$(1+\eps)$-approximate answer} to \((\phi,\target)\) on $G$
    includes values $\val^-, \val^- \in \N$ such that
    \begin{equation}\label{eq:sandwhich}
        \textnormal{oversatisfy}(G, \phi, \target, 1+\eps) \le
        \val^- \le
        \val(G,\phi,\target) \le
        \val^+ \le
        \textnormal{undersatisfy}(G, \phi, \target, 1+\eps),
    \end{equation}
    as well as 
    $\bar W^-$ with $G \models \phi(\bar W^-)$  and $\target(\bar{W}^-) = \val^-$ (if \(\val^- \neq -\infty\)),
    and 
    $\bar W^+$ with $G \umodels[\alpha] \phi(\bar W^+)$  and $\target(\bar{W}^+) = \val^+$ (if \(\val^+ \neq -\infty\)).

    When \ref{itm:aa} holds, that is, $\textnormal{oversatisfy}(G, \phi, \target, 1+\eps) = -\infty$, we set $\val^- := -\infty$, 
    satisfying the upper and lower bound on $\val^-$, as required by in (\ref{eq:sandwhich}).
    When \ref{itm:bb} holds, we set $\val^- := \target(\bar W^-) \ge \textnormal{oversatisfy}(G, \phi, \target, 1+\eps)$
    and additionally return the witness $\bar W^-$ with $G \models \phi(\bar W^-)$  and $\target(\bar{W}^-) = \val^-$.
    Since $G \models \phi(\bar W^-)$, we have $\val^- = \target(\bar W^-) \le \val(G,\phi,\target)$,
    satisfying both requirements on $\val^-$ in (\ref{eq:sandwhich}).

    Similarly, when \ref{itm:cc} holds, that is, $\val(G, \phi, \target) = -\infty$, we set $\val^+ := -\infty$, 
    satisfying both requirements on $\val^+$ in (\ref{eq:sandwhich}).
    When \ref{itm:dd} holds, we set $\val^+ := \target(\bar W^+) \ge \val(G, \phi, \target)$
    and additionally return the witness $\bar W^+$ with $G \umodels[1+\eps] \phi(\bar W^+)$  and $\target(\bar{W}^+) = \val^+$.
    Since $G \umodels[1+\eps] \phi(\bar W^+)$, we have $\val^+ = \target(\bar W^+) \le \textnormal{undersatisfy}(G, \phi, \target, 1+\eps)$,
    satisfying both requirements on $\val^+$ in (\ref{eq:sandwhich}).
    This means in all cases we return an adequate answer.

    \paragraph{Running Time.}
    The running time of our algorithm is dominated by the two invocations of \Cref{lem:asdfadsfEps} with formulas $\phi^-,\phi^+$, accuracy $1+\eps/3$ and granularity \(\gran=b(b+1)\), where \(b = 3/\eps\).
    When constructing $\phi^-$ and $\phi^+$ from $\phi$, 
    the highest number that occurs in any weight term has increased by at most a factor of $1+\eps/3\le 2$.
    Furthermore, the formulas $\phi^-$, $\phi^+$ and $\phi$ all have the same length.
    Thus, the running time is bounded by
    $$2 \cdot |\chi| \cdot \left(\frac{d \cdot \log(\gran) \cdot \log(|V(G)| + 2\MM)}{\eps/3}\right)^{f'(|\phi|,k)}
    = |\chi| \cdot \left(\frac{d \cdot \log(|V(G)|+\MM)}{\eps}\right)^{f(|\phi|,k)},
    $$    
    for an appropriately chosen computable function $f$.
\end{proof}

\subsection{Proofs of Our Algorithmic Meta-Theorems}
\label{sec:metaproofs}
We conclude this section by showing how Theorem~\ref{thm:maintheorem} implies the more explicit formulation of our meta-theorems for cliquewidth and treewidth presented in \Cref{sec:metat}. To this end, we will make use of the following folklore observation to avoid parameter dependencies in the exponents of logarithmic terms.

		\begin{observation}
		\label{obs:runtime}
		For every \(\delta > 0\) and every function $g:\N^2\rightarrow \N$, there exists a function $f_\delta:\N^2\rightarrow N$ such that $\log(Z)^{g(a,b)} \leq f_\delta(a,b)\cdot Z^\delta$ holds for all \(Z,a,b \in \N\).
		\end{observation}
		
		\begin{proof}
		We begin by recalling the well-known fact that for every fixed $\ell\in \Nat$, $\frac{(\log Z)^\ell}{Z^\delta} \xrightarrow[Z\to \infty]~ 0$. Hence, for each $\ell$ there exists an $h_\delta(\ell)$ such that for each $n\geq h_\delta(\ell)$, $\frac{(\log n)^\ell}{n^\delta} <1$. Let $f_\delta(a,b):=(h_\delta(g(a,b)))^{g(a,b)}$. 
		To argue that the inequality holds, let us consider the following case distinction. If $Z\geq h_\delta(g(a,b))$, then $(\log Z)^{g(a,b)} = \frac{(\log Z)^{g(a,b)}}{Z^\delta}\cdot Z^\delta \leq 1\cdot Z^\delta$ by the definition of $h_\delta$. On the other hand, if $Z<h_\delta(g(a,b))$ then we immediately see that $Z^{g(a,b)}\leq h_\delta(g(a,b))^{g(a,b)}$, which in turn implies $(\log(Z))^{g(a,b)} \leq (h(g(a,b)))^{g(a,b)}$. At the same time, $(h(g(a,b)))^{g(a,b)}\leq (h(g(a,b)))^{g(a,b)} \cdot Z^\delta = f_\delta(a,b)\cdot Z^\delta$, and hence $(\log(Z))^{g(a,b)} \leq f_\delta(a,b)\cdot Z^\delta$ as desired.
\end{proof}

We can use Theorem~\ref{thm:LogarithmicDepth} to drop the requirement of receiving a $k$-expression on the input, and combine this with Observation~\ref{obs:runtime} to obtain our algorithmic meta-theorem for cliquewidth. We emphasize, once again, that the quadratic dependence on $|V(G)|$ stems from the quadratic running time of known algorithms for computing cliquewidth-expression; if a suitable decomposition is given, the dependency on the size of the graph can be reduced to, e.g., $|V(G)|^{1.001}$. 

\mainbase*

\begin{proof}
We first use Theorem~\ref{thm:LogarithmicDepth} to compute a 
\(2^{\cw(G) + 2}\)-expression of depth at most \(f'(\cw(G))\cdot\log(|V(G)|)\) and length at most \(f'(\cw(G))\cdot|V(G)|\) in time at most $f'(\cw(G)) \cdot |V(G)|^2$ for some function $f'$. Since $\MM\geq 2$ and we can assume $|V(G)|\geq 2$ as well, feeding this expression into Theorem~\ref{thm:maintheorem} results in a running time of at most
\[
f'(\cw(G)) \cdot |V(G)|^2 + (f'(k)\cdot|V(G)|) \cdot \left(\frac{\bigl(f'(\cw(G))\cdot \log(|V(G)|)\bigr) \cdot \log(\MM)}{\eps}\right)^{f^\star(|\phi|,2^{\cw(G) + 2})},
\]
where $f^\star$ is the function arising in the running time bound in Theorem~\ref{thm:maintheorem}. By setting $f''$ to a sufficiently generous combination of $f^\star$ and $f'$ and reordering the terms in the fraction, we can reformulate the second part of the bound as
\[
    f'(\cw(G)) \cdot |V(G)|^2 + \left(\frac{\log(\MM)}{\eps}\right)^{f''(|\phi|,\cw(G))}\cdot \log(V(G))^{f^\star(|\phi|,2^{\cw(G)+2})} \cdot |V(G)|.
\]

We can now invoke Observation~\ref{obs:runtime} with $Z:=V(G)$ and $\delta:=0.001$ to identify a function $f^\circ$ such that $\log(|V(G)|)^{f^\star(|\phi|,2^{\cw(G)+2})}\leq {f^\circ(|\phi|,\cw(G))}\cdot |V(G)|^{{0.001}}$. By setting $f_1$ to suitably subsume $f'$, $f''$ and $f^\circ$, we obtain the first claimed running time of

\[
\left(\frac{\log(\MM)}{\eps}\right)^{f_1(|\phi|,\cw(G))} |V(G)|^2.
\]

For the second bound, we invoke Observation~\ref{obs:runtime} with $Z:=\MM$ and $\delta:=0.001$ and aggregate the resulting function to obtain a function $f_2$ such that the running time above can be upper-bounded by
\[
\left(\frac{1}{\eps}\right)^{f_2(|\phi|,\cw(G))}\cdot 
\MM^{0.001}
\cdot |V(G)|^2. \qedhere
\]
\end{proof}

\mainbasetw*
	\begin{proof}
		It is well known (and easy to see) that the interpretation of any \CMSO\(_2\)-formula \(\psi\) on an edge-colored graph \(G\) can be rewritten as a \CMSO\(_1\)-formula $\psi'$ with length proportional to the length of \(\psi\) on a graph $G'$ obtained as the subdivision of \(G\), in which the newly created vertices are marked by an auxiliary color. Indeed, one can simply replace each edge (set) variable with a vertex (set) variable and use the auxiliary color to distinguish vertices in $G$ from edges in $G$. The analogous statement is true for \blockMSOtwo-formulas on edge-colored edge-weighted graphs, where the subdivision vertices inherit the weight of the edge they subdivide.
        The construction does not have any impact on the highest number occurring in the weight terms of the formula or as weight on the graph. If in the transformation \(\phi'\) of \(\phi\) we use the same variable names, this transformation does not impact the interpretation of any weight terms in the free variables of \(\phi\) (which are then the same as the free variables in \(\phi'\)).
		This means we can rewrite \((\phi,\target)\) on \(G\) as a \blockMSOone-query \((\phi',\target)\) on the subdivision \(G'\) of \(G\) with an additional color and edge-weights reinterpreted as vertex-weights. 
        Thus, a \((1+\eps)\)-approximate answer to \((\phi',\target)\) on $G'$ is equivalent to a \((1+\eps)\)-approximate answer to \((\phi,\target)\) on $G$.

		We can use \Cref{obs:subdivision} and \Cref{thm:LogarithmicDepthtw} to obtain an \(f'(\tw(G))\)-expression of \(G'\) for some computable \(f'\), which can then be fed into Theorem~\ref{thm:maintheorem} in the same way as in the proof of Theorem~\ref{thm:mainapproxcw}. The rest of the proof then follows analogously to the proof of Theorem~\ref{thm:mainapproxcw}. The running time required to invoke \Cref{thm:LogarithmicDepthtw} only has a linear dependence on $|V(G)|$, which means that in the step where we introduce $f_1$, the running time dependence on $|V(G)|$ can be upper-bounded by $|V(G)|^{1.001}$ as opposed to $|V(G)|^{2}$. 
        The difference between the length of \(\phi\) and \(\phi'\) and the additional color can all be accommodated by appropriately increasing the functions $f_1$ and \(f_2\).
	\end{proof}

By choosing an instance-dependent accuracy, we can derive a slicewise polynomial results for non-approximate answers to \blockMSO-queries. 
\exact*
\begin{proof}
	Let us fix \((\phi,\target)\) and \(G\) as in the theorem statement.
	We begin by rewriting the constraint \(\phi\) to not include any strict inequalities in weight comparisons by replacing all weight comparisons of the form \(t < t'\) by the logically equivalent comparisons \(t + 1 \leq t'\), and analogously all weight comparisons of the form \(t > t'\) by the logically equivalent comparisons \(t \geq t' + 1\).
	Note that this increases the highest number that occurs in any weight term of $\phi$ to at most \(\MM + 1\).
	Let \(N\) be minimal such that for the free variables \(\terms\) of \(\phi\), the \(\terms\)-range of \(G\) is contained in \(\{0, \dots, N\}\).
	According to \Cref{lem:boundTermsRange}, \(N\) is bounded by \((\MM + |V(G)|)^{g(|\phi|)}\) for some computable function \(g\).
	
	We set \(\eps := \frac{1}{2N}\) and apply \Cref{thm:mainapproxcw} (or \Cref{thm:mainapproxtw}) to find a \((1+\eps)\)-approximate answer \((\bar{W}^-,\val^-,\bar{W}^+,\val^+)\) to \((\phi,\target)\) on \(G\), and return \(\bar W^+\) if \(\val^+ \neq -\infty\), and decide that no solution exists otherwise.
	Notice that by inserting our choice of \(\eps\) into the running time bound from \Cref{thm:mainapproxcw} (\Cref{thm:mainapproxtw}) and choosing \(f\) appropriately, we arrive at a running time bound of 
	\[(\MM+|V(G)|)^{f(|\phi|,k)}\] for some computable function \(f\), as desired.

	In the remainder of the proof, we will argue that \(G \models \phi(\bar{W})\) if and only if \(G \umodels[1 + \frac{1}{2N}] \phi(\bar{W})\) for an arbitrary interpretation \(\bar{W}\) in \(G\) of the free variables in \(\phi\).
	After doing so, we have \(\val(G,\phi,\target) = \textnormal{undersatisfy}(G, \phi, \target, 1 + \frac{1}{2N})\), and hence both of these are equal to \(\max^+\). This means that if \(\val^+ \neq -\infty\), then \(\bar W^+\) is indeed a (non-approximate) answer to \((\phi,\target)\) on \(G\), and if \(\val^+ = \val = - \infty\) then indeed there is no solution for \(\phi\) on \(G\).
	In particular, we return an exact answer to the query.

It suffices to show that \(G \umodels[1 + \frac{1}{2N}] \phi(\bar{W})\) implies \(G \models \phi(\bar{W})\), as the converse is true by equation~(\ref{eq:chain1}) in \Cref{sub:approx}.
	Consider a natural weight comparison of the form \(\psi:=t(\bar{W}) \leq t'(\bar{W})\).
	Now if \(G \umodels[1 + \frac{1}{2N}] \psi(\bar{W})\) then \(\frac{2N}{2N + 1}t^G(\bar{W}) \leq \frac{2N+1}{2N}{t'}^G(\bar{W})\).
	Since \(t^G(\bar{W})\) and \({t'}^G(\bar{W})\) must be numbers that are strictly smaller than $N+1$,
	\[t(\bar{W}) - 0.5 < \frac{2N}{2N + 1}t(\bar{W}) \leq \frac{2N+1}{2N}t'(\bar{W}) \leq t'(\bar{W}) + 0.5.\]
	Hence,
	\[t^G(\bar{W}) - 1 < {t'}^G(\bar{W}),\]
	and because these weight terms are interpreted as natural numbers this implies that
	\[t^G(\bar{W}) \leq {t'}^G(\bar{W}).\]
	Altogether we have showed that \(G \umodels[1 + \frac{1}{2N}] \psi(\bar{W})\) implies \(G \models \psi(\bar{W})\).
	
	Since each weight comparison in \(\phi\) occurs non-negated (due to the assumption of negation normalization for formulas in \MSOcomp\ described in Subsection~\ref{sec:defExtendedLogic}), we have that \(G \umodels[1 + \frac{1}{2N}] \phi(\bar{W})\) implies \(G \models \psi(\bar{W})\), as desired.	
\end{proof}

    \section{Hardness for More General Formulas}\label{sec:hardness}
	A natural question arising from our meta-theorems and the accompanying logic \blockMSO\ is whether the restrictions imposed on the formula are truly necessary. 
    Our definition of \blockMSO\ enforces that outermost universally quantified variables may appear only in weight terms on a single side of a single comparison (for each block).
    Below we show that even a slight relaxation of this already results in intractability under the usual complexity-theoretic assumptions.
     In particular:

    \begin{enumerate}
        \item As soon as we allow outermost universally quantified variables to appear in two comparisons, approximation is not \FPT even when restricted to paths (\Cref{thm:whard-notnice}).
        \item We also rule out a similar \XP\ result for exact evaluation on trees (\Cref{thm:nphard-notnice}).
    \end{enumerate}
    Note that \Cref{thm:whard-notnice} shows hardness on paths and \Cref{thm:nphard-notnice} shows hardness
    on trees of depth three. With minor modifications, both proofs can be adapted to show hardness for both paths and trees of depth three.
    These results trivially extend to comparisons involving set variables of deeper alternation depth,
    as for example \(\forall Y \psi(Y)\) is equivalent to \(\forall Y \exists Z (Y=Z) \land \psi(Z)\).

    To formally state our hardness results, we extend our notion of approximation for
    \blockMSO-queries (\Cref{sec:bMSO} and \Cref{def:approxAnswer}) as one
    would expect to the more general setting of \MSOcomp-queries. 
    The proof proceeds by constructing graphs \(G\) and formulas \(\phi(X)\) that encode hard problems and are robust with respect to \(\alpha\)-loosening or \(\alpha\)-tightening, that is, for all inputs \(W\),
    \(G \models \phi(W) \Longleftrightarrow G\models \oversatformula(W) \Longleftrightarrow G \models \undersatformula(W)\). We remark that the values for the accuracy $\alpha$ used in following lower bound are larger than the values $(1+\eps)$ considered in our meta-theorems
    and that lower bounds for large numbers imply lower bounds for small numbers.
    We thus exclude fixed-parameter tractability in the style of Theorems~\ref{thm:mainapproxcw},~\ref{thm:mainapproxtw}, even for much worse accuracy.

	\begin{theorem}
		\label{thm:whard-notnice}
        For every constant accuracy $\alpha \in \N$ with \(\alpha > 1\),
        it is \textnormal{\W[1]}-hard to find \(\alpha\)-approximate answers to \MSOcompone-queries \((\phi(X),0)\) parameterised by \(|\phi(X)|\) on the class of all vertex-colored paths,
        even if \(\phi(X)\) is required to have the form
		\[
            \forall Y_1\forall Y_2\bigl(|Y_1| > |Y_2|\cdot \alpha^2 \lor |Y_2| > |Y_1|\cdot \alpha^2 \lor \psi(X Y_1 Y_2)\bigr)
		\]
        where \(\psi(X Y_1 Y_2)\) is a \mso$_1$-formula.
	\end{theorem}
	\begin{proof}
		We reduce from the following \W[1]-hard problem~\cite[Theorem~14.28]{cygan2015parameterized}.
			\mybox{
				\textsc{Grid Tiling} \\
				
				\emph{Input:} $n,k \in \N$ encoded in unary, $S_{i,j} \subseteq [n] \times [n]$ for each pair \((i,j) \in [k] \times [k]\). \\
				\emph{Parameter:} $k$ \\
				\emph{Question:} Is there a set \(\{s_{i,j} \in S_{i,j}\mid i,j \in [k]\}\) such that
				\begin{itemize}
					\item if \(s_{i,j} = (a,b)\) then there is some \(b' \in [k]\) such that  \(s_{i + 1, j} = (a,b')\), and
					\item if \(s_{i,j} = (a,b)\) then there is some \(a' \in [k]\) such that  \(s_{i, {j + 1}} = (a',b)\).
				\end{itemize}}

                Intuitively, this problem asks us to select a pair of numbers from each grid cell such that in each column the first coordinates agree and on each row the second coordinates agree.

        Assume $\alpha$ to be an arbitrary, fixed constant. In the proof, we describe how to construct a path \(G\) and a formula \(\phi(X)\) from a given instance of \textsc{Grid Tiling}, and then verify that the pair $(G,\phi)$ meets the requirements to show the theorem statement.
		Let us first explain the general idea of the reduction.
		As a first step, we observe how each pair of numbers in \(S_{i,j}\) can be encoded using a \emph{bitstring path} containing \(\lceil \log(n) \rceil\) many ``bit'' vertices, colored according to the binary encodings of the numbers.
        Using \(k^2\) additional colors, we can identify for each bitstring path the row \(i\) and column \(j\) that the pair of numbers comes from.
        Our formula $\phi(X)$ will be constructed in a way which forces $X$ to contain precisely one bitstring path for each grid cell $(i,j)$, giving us a one-to-one correspondence between $X$ and a ``solution set'' for \textsc{Grid Tiling}.
		Hence, what is left to do is to compare that two bitstring paths actually describe the same number (in their first or second coordinates).
	    By spacing the ``bit'' vertices along the bitstring path exponentially apart, it suffices to check whether two ``bit'' vertices have roughly the same ``offset'' 
        along their bitstring path. This property can be expressed in a way that is robust with respect to approximation.
		We now proceed to the formal construction and arguments.
		
        \paragraph*{Encoding of Numbers.}
        We encode numbers \(s \in [n]\) as zero-padded binary bit strings \(x_{\lceil \log(n)\rceil},\dots,v_{x_1}\), where \(x_1\) is the least significant bit and \(x_{\lceil \log(n) \rceil}\) is the most significant bit.
        We pad with leading zeros, that is, \(x_{\lceil\log(n)\rceil}, \dots, x_{\lceil\log(s)\rceil + 1}\) consist entirely of zeros.
        As an example, if $s=3$ and $n=13$ we obtain $x_3=x_4=0$ and $x_1=x_2=1$.         We say that \(x_\ell \in \{0,1\}\) is the \emph{\(\ell\)-th bit of} \(s\).

		\paragraph*{Construction of \(\bm G\).}
		For each \(i,j \in [k]\) and each \((a,b) \in S_{i,j}\), we define a so-called \emph{bitstring path}.
        This is a vertex-colored undirected path traversing the following vertices in order (see also \Cref{fig:abpath}):
        \[
           v_\text{start}, v_1,\dots,v_{\alpha^{5\lceil \log(n)\rceil}}, v_\text{end}.
        \]
        To indicate that we are encoding an entry of \(S_{i,j}\), all vertices of the path are equipped with the color \(C_{i,j}\).
		The start- and end-vertices \(v_\text{start}\) and \(v_\text{end}\) are equipped with the additional colors \texttt{start} and  \texttt{end}, respectively.
        Moreover, for \(\ell \in [\lceil \log(n) \rceil]\), each vertex \(v_{\alpha^{5\ell}}\) has additional colors
        \begin{itemize}
            \item \texttt{bit},
            \item \texttt{bit-a-one} if and only if the \(\ell\)-th bit of \(a\) is one,
            \item \texttt{bit-b-one} if and only if the \(\ell\)-th bit of \(b\) is one.
        \end{itemize}
        \begin{figure}[h]
        \begin{center}
        \includegraphics{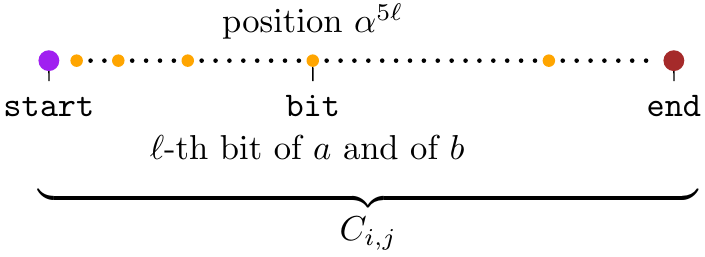}
        \end{center} 
        \caption{Visualization of a bitstring path in the proof of \Cref{thm:whard-notnice}, encoding that \((a,b) \in S_{i,j}\). The bits encoding the values of \(a\) and \(b\) are spaced with exponential distances
        such that we can identify them even when approximating with accuracy~\(\alpha\). \label{fig:abpath}}
        \end{figure}
		To obtain a single path \(G\), 
        we concatenate a copy of all bitstring paths (for all \(i,j \in [k]\) and \((a,b) \in S_{i,j}\))
        in an arbitrary order, connecting each end-vertex of a path to the start vertex of its succeeding path by an edge.
		This completes the description of \(G\).
		Notice that the number of colors is bounded by \(k^2 + 5\), and the number of vertices is bounded by \(k^2 n^2 ( \alpha^{5\lceil\log(n)\rceil} + 2) \leq k^2 \cdot n^{\bigoh(\log(\alpha))}\) (recall that \(\alpha\) is a fixed constant).

		\paragraph*{Construction of \(\bm\phi\).}

        At first, we will define a \CMSO$_1$-formula \(\phi_{\textsf{choice}}(X)\) expressing that \(X\) consists
        of the vertex sets of exactly one bitstring path of color \(C_{i,j}\) for each \(i,j \in [k]\), 
        Thus, a set \(X\) satisfying \(\phi_\textsf{choice}(X)\) chooses exactly one tuple \(s_{i,j} \in S_{i,j}\) for \(i,j \in [k]\)
        to be part of our solution. 
        This is done by including all the vertices of the bitstring path representing $s_{i,j}$ in $X$, while keeping $X$ disjoint from the vertices of all $s'\in S_{i,j}\setminus \{s_{i,j}$.

        Towards constructing the formula \(\phi_{\textsf{choice}}(X)\), we recall the well-known fact that a constant-length \MSO\(_1\)-formula can express that a vertex lies on a path between two other vertices 
        (that is, the vertex is not a separator between the other two~\cite{CourcelleE12}).
        We can thus identify whether two vertices are part of the same bitstring path by saying that no \texttt{start}- or \texttt{end}-colored vertex lies between them.
        We first express for each \(i,j \in [k]\) that there is exactly one \texttt{start}-colored and one \texttt{end}-colored 
        vertex in \(X\) that has color \(C_{i,j}\).
        We then enforce, for each \(C_{i,j}\), that $X$ contains all vertices on the path between the two selected \texttt{start} and the \texttt{end} vertices,
        and (to make sure we select \texttt{start}- and \texttt{end}-vertices of the same bitstring path)
        that the additional \(C_{i,j}\)-colored vertices are non-empty and all \texttt{bit}-vertices.
        The formula \(\phi_{\textsf{choice}}(X)\) has length \(\mathcal{O}(k^2)\).

        We will next define a formula \(\phi_{\textsf{gridtiles}}(X)\). To present the formula, we will provide both a definition of the formula in ``pseudocode'' (on the left) and the corresponding interpretation of each line (on the right), where we already assume that the set \(X\) satisfies
        \(\phi_{\textsf{choice}}(X)\).
        This is justified since our final formula will be the conjunction \(\phi(X) := \phi_{\textsf{choice}}(X) \land \phi_{\textsf{gridtiles}}(X)\).
        Afterwards, we will discuss how to convert the formula into the restricted structure required by the theorem. In our description of $\phi_{\textsf{gridtiles}}(X)$, we will use a constant-length \MSO\(_1\)-formula \(\psi_\textsf{offset}(x,Y)\) as a subformula which expresses that \(Y\) is a subset of the bitstring path of \(x\) containing exactly the vertices between the start-vertex and \(x\), including \(x\) but excluding the start-vertex (see \Cref{fig:offset}).
        \begin{figure}[h]
        \begin{center}
        \includegraphics{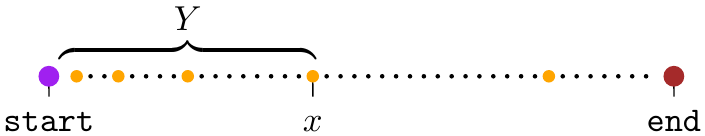}
        \end{center} 
        \caption{Visualization of the accepting inputs to \(\psi_\textsf{offset}(x,Y)\): The set \(Y\) contains all vertices along the bitstring path of \(x\) until including \(x\), excluding the start-vertex.
        \label{fig:offset}}
        \end{figure}

        \mybox{
        \noindent\begin{minipage}[t]{0.48\textwidth}
            \begin{flushleft}
                For all $i,i',j \in [k]$, \\~\\~\\

                for all $x_1, x_2$ such  that \\
                \( x_1 \in X \land x_2 \in X \land {}\)
                \(\texttt{bit}(x_1) \land \texttt{bit}(x_2) \land C_{i,j}(x_1) \land C_{i',j}(x_2) \),
                \\~\\

                for all \(Y_1, Y_2\) such that \\
                \( \psi_\textsf{offset}(x_1,Y_1) \land \psi_\textsf{offset}(x_2,Y_2) \land {}\) \\
                \( |Y_1|=|Y_2| \),
                \\~\\

                \(\texttt{bit-a-one}(x_1) \leftrightarrow \texttt{bit-a-one}(x_2)\).
                \\~\\~\\

                Additionally, for all $i,j,j' \in [k]$, \\~\\~\\

                for all $x_1, x_2$ such  that \\
                \( x_1 \in X \land x_2 \in X \land {}\)
                \(\texttt{bit}(x_1) \land \texttt{bit}(x_2) \land C_{i,j}(x_1) \land C_{i,j'}(x_2) \),
                \\~\\

                for all \(Y_1, Y_2\) such that \\
                \( \psi_\textsf{offset}(x_1,Y_1) \land \psi_\textsf{offset}(x_2,Y_2) \land {}\) \\
                \( |Y_1|= |Y_2|\),
                \\~\\

                \(\texttt{bit-b-one}(x_1) \leftrightarrow \texttt{bit-b-one}(x_2)\).
            \end{flushleft}
        \end{minipage}%
        \begin{minipage}[t]{0.50\textwidth}
            \begin{flushright}
                \it
                For all columns \(j\) and \\ cells \(i,i'\) in that column, \\~ \\

                for all ``bits'' \(x_1,x_2\) in the encoding \\ 
                of the solution elements \(s_{i,j} \in S_{i,j}\) \\ 
                and \(s_{i',j} \in S_{i',j}\), 
                as selected by \(X\), \\~\\

                for all sets \(Y_1,Y_2\) witnessing that \\
                \(x_1\) and \(x_2\) are bits with the same \\
                offset in their bitstring path, \\~\\

                the \(x_1\)-bit of \(a_1\) and the \(x_2\)-bit of \(a_2\) are equal, \\
                where \((a_1,\circ) = s_{i,j}\) and \((a_2,\circ) = s_{i',j}\).\\~\\

                Additionally, for all rows \(i\) and \\ cells \(j,j'\) in that row, \\~ \\

                for all ``bits'' \(x_1,x_2\) in the encoding \\ 
                of the solution elements \(s_{i,j} \in S_{i,j}\) \\ 
                and \(s_{i,j'} \in S_{i,j'}\), 
                as selected by \(X\), \\~\\

                for all sets \(Y_1,Y_2\) witnessing that \\
                \(x_1\) and \(x_2\) are bits with the same \\
                offset in their bitstring path, \\~\\

                the \(x_1\)-bit of \(b_1\) and the \(x_2\)-bit of \(b_2\) are equal,
                where \((\circ,b_1) = s_{i,j}\) and \((\circ,b_2) = s_{i,j'}\).
            \end{flushright}
        \end{minipage}%
        }

        \bigskip

        A set \(W\) with \(G \models \phi_{\textsf{choice}}(W)\) corresponds directly to a choice of tuples \(s_{i,j}\) for each cell \(i,j \in [k]\).
        The first half of \(\phi_{\textsf{gridtiles}}(W)\) encodes (where \(\circ\) is a placeholder for any number)
        \begin{itemize}
        \item if \(s_{i,j} = (a_1,\circ)\) and \(s_{i', j} = (a_2,\circ)\) then \(a_1=a_2\).
        \end{itemize}
        Similarly, the second half of the formula encodes 
        \begin{itemize}
        \item if \(s_{i,j} = (\circ,b_1)\) and \(s_{i, j'} = (\circ,b_2)\) then \(b_1=b_2\).
        \end{itemize}
        Therefore, the \textsc{Grid Tiling} instance is a yes-instance
        if and only if
        there exists a vertex set \(W\) with \(G \models \phi_{\textsf{choice}}(W) \land \phi_{\textsf{gridtiles}}(W)\).

        By commuting the universal quantifiers, we can rewrite \(\phi_{\textsf{choice}}(X) \land \phi_{\textsf{gridtiles}}(X)\) equivalently as follows.

        \mybox{
            \smallskip
            \noindent For all \(Y_1, Y_2\) such that \( |Y_1| = |Y_2|\), the following properties hold:

            \bigskip
            We require \(\phi_{\textsf{choice}}(X)\) to hold.

            \bigskip
            \noindent For all $i,i',j \in [k]$ and all $x_1, x_2$ such  that \( x_1 \in X \land x_2 \in X \land \texttt{bit}(x_1) \land \texttt{bit}(x_2) \land C_{i,j}(x_1) \land C_{i',j}(x_2) \land \psi_\textsf{offset}(x_1,Y_1) \land \psi_\textsf{offset}(x_2,Y_2)\), we require 
            \(\texttt{bit-a-one}(x_1) \leftrightarrow \texttt{bit-a-one}(x_2)\).

            \bigskip
            \noindent For all $i,j,j' \in [k]$ and all $x_1, x_2$ such  that \( x_1 \in X \land x_2 \in X \land \texttt{bit}(x_1) \land \texttt{bit}(x_2) \land C_{i,j}(x_1) \land C_{i,j'}(x_2) \land \psi_\textsf{offset}(x_1,Y_1) \land \psi_\textsf{offset}(x_2,Y_2)\), we require 
            \(\texttt{bit-b-one}(x_1) \leftrightarrow \texttt{bit-b-one}(x_2)\).
            \smallskip
        }

        The statement ``for all \(Y_1, Y_2\) such that \( |Y_1| = |Y_2|\), we have [...]'' is equivalent to
        ``for all \(Y_1, Y_2\) either \( |Y_1| > |Y_2|\) or \(|Y_2| > |Y_1|\) or [...]''.
        By grouping paragraph two to four of the above block into a single \mso$_1$-formula \(\psi(XY_1 Y_2)\) of length \(\bigoh(k^3)\),
        this is equivalent to the \MSOcomp$_1$-formula
        \[
            \xi(X) :=
            \forall Y_1\forall Y_2\bigl((|Y_1| > |Y_2|) \lor (|Y_2| > |Y_1|) \lor \psi(X Y_1 Y_2)\bigr).
        \]
        Hence, we have a \textsc{Grid Tiling} yes-instance if and only if there exists a set \(W\) with \(G \models \xi(W)\).
        We will define our final formula \(\phi(X) := \xi^\alpha(X)\) as the \(\alpha\)-tightening of \(\xi(X)\).

		\paragraph*{Robustness with Respect to Approximation.}

        With all components in place, we are now ready to show the correctness of the reduction.
        To be explicit, we will show that an instance \(n,k,\{S_{i,j} \subseteq [n] \times [n] \mid i,j \in [k]\}\) is a yes-instance of \textsc{Grid Tiling} if and only if the \(\alpha\)-approximate answer to the \MSOcompone-query \((\phi,0)\) on the path \(G\) is \((0,0)\) (rather than \((-\infty,-\infty)\)).
        We have already established that we have a yes-instance if and only if \(\val(G,\xi,0) = 0\) (rather than \(-\infty\)).
		By the definitions in the beginning of \Cref{sub:approx},
        it suffices to show that for all \(W \subseteq V(G)\), \(G \models \phi(W)
        \Longleftrightarrow G\models \oversatformula(W) \Longleftrightarrow G
        \models \undersatformula(W)\). We equivalently show that \(G \models
        \xi(W) \Longleftrightarrow G \models \xi^{\alpha^2}(W)\).

        Recall that when translating the ``pseudocode'' into our \MSOcomp$_1$-formula \(\xi\),
        the comparison \(|Y_1| = |Y_2| \) restricting the quantification of \(Y_1,Y_2\) was negated into \(|Y_1| > |Y_2| \lor |Y_2| > |Y_1|\).
        The \(\alpha^2\)-tightened formula \(\xi^{\alpha^2}(W)\) instead contains the harder-to-satisfy statement
        \(|Y_1|/\alpha^2 > |Y_2|\cdot \alpha^2 \lor |Y_2|/\alpha^2 > |Y_1|\cdot \alpha^2\).
        This is equivalent to 
        \(|Y_1| > |Y_2|\cdot \alpha^4 \lor |Y_2| > |Y_1|\cdot \alpha^4\); notice that the negation of this condition is precisely
        \(|Y_1| \le |Y_2|\cdot \alpha^4 \land |Y_2| \le |Y_1|\cdot \alpha^4\).
        The formula $\xi^{\alpha^2}(X)$
        can thus be equivalently stated as a conjunction of \(\phi_{\textsf{choice}}(X)\) 
        and the following modification of \(\phi_{\textsf{gridtiles}}(X)\), which we provide in the same ``pseudocode notation'' as before.
        We only highlight the changed parts.

        \mybox{
        \noindent\begin{minipage}[t]{0.48\textwidth}
            \begin{flushleft}
                [...]\\~\\

                for all \(Y_1, Y_2\) such that \\
                \( \psi_\textsf{offset}(x_1,Y_1) \land \psi_\textsf{offset}(x_2,Y_2) \land {}\) \\
                \( |Y_1| \le |Y_2|\cdot \alpha^4 \land |Y_2| \le |Y_1|\cdot \alpha^4 \),
                \\~\\

                [...]\\~\\

                for all \(Y_1, Y_2\) such that \\
                \( \psi_\textsf{offset}(x_1,Y_1) \land \psi_\textsf{offset}(x_2,Y_2) \land {}\) \\
                \( |Y_1| \le |Y_2|\cdot \alpha^4 \land |Y_2| \le |Y_1| \cdot \alpha^4 \),
                \\~\\

                [...]

            \end{flushleft}
        \end{minipage}%
        \begin{minipage}[t]{0.50\textwidth}
            \begin{flushright}
                \it~\\~\\

                for all sets \(Y_1,Y_2\) witnessing that \\
                \(x_1\) and \(x_2\) are bits with ``roughly'' the \\ same 
                offset in their bitstring path, \\~\\

                ~\\~\\

                for all sets \(Y_1,Y_2\) witnessing that \\
                \(x_1\) and \(x_2\) are bits with ``roughly'' the \\ same offset in their bitstring path, \\~\\
            \end{flushright}
        \end{minipage}%
        }

        The vertices \(x_1\) and \(x_2\) are colored with color \texttt{bit} and are thus exponentially spaced along the bitstring path with a factor \(\alpha^5\).
        There are unique sets \(Y_1\) and \(Y_2\) satisfying \( \psi_\textsf{offset}(x_1,Y_1) \land \psi_\textsf{offset}(x_2,Y_2)\).
        If \(x_1\) and \(x_2\) have a different offset, that is \(|Y_1| \neq |Y_2|\), then the exponential spacing implies that they are actually very far away,
        that is, either \(|Y_1| \ge |Y_2|\cdot \alpha^5\) or \(|Y_2| \ge |Y_1|\cdot \alpha^5\).
        Therefore, either \(|Y_1| > |Y_2|\cdot \alpha^4\) or \(|Y_2| > |Y_1|\cdot \alpha^4\), which of course again implies \(|Y_1| \neq |Y_2|\).
        By negating both statements, we see that
        on the graphs we are considering, the condition \(|Y_1| \le |Y_2|\cdot \alpha^4 \land |Y_2| \le |Y_1|\cdot \alpha^4 \)
        can be equivalently replaced with \(|Y_1| = |Y_2|\).
        Hence, for all sets \(W\), we have \(G \models \xi(W) \Longleftrightarrow G \models \xi^{\alpha^2}(W)\).
	\end{proof}

	If we are interested in the difficulty of finding exact rather than approximate answers, we can simply encode numbers using linearly many vertices rather than logarithmically many with exponential spacing. Hence, in the exact setting we can obtain a simpler reduction which even rules out polynomial-time query evaluation for fixed formulas. To clarify the connection of Theorem~\ref{thm:nphard-notnice} to our logic, we note that the ``$\neq$'' (``='') operator used in the statement below can be easily modeled by a disjunction (conjunction) of two weight comparisons.
    \begin{theorem}
		\label{thm:nphard-notnice}
        It is \NP-hard to find exact answers to a fixed \MSOcompone-query \((\phi(X),0)\) on the class of all vertex-colored trees of depth three,
        even if \(\phi(X)\) is required to have the form 
            \[
                        \psi_1(X) \land
            \forall Y_1 \forall Y_2 \forall Y_3 \forall Y_4  \bigl( (|Y_1| \neq |Y_2|) \lor (|Y_3|=|Y_4|) \lor \psi_2(XY_3Y_4) \bigr),
		\]
        where \(\psi_1\) and \(\psi_2\) are \mso$_1$-formulas.
	\end{theorem}
	\begin{proof}
		We reduce from the \textsc{non-parameterized Grid Tiling} problem which can be seen to be \NP-hard using the same reduction used to establish \W[1]-hardness of the parameterized version~\cite[Theorem~14.28]{cygan2015parameterized}.
		
		Given an instance of \textsc{non-parameterized Grid Tiling} we describe how to construct a tree \(G\) and a formula \(\phi\) and then verify that these instances meet the requirements to show the theorem statement.
		
		\paragraph*{Construction of \(\bm G\).}
        We start by adding a root to the tree.
        For every \(i,j \in [k]\), we add a child \(v_{i,j}\) below the root with color \texttt{cell}.
        To check whether two cells are in the same row or column, we attach \(i\) many leaves with color \texttt{row} and \(j\) many leaves with color \texttt{column} to \(v_{i,j}\).
        Additionally, for each \((a,b) \in S_{i,j}\), we attach a vertex of color \texttt{pair} to \(v_{i,j}\) and give it \(a\) many children of color \texttt{first} and \(b\) many children of color \texttt{second}.
		See \Cref{fig:hardnesstree} for a visualization.
		Notice that \(G\) has depth \(3\) and polynomial size in the grid tiling instance.
		\begin{figure}[h]
			\begin{center}
				\includegraphics{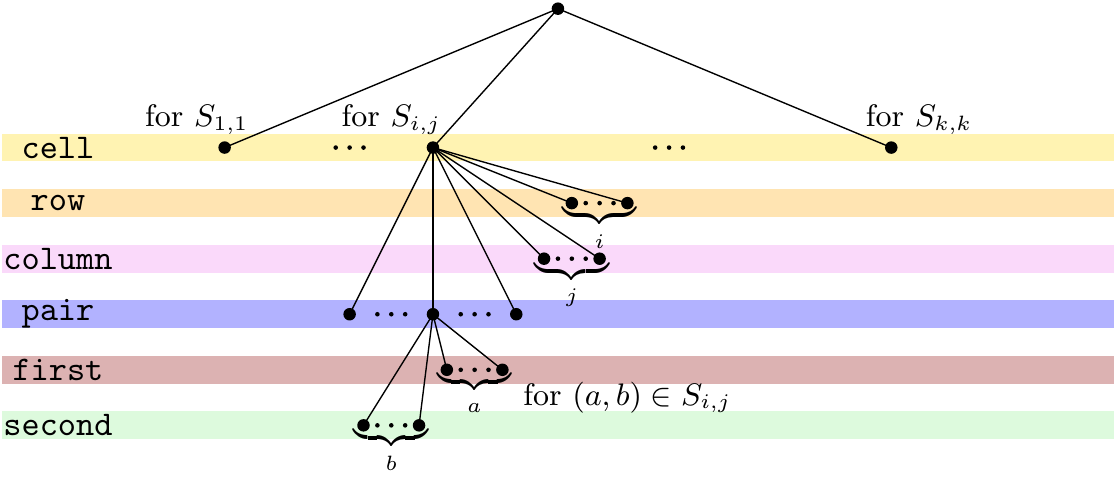}
			\end{center} 
			\caption{Visualization of the tree \(G\) in the proof of \Cref{thm:nphard-notnice}. The colored layers indicate the colors of the vertices
            and the curly underbraces their count.}
             \label{fig:hardnesstree}
		\end{figure}

		\paragraph*{Construction of \(\bm \phi\).}

        We first define a \CMSO$_1$-formula \(\phi_{\textsf{choice}}(X)\) expressing 
        that all vertices in \(X\) have color \texttt{pair} and
        each \texttt{cell}-colored vertex has exactly one neighbor in \(X\).
        Thus, as in the proof of \Cref{thm:whard-notnice}, 
        a set \(X\) satisfying \(\phi_\textsf{choice}(X)\) chooses exactly one tuple \(s_{i,j} \in S_{i,j}\) for \(i,j \in [k]\)
        to be part of our solution. 
        Then we define \(\phi(X)\) to be the conjunction of \(\phi_{\textsf{choice}}(X)\) and the following formula
        \(\phi_{\textsf{gridtiles}}(X)\) provided in a similar ``pseudocode'' style as in the previous proof:

        \mybox{%
            \begin{flushleft}%
                For all $x_1, x_2$ and all \(Y_1, Y_2, Y_3, Y_4\) such that
                    \begin{itemize}
                        \item \(x_1,x_2\) are colored with \texttt{cell},
                        \item \(Y_1\) are precisely the \texttt{row}-children below \(x_1\), 
                        \item \(Y_2\) are precisely the \texttt{row}-children below \(x_2\),
                        \item \(Y_3\) are precisely the \texttt{first}-children below the \texttt{pair}-vertex in \(X\) below \(x_1\), and
                        \item \(Y_3\) are precisely the \texttt{first}-children below the \texttt{pair}-vertex in \(X\) below \(x_2\),
                    \end{itemize}
                we require \(|Y_1|=|Y_2| \rightarrow |Y_3|=|Y_4|\). \\~\\
                    \emph{``That is, if \(x_1\) and \(x_2\) correspond to cells of the same row, \\ then the first coordinates of the selected tuples agree.''} \\~\\
                Moreover, for all $x_1, x_2$ and all \(Y_1, Y_2, Y_3, Y_4\) such that
                    \begin{itemize}
                        \item \(x_1,x_2\) are colored with \texttt{cell},
                        \item \(Y_1\) are precisely the \texttt{column}-children below \(x_1\), 
                        \item \(Y_2\) are precisely the \texttt{column}-children below \(x_2\),
                        \item \(Y_3\) are precisely the \texttt{second}-children below the \texttt{pair}-vertex in \(X\) below \(x_1\), and
                        \item \(Y_3\) are precisely the \texttt{second}-children below the \texttt{pair}-vertex in \(X\) below \(x_2\),
                    \end{itemize}
                we require \(|Y_1|=|Y_2| \rightarrow |Y_3|=|Y_4|\). \\~\\
                    \emph{``That is, if \(x_1\) and \(x_2\) correspond to cells of the same column, \\ then the second coordinates of the selected tuples agree.''}
            \end{flushleft}
        }

        It is clear from our construction that the \textsc{Grid Tiling} instance is a yes-instance
        if and only if
        there exists a vertex set \(W\) with \(G \models \phi_{\textsf{choice}}(W) \land \phi_{\textsf{gridtiles}}(W)\).
        It remains to bring this formula into the desired form.
        Observe that 
        \[
            \textnormal{for all \(x_1,x_2,Y_1,\dots,Y_4\) such that \(\rho_1 \lor \rho_2\), we require \(|Y_1|=|Y_2| \rightarrow |Y_3|=|Y_4|\)}
        \]
        is equivalent to 
        \[
             \textnormal{for all \(x_1,x_2,Y_1,\dots,Y_4\), either \(|Y_1|=|Y_2| \rightarrow |Y_3|=|Y_4|\) or \(\neg\rho_1 \land \neg\rho_2\)}.
        \]
        Hence, if we use $\rho_1$ and $\rho_2$ to represent the \mso$_1$ formulas capturing the former and latter block of five bullet points in the pseudocode description of \(\phi_{\textsf{gridtiles}}(X)\),       we can rewrite \(\phi_{\textsf{choice}}(X) \land \phi_{\textsf{gridtiles}}(X)\) as 
        \[
            \phi_{\textsf{choice}}(X) \land
            \forall Y_1 \forall Y_2 \forall Y_3 \forall Y_4 \bigl(|Y_1|=|Y_2| \rightarrow |Y_3|=|Y_4|\bigr) \lor 
            \bigl(\forall x_1 \forall x_2 \neg\rho_1 \land \neg\rho_2\bigr).
        \]
        At this point, it remains to rephrase the above as 
        \[
            \phi(X) :=
            \psi_1(X) \land
            \forall Y_1 \forall Y_2 \forall Y_3 \forall Y_4  \bigl( (|Y_1| \neq |Y_2|) \lor (|Y_3|=|Y_4|) \lor \psi_2(XY_3Y_4) \bigr). \qedhere
        \]
	\end{proof}
	
	\bibliographystyle{plainurl}
	\bibliography{references}

\end{document}